\documentclass[11pt]{article}
\usepackage[margin=1in]{geometry}
\usepackage[round]{natbib}
\usepackage{outlines}
\usepackage[splitrule]{footmisc}
\usepackage[T1]{fontenc}
\usepackage{authblk}
\usepackage[utf8]{inputenc}
\usepackage{hyperref}
\usepackage{amsmath}
\usepackage{amssymb}
\usepackage{amsthm}
\usepackage{amsfonts}
\usepackage{algorithm}
\usepackage{algorithmic}
\usepackage{xcolor}
\usepackage{color}
\usepackage{bbding}
\usepackage{tcolorbox}
\usepackage{tikz}
\usepackage{natbib}
\usepackage{lipsum}
\usepackage{cuted}
\usepackage{multirow}
\usepackage{colortbl}
\usepackage{makecell}
\usepackage{booktabs}
\newtheorem{theorem}{Theorem}[section]

\newtheorem{lemma}{Lemma}

\newtheorem{definition}{Definition}
\newtheorem{assumption}[theorem]{Assumption}

\usepackage{graphicx, comment}
\usepackage{float}
\usepackage{mathrsfs}
\usepackage{subfigure}
\def\defeq{:=} 
\newcommand{\MLD}{{MLD}}
\newcommand{\UMLD}{{MULD}}
\newcommand{\MULA}{{MULA}}
\newcommand{\ULD}{{ULD}}
\newcommand{\ULMC}{{ULMC}}
\newcommand{\UMLA}{{MULA}}
\newcommand{\UNLA}{{N-ULA}}
\newcommand{\UNLD}{{N-ULD}}
\newcommand{\MLA}{{MLA}}
\newcommand{\EMUNLA}{{EM-N-ULA}}
\newcommand{\NLA}{{N-LA}}
\usepackage{hyperref}
\hypersetup{
    colorlinks,
    linkcolor={red!80!black},
    citecolor={blue!50!black},
}
\usepackage[shortlabels]{enumitem}
\usepackage[nameinlink,capitalise]{cleveref}

\makeatletter
\newcommand*{\rom}[1]{\expandafter\@slowromancap\romannumeral #1@}
\makeatother
\makeatletter
\def\@maketitle{%
  \newpage
  \begin{center}%
  \let \footnote \thanks
    {\LARGE \bf \@title \par}%
  \end{center}%
  \par
  \vskip 1em}
\makeatother
\title{Mean-field underdamped Langevin dynamics and its spacetime discretization}

\begin{document}
\maketitle
\begin{center}
{\large
\begin{tabular}{ccc}
    \makecell{Qiang Fu\(^\star\) \\{\normalsize\texttt{fuqiang7@mail2.sysu.edu.cn}}} & &
    \makecell{Ashia Wilson\(^\dagger\)\\{\normalsize\texttt{ashia07@mit.edu}}}
\end{tabular}
\vskip 1em
}
 \begin{tabular}{c}
   \(^\star\)School of Mathematics, Sun Yat-sen University \\
\(^\dagger\)Department of Electrical Engineering and Computer Science, MIT

\end{tabular}

\end{center}
\maketitle
\begin{abstract}
We propose a new method called the N-particle underdamped Langevin algorithm for optimizing a special class of non-linear functionals defined over the space of probability measures. Examples of problems with this formulation include training mean-field neural networks, maximum mean discrepancy minimization and kernel Stein discrepancy minimization. Our algorithm is based on a novel spacetime discretization of the mean-field underdamped Langevin dynamics, for which we provide a new, fast mixing guarantee. In addition, we demonstrate that our algorithm converges globally in total variation distance, bridging the theoretical gap between the dynamics and its practical implementation.
\end{abstract}
\section{Introduction}
\label{intro}
The mean-field Langevin dynamics (MLD) has recently received renewed interest due to its connection to gradient-based techniques used in supervised learning problems such as training neural networks in a limiting regime~\citep{mei2018mean}. Theoretical characterizations of the convergence properties of MLD has been the particular focus of several recent works~\citep{hu2019mean,chizat2022mean,nitanda2022convex,chen2022uniform,claisse2023mean}. 
More generally, MLD{} can be used to solve problems that can be posed as an entropy regularized mean-field optimization  (EMO) problem. Other examples of such problems include density estimation via maximum mean discrepancy (MMD) minimization~\citep{gretton2006kernel,arbel2019maximum,chizat2022mean,suzuki2023convergence} and sampling via kernel Stein discrepancy (KSD) minimization~\citep{liu2016kernelized,chwialkowski2016kernel,suzuki2023convergence}.
A more detailed synthesis of recent theoretical developments for MLD can be summarized as follows. \citet{hu2019mean} show that MLD finds EMO solutions asymptotically when problems can be expressed as optimizing a convex functional. 
If in addition, the EMO satisfies a uniform logarithmic Sobolev inequality, several studies have established that this convergence occurs exponentially quickly~\citep{chizat2022mean, nitanda2022convex, chen2022uniform}.

However, implementing MLD{} is not a straightforward task; to arrive at a practical algorithm requires both spatial and temporal discretizations of the dynamics.
\citet{nitanda2022convex} study a time-discretization of MLD{} by
extending an interpolation argument introduced by \citet{vempala2019rapid} to a non-linear Fokker-Planck equation. They establish a non-asymptotic rate of convergence for the discrete-time process. \citet{chen2022uniform} study a space-discretization consisting of a finite-particle approximation to the density of MLD{} (referred to as a finite-particle system) and show the finite-particle system finds the solution to the EMO problem exponentially fast, with a bias related to the number of particles. More practically, \citet{suzuki2023convergence}  analyze a spacetime discretization of the MLD and establish the non-asymptotic convergence of the resulting algorithm to a biased limit related to both the number of particles used and stepsize. 
Their analysis 
applies to several important learning problems and improves the results of the standard gradient Langevin dynamics.
A natural candidate method for finding solutions to EMO problems faster is the
mean-field {\em underdamped} Langevin dynamics (MULD). MULD resemble several techniques for adding momentum to gradient descent in optimization, many of which are known to result in provably faster convergence in a variety of settings~\citep{nesterov1983method,wilson2016lyapunov,laborde2020lyapunov,hinder2020near,fu2023accelerated}. Moreover, training neural networks using momentum-based gradient descent is considered effective in several applications~\citep{sutskever2013importance,kingma2014adam,ruder2016overview}. 
\citet{kazeykina2020ergodicity} and \citet{chen2023uniform} confirm that a naive spacetime discretization of
MULD has impressive empirical performance when compared to a naive discretization of the  
MLD on applications such as training mean-field neural networks.
 \citet{chen2023uniform} introduce a space-discretization  of MULD consisting of a finite particle approximation to the density and show it finds the EMO solution exponentially fast, albeit with several additional assumptions that are easy to verify for the problem of training mean-field neural networks. 
In addition,~\citet{chen2023uniform} implement an Euler-Maruyama discretization of the finite-particle system 
and show that it performs empirically faster when compared with the spacetime discretization of the mean-field Langevin dynamics in training a toy neural network model.
However, spacetime discretizations of MULD  
are not yet theoretically well understood. Furthermore, the rate obtained by \citet{chen2023uniform} for the dynamics does not resemble an ``accelerated rate'' when compared with recent results for MLD.

\subsubsection*{A summary of our work} 
A remaining question is whether we can theoretically characterize the behavior of an implementable algorithm based on discretizing the mean-field underdamped dynamics. If there is a limiting bias, how does it scale with the number of particles and other problem parameters? Ideally, this characterization would give a sharper rate of convergence than~\citet{suzuki2023convergence}'s spacetime discretization of the mean-field Langevin dynamics, suggesting there might be an advantage to adding momentum in the mean-field setting (at least in the worst case). 
In this paper, we introduce a fast implementable algorithm for solving EMO problems based on the mean-field underdamped Langevin dynamics. We prove that our proposed algorithm converges to a small limiting bias under a set of assumptions that subsumes many problems of interest. In particular, our contributions are summarized as follows. 

\begin{enumerate}
    \item 


 We sharpen the convergence bound for MULD and its space-discretization established by \citet{chen2023uniform} under the same set of assumptions utilized by \citet{chen2023uniform}  (Theorems \ref{convergence rate of the UMFLD} and~\ref{convergence of particle dynamics} and Table~\ref{summary}). 
    \item 
    We show the global convergence of our proposed algorithm in total variation (TV) distance (Theorem \ref{iter complexity for particle system}). Importantly, our results improve on \citet{suzuki2023convergence}'s analysis of the spacetime discretization of the MLD. While we require additional assumptions \ref{second moment bound}-\ref{initialization and invariant for N-particle}, our 
    results hold in several real-world applications including training neural networks, density estimation via MMD minimization and sampling via KSD minimization. 
\end{enumerate}
\paragraph{Organization} The remainder of this work is organized as follows. Section \ref{preliminaries} presents the formal definitions and assumptions as well as important related work. Section \ref{s-t disc} proposes our main methods and theoretical results. Section \ref{sec:ex} discusses the application of our methods to some classical problems. Section \ref{NE} describes our numerical experiments verifying the effectiveness of our proposed methods.
\section{Preliminaries}
\label{preliminaries}
We begin by introducing some general notation that will be used throughout this work.
\subsection{Notation}
The Euclidean and operator norms are denoted by $\|\cdot\|$  and $\|\cdot\|_{\textsf{op}}$. The space of probability measures on $\mathbb{R}^d$ with finite second moment is denoted by $\mathcal{P}_2(\mathbb{R}^d)$. Throughout, let $\rho$ and $\mu$ denote general distributions in $\mathcal{P}_2(\mathbb{R}^d)$ and $\mathcal{P}_2(\mathbb{R}^{2d})$ respectively. 
The TV distance between $\rho$ and $\pi\in\mathcal{P}_2(\mathbb{R}^d)$ is denoted by $\|\rho-\pi\|_{\textsf{TV}}\defeq\sup|\rho(A)-\pi(A)|$  
where the sup is over all Borel measurable sets $A\subset\mathbb{R}^d$. 
 The $p$-Wasserstein distance and Kullback-Leibler divergence between $\rho$ and $\pi$ is denoted by $W_p(\rho,\pi)\defeq\inf_{\Pi}\mathbb{E}_{\Pi}[\|x-y\|^p]^{1/p}$ where the infimum is over joint distributions $\Pi$ of $(x,y)$ with the marginals $x\sim\rho,y\sim\pi$
 and $\textsf{KL}(\rho\|\pi)\defeq\int \rho\log\frac{\rho}{\pi}$.  The relative Fisher information is denoted by $\textsf{FI}(\rho\|\pi)\defeq\mathbb{E}_{\rho}\|\nabla\log\frac{\rho}{\pi}\|^2$, and more generally we use the notation
 $\textsf{FI}_S(\rho\|\pi)\defeq\mathbb{E}_{\rho}\|S^{1/2}\nabla\log\frac{\rho}{\pi}\|^2$ for a positive definite symmetric matrix $S$. $\text{Ent}(\rho)\defeq\int\rho\log\rho$ denotes the negative entropy of $\rho$.
The functional and intrinsic derivatives of $F$ are denoted by $\frac{\delta F}{\delta\rho}:\mathcal{P}_2(\mathbb{R}^d)\times\mathbb{R}^d\rightarrow\mathbb{R}$ 
and 
${D}_{\rho}F\defeq\nabla\frac{\delta F}{\delta\rho}:\mathcal{P}_2(\mathbb{R}^d)\times\mathbb{R}^d\rightarrow\mathbb{R}^d$, respectively. A $d$-dimensional Brownian motion is denoted by $\mathrm{B}_t$. 
%
 %
%
We use notation
$a\lesssim b$, $a_n=\Theta(b_n)$ and $a_n=\widetilde{\Theta}(b_n)$ to denote that there exist $c,C>0$ such that $a\leq Cb$, $cb_n\leq a_n\leq Cb_n$ for $n\geq N'$
and $a_n=\Theta(b_n)$ up to logarithmic factors, respectively.
\subsection{Background}
We consider the following 
problem described by minimizing the entropy regularized mean-field objective (EMO), \looseness=-1
\begin{equation}
\label{mean-field opt}
    \min_{\rho\in\mathcal{P}_2(\mathbb{R}^d)}F(\rho)+\lambda\text{Ent}(\rho),    
\end{equation}
where $F: \mathcal{P}_2(\mathbb{R}^{d}) \rightarrow \mathbb{R}$ is a potentially non-linear functional and $\lambda>0$ is a regularization constant. 
Without loss of generality, we will take $\lambda = 1$ throughout. 
 \citet{hu2019mean} study the gradient flow dynamics 
 of the EMO in 2-Wasserstein metric called the {\em mean-field Langevin dynamics} (\hyperlink{MLD}{ {MLD}}): 
\hypertarget{MLD}{\begin{equation*}
    \label{MFLD}
    \mathrm{d}x_t=-
     D_{\rho} F(\rho_t,x_t)\mathrm{d}t
    +\sqrt{2} \mathrm{dB}_t,
    \tag{{MLD}}
\end{equation*}}
where $\rho_t\defeq\text{Law}(x_t)\in\mathcal{P}_2(\mathbb{R}^d)$.
 Under mild conditions, the \MLD{} finds the solution to the EMO, given by $\rho_*(x)\propto\exp\left(-\frac{\delta F}{\delta\rho}(\rho_*,x)\right)$~\citep{hu2019mean}.


This paper introduces a new sharp mixing-time bound for the 
 {\em mean-field underdamped Langevin dynamics} (\hyperlink{UMLD}{ {MULD}}):
\hypertarget{UMLD}{
\begin{equation}
\label{UMFLD}
    \begin{aligned}
   \mathrm{d}x_t&=v_t\mathrm{d}t, \\
   \mathrm{d}v_t&=\hspace{-1pt}-\gamma v_t\mathrm{d}t\hspace{-1pt}-\hspace{-1pt} D_{\rho}F(\mu_t^X,x_t)\mathrm{d}t\hspace{-1pt}+\hspace{-2pt}\sqrt{2\gamma}\mathrm{dB}_t.
    \end{aligned} 
    \tag{{MULD}}
\end{equation}}
Here, $\mu_t\defeq\text{Law}(x_t,v_t)\in\mathcal{P}_2(\mathbb{R}^{2d})$, $\gamma >0$ is the damping coefficient, and $\mu_t^X\defeq\text{Law}(x_t)=\int\mu_t(x,v)\mathrm{d}v$ is the $X$-marginal of $\mu_t$.
The limiting distribution of~\UMLD{} is the solution to the augmented EMO problem,  
\begin{equation}
    \label{augmented mean-field opt}
    \min_{\mu \in\mathcal{P}_2(\mathbb{R}^{2d})}F(\mu^X)+\text{Ent}(\mu)+\int\frac{1}{2}\|v\|^2\mu(\mathrm{d}x\mathrm{d}v),
\end{equation}
where a momentum term is added to the EMO.
The minimizer of the augmented EMO is given by $\mu_*(x,v)\propto\exp\left(-\frac{\delta F}{\delta\rho}(\mu^X_*,x)-\frac{1}{2}\|v\|^2\right)$. We provide details of the derivation of the limiting distributions of \MLD{} and \UMLD{} in Appendices \ref{MLD background} and \ref{MULD background} respectively.
To obtain the solution of the EMO problem, the minimizer $\mu_*(x,v)$ 
can be $X$-marginalized.
%
%
This work also sharpens the analysis of the space-discretization of MULD introduced by~\citet{chen2023uniform}, which we refer to as the {\em  $N$-particle underdamped Langevin dynamics} (N-ULD) for $i=1,...,N$:
 \begin{subequations}
     \begin{align*}
     \label{Par-UMFLD}
     \tag{{N-ULD}}
        \mathrm{d}x_t^i&=v_t^i\mathrm{d}t,\\
        \mathrm{d}v_t^i&=-\gamma v_t^i\mathrm{d}t\vspace{-10cm}-D_{\rho}F(\mu_{\textbf{x}_t},x^i_t)\mathrm{d}t+\sqrt{2\gamma}\mathrm{dB}^i_t,
     \end{align*} \notag
 \end{subequations}
where $\mu_{\textbf{x}_t}\defeq\frac{1}{N}\sum_{i=1}^N\delta_{x_t^i}$, $\mu^i_t\defeq\text{Law}(x_t^i,v_t^i)$ and $(\mathrm{B}_t^i)_{i=1}^N$ are $d$-dimensional Brownian motions.

To motivate our algorithm as a time-discretization of~\UNLD, we review discretizations of the {\em underdamped Langevin dynamics} (\ULD), which is a special case of~\UMLD{} where $F(\mu)=\int V(x)\mu(\mathrm{d}x)$ is a linear functional of $\mu$: 
\begin{equation}
    \label{ULD}
    \begin{aligned}
   \mathrm{d}x_t&=v_t\mathrm{d}t \\
   \mathrm{d}v_t&=\hspace{-1pt}-\gamma v_t\mathrm{d}t\hspace{-1pt}-\hspace{-1pt}\nabla V(x_t)\mathrm{d}t\hspace{-1pt}+\hspace{-2pt}\sqrt{2\gamma}\mathrm{dB}_t.
    \end{aligned} 
    \tag{{ULD}}
\end{equation}
The~\ULD{} was first studied in \citet{kolmogoroff1934zufallige} and \citet{hormander1967hypoelliptic}. Under functional inequalities such as Poincar{\'e}'s inequality on the target distribution $\rho_*\propto\exp(-V)$
, the convergence guarantee of the~\ULD{} was studied by Villani using a hypocoercivity approach \cite{villani2001limites, villani2009hypocoercivity}, but without capturing the acceleration phenomenon when compared to the overdamped Langevin dynamics. \citet{cao2023explicit} are the first to show~\ULD{} converges in $\chi^2$-divergence at an accelerated rate when $V$ is convex and the target distribution $\rho_*$ satisfies LSI defined in \eqref{LSI} with $\mathscr{C}_{\textsf{LSI}}>0$.
They prove that when $\mathscr{C}_{\textsf{LSI}}\ll 1$, the decaying rate of \ULD{} is $O(\sqrt{\mathscr{C}_{\textsf{LSI}}})$ whereas the decaying rate of the overdamped Langevin dynamics is $O(\mathscr{C}_{\textsf{LSI}})$.

A discretization of \ULD{} is referred to as an {\em underdamped Langevin Monte Carlo} (\ULMC{}) algorithm. 
There are various discretization schemes proposed for implementing \ULD{}. The \emph{Euler-Maruyama} (EM) discretization of \ULD{}~\citep{kloeden1995numerical,platen2010numerical}, \looseness=-1 
\begin{equation*}
    \begin{aligned}
        x_{k+1}&=x_k+hv_k,\\
        v_{k+1}&=(1-\gamma h)v_k-h\nabla V(x_k)+\sqrt{2\gamma h}\xi_k,
    \end{aligned}
    \tag{{EM-ULMC}}
\end{equation*}
for stepsize $h$, $\xi_k\sim\mathcal{N}(0,I_d)$ and $t\in[kh,(k+1)h]$, has been well-studied and it incurs the largest discretization
error in several metrics including KL divergence and Wasserstein distance. 
Recently, however, several works have studied the ULMC obtained from a more precise discretization scheme called the the \emph{exponential integrator} (EI)~\citep{cheng2018underdamped}:
\begin{subequations}
    \begin{align*}
    \label{discretization of ULD}
    \tag{{EI-ULMC}}
       \mathrm{d}x_t&=v_t\mathrm{d}t,\\
       \mathrm{d}v_t&=-\gamma v_t\mathrm{d}t-\nabla V(x_{kh})\mathrm{d}t+\sqrt{2\gamma}\mathrm{dB}_t,
    \end{align*}
\end{subequations}
for $t\in[kh,(k+1)h]$. Unlike the EM integrator, EI only fixes the drift term in each small interval, creating a group of linear stochastic differential equations (SDE) that can be exactly integrated. \citet{leimkuhler2023contraction} show that the EI incurs weaker stepsize restriction when compared with EM scheme. Other works have derived its convergence
in Wasserstein distance \citep{cheng2018underdamped}, KL divergence \citep{ma2021there} and R{\'e}nyi divergence \citep{zhang2023improved}. 
Other discretization schemes are proposed in \citet{shen2019randomized,li2019stochastic,he2020ergodicity,foster2021shifted,monmarche2021high,foster2022high,johnston2023kinetic}, whose convergence guarantee are obtained in Wasserstein distance without achieving better dependence on terms such as the smoothness and LSI constants.
In this work, we show that EI can be applied to discretize both \ref{UMFLD} and \ref{Par-UMFLD} to achieve fast convergence.
\subsection{Definitions and assumptions}
For each method considered, we study their behavior in settings where the minimizing distribution 
satisfies a Log-Sobolev inequality.
\begin{definition}[LSI]
    A measure $\pi\in\mathcal{P}_2(\mathbb{R}^d)$ satisfies Log-Sobolev Inequality ({LSI}) with parameter $\mathscr{C}_{\textsf{LSI}}>0$, if for any $\rho\in\mathcal{P}_2(\mathbb{R}^d)$  
    \looseness=-1
        \begin{equation}
    \label{LSI}
        \textsf{\em KL}(\rho\|\pi)\leq\frac{1}{2\mathscr{C}_{\textsf{\em LSI}}}\textsf{\em FI}(\rho\|\pi).
    \end{equation}
\end{definition}
We also work with the following distribution $\hat{\mu}\in\mathcal{P}_2(\mathbb{R}^{2d})$ that appears in the Fokker-Planck equation \eqref{f-p MULD} of \ref{UMFLD} (see Appendix~\ref{MULD background}). Note that the limiting distribution $\mu_*\in\mathcal{P}_2(\mathbb{R}^{2d})$ of \UMLD{} satisfies $\mu_*=\hat{\mu}_*$.
\begin{definition}
    Throughout, we define the distribution $\hat{\mu}$ associated with the $X$-marginal of distribution $\mu$
    and a functional $F$ to be 
    \begin{align}\label{eq:gibbs}
    \hat{\mu}(x,v)\propto\exp\left(-\frac{\delta F}{\delta\rho}(\mu^X,x)-\frac{1}{2}\|v\|^2\right).
\end{align} 
\end{definition}
We also introduce the same three assumptions on $F$ as  \citet{chen2023uniform} for establishing the non-asymptotic convergence of the~\ref{UMFLD} and \ref{Par-UMFLD}.
\begin{assumption}[Convexity]
\label{convexity}
    $F$ is convex in the linear sense, which means for any $\rho_1,\rho_2\in\mathcal{P}_2(\mathbb{R}^d)$ and $t\in[0,1]$ the functional satisfies
    \begin{equation}
        F(t\rho_1+(1-t)\rho_2)\leq tF(\rho_1)+(1-t)F(\rho_2).
    \end{equation}
\end{assumption}
\begin{assumption}[$\mathscr{L}$-smoothness]
\label{smoothness}
    $F$ is smooth, which means the intrinsic derivative 
    exists and for any $\rho_1,\rho_2\in\mathcal{P}_2(\mathbb{R}^d)$, $x_1,x_2\in\mathbb{R}^d$ and some $1\leq\mathscr{L}<\infty$ satisfies
    \begin{equation}
    \begin{aligned}
        \label{L-smooth}
        \|D_{\rho}F(\rho_1,x_1)-D_{\rho}F(\rho_2,x_2)\|\leq\mathscr{L}(W_1(\rho_1,\rho_2) 
        +\|x_1-x_2\|).
        \end{aligned}
    \end{equation}
\end{assumption}
\begin{assumption}[LSI]
\label{LSI of PGD}
The distribution \eqref{eq:gibbs}
satisfies LSI with constant $0<\mathscr{C}_{\textsf{\em LSI}}\leq 1$ for any $\mu\in\mathcal{P}_2(\mathbb{R}^d)$.
\end{assumption}
The $X$-marginal of distribution \eqref{eq:gibbs}, which is related to the optimization gap, was first utilized by~\citet{nitanda2022convex} to establish convergence of \MLD.
    Note that if $\hat{\mu}^X(x)\propto\exp(-\frac{\delta F}{\delta\rho}(\mu^X,x))$ satisfies LSI for any $\mu\in\mathcal{P}_2(\mathbb{R}^{2d})$ with constant $\tau>0$, then Assumption \ref{LSI of PGD} is satisfied with the choice $\mathscr{C}_{\textsf{LSI}}=\min\{{1}/{2},\tau\}$. We refer our readers to \citet{chen2022uniform,chen2023uniform,suzuki2023convergence} for the verification of Assumptions~\ref{convexity} and \ref{LSI of PGD} in a variety of settings.
    \citet{suzuki2023convergence} consider a weaker smoothness assumption than Assumption \ref{smoothness} where they use $W_2$ distance in place of $W_1$ distance.
    They verify smoothness in $W_2$ distance
    for three examples including training mean-field neural networks, MMD minimization and KSD minimization, whereas \citet{chen2022uniform} verify smoothness in $W_1$ distance only for the example of training mean-field neural networks. In this paper, we verify $\mathscr{L}$-smoothness in $W_1$ distance (Assumption \ref{smoothness}) for the other two examples (see Section~\ref{verification of smoothness}).
Beyond Assumptions~\ref{convexity}-\ref{LSI of PGD}, we introduce four additional assumptions that are sufficient for our spacetime discretization analysis.
\begin{assumption}[Bounded Gradient]
\label{bounded grad}
For any $\rho\in\mathcal{P}_2(\mathbb{R}^d)$, the intrinsic derivative of $F$ satisfies (where $\mathscr{L}>0$)
    \begin{equation}
        \|D_{\rho}F(\rho,x)\|\leq \mathscr{L}(1+\|x\|).
    \end{equation}
\end{assumption}
    Notably, \citet{suzuki2023convergence} assume that $F$ can be decomposed as $F(\rho)=U(\rho)+\mathbb{E}_{x\sim\rho}[r(x)]$ where $\|D_{\rho}U(\rho,x)\|\leq R$ for any $\rho\in\mathcal{P}(\mathbb{R}^d)$, $x\in\mathbb{R}^d$, and where $r(x)$ is a differentiable function satisfying $\|\nabla r(x) - \nabla r(y) \| \leq \lambda_2\| x-y \|$  with $\nabla r(0)=0$ in order to establish the convergence of their spacetime discretization of \ref{MFLD}. Thus, their assumption that $\|D_{\rho}F(\rho,x)\|\leq\|D_{\rho}U(\rho,x)\|+\|\nabla r(x)\|\leq R+\lambda_2\|x\|$ implies Assumption~\ref{bounded grad} holds with the choice $\mathscr{L}\geq\max\{R,\lambda_2\}$. The next three assumptions are needed for bounding the second moment of the iterates $(x_t,v_t)_{t\geq 0}$ and $(x_t^i,v_t^i)_{t\geq 0}$ along~\UMLD{} and \UNLD{}, which is crucial for the establishment of our discrete-time convergence.
\begin{assumption}
\label{second moment bound}
    For all $\mu\in\mathcal{P}_2(\mathbb{R}^{2d})$, the distribution~\eqref{eq:gibbs} given $F$ 
    satisfies
 $\mathbb{E}_{\hat{\mu}}\|\cdot\|^2\lesssim d.$
\end{assumption}
\begin{assumption}
\label{initialization and invariant}
    Given the initial distribution $\mu_0\in\mathcal{P}_2(\mathbb{R}^{2d})$ of the discrete-time process of \UMLD{}, functional $F$ and satisfies
       $F(\mu^X_0)\lesssim\mathscr{L}d.$
\end{assumption}
\begin{assumption}
\label{initialization and invariant for N-particle}
    Given the initial distribution $\mu_0^N\in\mathcal{P}_2(\mathbb{R}^{2Nd})$ of the discrete-spacetime process of \UMLD{}, functional $F$ satisfies $\mathbb{E}_{\textbf{x}_0\sim(\mu_0^X)^N}F(\mu_{\textbf{x}_0})\lesssim\mathscr{L}d,$ where $\mu_0^N$ is the N-tensor product of $\mu_0$ and $\mu_{\textbf{x}_0}=\frac{1}{N}\sum_{i=1}^N\delta_{x_0^i}$ with $x_0^i\sim\mu^X_0$.
\end{assumption}
While Assumptions \ref{second moment bound}-\ref{initialization and invariant for N-particle}   are sufficient, they may not be necessary for the iterates to be bounded. Nevertheless, we argue these assumptions are not too restrictive by verifying them for three examples introduced above including training mean-field neural networks, MMD minimization and KSD minimization in Section~\ref{sec:ex}.
\subsection{Related work}
\label{related works}
Techniques for establishing the continuous-time convergence of the mean-field underdamped systems and their space-discretization (N-particle systems) are centered around \textit{coupling} and \textit{hypocoercivity}. The latter one is also known as functional approaches~\citep{villani2009hypocoercivity}. The coupling approach generally constructs a joint probability of the mean-field and N-particle systems to make the analytic comparison between them. Based on coupling approaches, \citet{guillin2022convergence,bolley2010trend,bou2023convergence} show convergence of the underdamped dynamics with mean-field interaction and its space-discretization. \citet{duong2018vlasov,kazeykina2020ergodicity} study the ergodicity of the~\ref{UMFLD} without a quantitative rate. Under the setting of small mean-field dependence, \citet{kazeykina2020ergodicity} show exponential contraction using coupling techniques in \citet{eberle2019couplings,eberle2019quantitative}. The functional approach (hypocoercivity) generally constructs appropriate Lyapunov functionals and studies how their values change along the dynamics. Based on hypocoercivity, \citet{monmarche2017long,guillin2021kinetic,guillin2021uniform,bayraktar2022exponential} establish the exponential convergence of the mean-field underdamped systems and its propagation of chaos by constructing a suitable Lyapunov functional. Nevertheless, most of the works above only consider specific settings of~\UMLD{} such as singular interactions and two-body interactions, which restricts the application to real-world problems. Setting $\gamma = 1$, \citet{chen2023uniform} establish the exponential convergence of~\UMLD{} and \ref{Par-UMFLD} using the hypocoercivity technique in \citet{villani2009hypocoercivity}. Under Assumptions~\ref{convexity}-\ref{LSI of PGD}, they derive the convergence without restricting the size of interactions, which subsumes many settings above. Notably, the techniques of our Theorems~\ref{convergence rate of the UMFLD} and~\ref{convergence of particle dynamics} are adopted from \citet{chen2023uniform} based on hypocoercivity where we consider other choices of $\gamma$ to improve the decaying rate of \UMLD{} and \UNLD{} established in \citet{chen2023uniform}.

\section{N-particle underdamped Langevin algorithm}
\label{s-t disc}
Our first step is to establish the global convergence of the {\em mean-field underdamped Langevin algorithm} ({MULA}),
\hypertarget{UMLA}{
\begin{subequations}
    \begin{align*}
    \label{dis-UMFLD}
    \tag{{MULA}}
       \mathrm{d}x_t&=v_t\mathrm{d}t,\\
       \mathrm{d}v_t&=-\gamma v_t\mathrm{d}t-D_{\rho}F(\mu^X_{kh},x_{kh})\mathrm{d}t+\sqrt{2\gamma}\mathrm{dB}_t,
    \end{align*}
\end{subequations}}
for stepsize $h$, $t\in[kh,(k+1)h]$ and $k=1,...,K$.  Note that~\MULA{} is the EI time-discretization of the~\ref{UMFLD}, where each step will now require integrating from $t=kh$ to $t=(k+1)h$ for stepsize $h$. 
\UMLA{} is intractable to implement in most instances given we do not often have access to $\mu_{kh}^X$ per iteration. This prompts us to consider the particle approximation which uses $\mu_{\textbf{x}_{kh}}=\frac{1}{N}\sum_{i=1}^N\delta_{x_{kh}^i}$ to approximate $\mu_{kh}^X$ where $(x_k^i)_{i=1}^N$ are iid samples from $\mu^X_k$:
\begin{equation}
    \label{dis-par-UMFLD}
    \begin{aligned}
       \mathrm{d}x_t^i&=v^i_t\mathrm{d}t,\\
       \mathrm{d}v_t^i&=-\gamma v^i_t\mathrm{d}t-D_{\rho}F(\mu_{\textbf{x}_{kh}},x^i_{kh})\mathrm{d}t+\sqrt{2\gamma}\mathrm{dB}^i_t,
    \end{aligned}
\end{equation}
for stepsize $h$, $t\in[kh,(k+1)h]$, $i=1,...,N$, $k \in \mathbb{N}$
and $\mu_{\textbf{x}_{kh}}=\frac{1}{N}\sum_{i=1}^N\delta_{x^i_{kh}}$.
Integrating the particle system~ \eqref{dis-par-UMFLD} from $t=kh$ to $t=(k+1)h$ for stepsize $h$ and $i=1,...,N$, we obtain our proposed Algorithm~\ref{UNLA} which we refer to as the {\em N-particle underdamped  Langevin algorithm} ({N-ULA}).

\hypertarget{UNLA}{
\begin{algorithm}[h]
\caption{ N-particle underdamped Langevin algorithm ({NULA})}
\label{UNLA}
\begin{algorithmic}[1]
\REQUIRE{$F$ satisfies Assumptions~\ref{convexity}-\ref{second moment bound}} and~\ref{initialization and invariant for N-particle}
\STATE{Initialize $\textbf{x}_0 = (x_0^1,...,x_0^N),\, \textbf{v}_0=(v_0^1,...,v_0^N)$, $h, \gamma$ \\
Specify $\varphi_0$, $\varphi_1$, $\varphi_2$, $\Sigma_{11}$, $\Sigma_{12}$, $\Sigma_{22}$ using \eqref{choice of varphi} and \eqref{choice of Sigma}.}
\FOR{$k=0,...,K-1$}
\FOR{$i=1,...,N$}
\STATE{$\left[\begin{array}{c}
    (\mathrm{B}_k^i)^x  \\
    (\mathrm{B}_k^i)^v \end{array}\right]\sim\mathcal{N}\left(0,\left[\begin{array}{cc}
    \Sigma^{11}I_d & \Sigma^{12}I_d \\
    \Sigma^{12}I_d & \Sigma^{22}I_d
\end{array}\right]\right)$}
\STATE{$x^i_{k+1}=x^i_k+\varphi_0\,v^i_{k}-\varphi_1\,
    D_{\mu}F(\mu_{\textbf{x}_{k}},x^i_{k})+(\mathrm{B}_k^i)^x$}
\STATE{$v^i_{k+1}=\varphi_2\,v^i_{k}-\varphi_0\, D_{\mu}F(\mu_{\textbf{x}_{k}},x^i_{k})+(\mathrm{B}_k^i)^v$}
\ENDFOR
\ENDFOR
\RETURN $(x^1_K,...,x^N_K)$
\end{algorithmic}
\end{algorithm}}
The update parameters of Algorithm \ref{UNLA}, $\varphi_0,\,\varphi_1,\,\varphi_2$ and  $\Sigma_{11}$, $\Sigma_{12}$, $\Sigma_{22}$, are functions of $\gamma$ and stepsize $h$. Thus, we need to specify the value of $\gamma$ and $h$ to compute the update parameters and initialize $(\textbf{x}_0,\,\textbf{v}_0)\sim\mu^N_0\in\mathcal{P}_2(\mathbb{R}^{2Nd})$ before running the algorithm.

\subsection{Convergence analysis}
We begin by leveraging entropic hypocoercivity and Theorems 2.1 and 2.2 from \citet{chen2023uniform} to analyze the continuous-time dynamics \UMLD{} and \UNLD. Let \begin{equation}
\label{matrix S}
S=\left(\begin{array}{cc}
     {1}/{\mathscr{L}}&{1}/{\sqrt{\mathscr{L}}}  \\
     {1}/{\sqrt{\mathscr{L}}}&2 
\end{array}\right)\otimes I_d.\end{equation} 
We construct the Lyapunov functional similar to~\citet{chen2023uniform}, but with a different choice of $S$. 
 Theorem~\ref{convergence rate of the UMFLD} is established by showing the following functional is decaying along the trajectory of MULD.
\begin{align}\label{eq:Lyap1}
   \mathcal{E}(\mu)&\defeq\mathcal{F}(\mu)+\textsf{FI}_{S}(\mu\|\hat{\mu}), \,\text{where}\\ 
      \mathcal{F}(\mu)&\defeq F(\mu^X)+\int\frac{1}{2}\|v\|^2\mu(\mathrm{d}x\mathrm{d}v)+\text{Ent}(\mu). \notag 
\end{align}
Our second Theorem~\ref{convergence of particle dynamics} establishes the convergence of~\ref{Par-UMFLD}. Denote
$\textbf{x}=(x^1,...,x^N)$, $\textbf{v}=(v^1,...,v^N)$, $\mu^N=\text{Law}(\textbf{x},\textbf{v})$, and $\mu_*^N$ as the limiting distribution of~\UNLD{} satisfying $\mu_*^N(\textbf{x},\textbf{v})\propto\exp\left(-NF(\mu_{\textbf{x}})-\frac{1}{2}\|\textbf{v}\|^2\right)$ (see the derivation of limiting distribution in Appendix \ref{NULD background}).
Denote $\nabla_i \defeq (\nabla_{x^i}, \nabla_{v^i})^{\mathsf{T}}$.
We obtain our guarantee by showing the functional is decaying along the trajectory of~\UNLD{}:
\begin{align}
\label{eq:Lyap2}
\mathcal{E}^N(\mu^N)&\defeq\mathcal{F}^N(\mu^N)+\textsf{FI}^N_S(\mu^N\|\mu_*^N), \,\text{where} \\  
\textsf{FI}^N_S(\mu^N\|\mu_*^N)&\defeq \sum_{i=1}^N\mathbb{E}_{\mu^N}\Big\|S^{1/2}\nabla_i\log\frac{\mu^N}{\mu_*^N}\Big\|^2, \, \text{and} \notag \\
\mathcal{F}^N(\mu^N)&\defeq\int N F(\mu_{\textbf{x}})+\frac{1}{2}\|\textbf{v}\|^2 \mu^N(\mathrm{d}\textbf{x}\mathrm{d}\textbf{v})+\text{Ent}(\mu^N).\notag
\end{align}
\begin{tcolorbox}[frame empty, left = 0.1mm, right = 0.1mm, top = 0.1mm, bottom = 0.1mm]
\begin{theorem}[Mean-field underdamped Langevin dynamics]
\label{convergence rate of the UMFLD}
If Assumptions \ref{convexity}-\ref{LSI of PGD} hold,   $\mu_0$ has finite second moment, finite entropy and finite Fisher information, then the law $\mu_t$ of the \ref{UMFLD} with $\gamma=\sqrt{\mathscr{L}}$ and $\mathcal{E}$ defined in~\eqref{eq:Lyap1} satisfy,
    \begin{equation*}
        \mathcal{F}(\mu_t)-\mathcal{F}(\mu_*)\leq (\mathcal{E}(\mu_0)-\mathcal{E}(\mu_*))\exp\left(-\frac{\mathscr{C}_{\textsf{\em LSI}}}{3\sqrt{\mathscr{L}}}t\right).
    \end{equation*}
\end{theorem}
\end{tcolorbox}
\begin{tcolorbox}[frame empty, left = 0.1mm, right = 0.1mm, top = 0.1mm, bottom = 0.1mm]
\begin{theorem}[N-particle underdamped Langevin dynamics] 
\label{convergence of particle dynamics}

 If Assumptions \ref{convexity}-\ref{LSI of PGD} hold, $\mu_0^N$ has finite second moment, finite entropy, finite Fisher information, and $N\geq(\mathscr{L}/\mathscr{C}_{\textsf{\em LSI}})\left(32+ 24\mathscr{L}/\mathscr{C}_{\textsf{\em LSI}}\right)$, then the joint law $\mu_t^N$ of the~\ref{Par-UMFLD}~with $\gamma=\sqrt{\mathscr{L}}$ and $\mathcal{E}^N$ defined in \eqref{eq:Lyap2} satisfy
    \begin{equation*}
        \frac{1}{N}\mathcal{F}^N(\mu_t^N)-\mathcal{F}(\mu_*)\leq \frac{\mathcal{E}_0^N}{N}\exp\left(-\frac{\mathscr{C}_{\textsf{\em LSI}}}{6\sqrt{\mathscr{L}}}t\right)+\frac{\mathcal{B}}{N},
    \end{equation*}
    where $\mathcal{B}=\frac{60\mathscr{L}d}{\mathscr{C}_{\textsf{\em LSI}}}+\frac{36\mathscr{L}^{2}d}{\mathscr{C}^2_{\textsf{\em LSI}}}$, $\mathcal{E}_0^N\defeq\mathcal{E}^N(\mu_0^N)-N\mathcal{E}(\mu_*)$. 
\end{theorem}
\end{tcolorbox}
Note that $\mathcal{E}_0^N=\mathcal{F}^N(\mu_0^N)-N\mathcal{F}(\mu_*)+\textsf{FI}^N_S(\mu_0^N\|\mu_*^N)\geq 0$ by Lemma \ref{Particle System's Entropy Inequality}.
The decaying rate given in Theorem~\ref{convergence rate of the UMFLD} resembles the decaying rate of ULD in \citet{zhang2023improved} with similar choices of $\gamma$ and $S$. Theorem~\ref{convergence of particle dynamics} implies the non-uniform-in-$N$ convergence of \ref{Par-UMFLD}, which incorporates a bias term involving $N$ due to the particle approximation. Our proof technique is more refined but parallel to that of~\citet{chen2023uniform} where 
our faster convergence and smaller bias is achieved by choosing $\gamma=\sqrt{\mathscr{L}}$ instead of $\gamma = 1$ (see Table~\ref{summary}).

Our main results analyze the convergence of the discrete-time processes \ref{dis-UMFLD} and \hyperlink{UNLA}{{N-ULA}} as well as their mixing time guarantees to generate an $\epsilon$-approximate solution in TV distance with the specific choice of initialization, damping coefficient $\gamma$, and stepsize $h$.
\begin{tcolorbox}[frame empty, left = 0.1mm, right = 0.1mm, top = 0.1mm, bottom = 0.1mm]
\begin{theorem}[Mean-field underdamped Langevin algorithm]
\label{iter complexity of mean-field system}
    In addition to the assumptions specified in Theorems \ref{convergence rate of the UMFLD}, let Assumptions~\ref{bounded grad}-\ref{initialization and invariant} hold. Denote $\bar{\mu}_{K}$ the law of $(x_K,v_K)$ of the \hyperlink{MULA}{\em \textsf{MULA}} and $\kappa\defeq \mathscr{L}/\mathscr{C}_{\sf{LSI}}$. Then in order to ensure $\|\bar{\mu}_{K}-\mu_*\|_{\textsf{\em TV}}\leq\epsilon$, it suffices to choose $\gamma=\sqrt{\mathscr{L}}$, $\bar{\mu}_{0}=\mathcal{N}(0,I_{2d})$, and
    \begin{equation*}
    h=\widetilde{\Theta}\left(\frac{\mathscr{C}_{\textsf{\em LSI}}\epsilon}{\mathscr{L}^{3/2}d^{1/2}}\right),\ K=\widetilde{\Theta}\left(\frac{\kappa^2d^{1/2}}{\epsilon}\right).
\end{equation*}
\end{theorem}
\end{tcolorbox}
A similar guarantee can be stated for the $N$-particle system~\eqref{dis-par-UMFLD} with the additional requirement that the number of particles scale according to the dimension of the problem and problem parameters. 
\begin{tcolorbox}[frame empty, left = 0.1mm, right = 0.1mm, top = 0.1mm, bottom = 0.1mm]
\begin{theorem}[N-particle underdamped Langevin algorithm]
\label{iter complexity for particle system}
      In addition to the assumptions specified in Theorem \ref{convergence of particle dynamics}, let Assumptions~\ref{bounded grad}, \ref{second moment bound} and~\ref{initialization and invariant for N-particle} hold. Denote $\bar{\mu}^i_{K}$ the law of $(x_K^i,v_K^i)$ of the \hyperlink{NULA}{\em \textsf{NULA}} for $i=1,...,N$ and $\kappa\defeq \mathscr{L}/\mathscr{C}_{\mathsf{LSI}}$. Then in order to ensure $ \frac{1}{N}\sum_{i=1}^N\|\bar{\mu}_{K}^i-\mu_*\|_{\textsf{\em TV}}\leq\epsilon$, it suffices to choose $\gamma=\sqrt{\mathscr{L}}$, $\bar{\mu}^N_0=\mathcal{N}(0,I_{2Nd})$,
     \begin{equation*}
         h=\widetilde{\Theta}\left(\frac{\mathscr{C}_{\textsf{\em LSI}}\epsilon}{\mathscr{L}^{3/2}d^{1/2}}\right),\ K=\widetilde{\Theta}\left(\frac{\kappa^2d^{1/2}}{\epsilon}\right),
     \end{equation*}
     and the number of particles
        $ N=\Theta\left(\kappa^{2}d/\epsilon^2\right).$
\end{theorem}
\end{tcolorbox}

\begin{table}[t]
    \centering
    \scalebox{0.95}{
    \begin{tabular}{cccc}
    \toprule
       \textbf{Discretization} &\textbf{Method} & \textbf{\# of particles} & \textbf{Mixing time} \\ \hline Time-discretizations & 
         MLA~\citep{nitanda2022convex}     &      *        &      $\widetilde{\Theta}\Big(\kappa^2 \mathscr{L}d/ \epsilon^2\Big)$       \\  \cline{2-4}
& \ref{discretization of ULD}~\citep{zhang2023improved}   &   *          &           $\widetilde{\Theta}\Big(\kappa^{3/2}d^{1/2}/\epsilon\Big)$       \\  \cline{2-4} &
{\hyperlink{MULA}{{MULA}}} {\bf (Ours)}  &   *           &        {\bf $\widetilde{\Theta}\left(\kappa^{2}d^{1/2}/\epsilon\right)$}         \\ \hline Space-discretizations &\ref{Par-UMFLD} (\citet{chen2023uniform})& ${\Theta}\left(\kappa^2\mathscr{L}d/\epsilon^2\right)$ & $\widetilde{\Theta}\Big(\kappa\Big)$\\
\cline{2-4} & \ref{Par-UMFLD}  {\bf (Ours)} &  $\Theta\Big(\kappa^2d/\epsilon^2\Big)$ & {$\widetilde{\Theta}\Big(\kappa/\mathscr{L}^{1/2}\Big)$}
\\
\hline
spacetime discretizations& \ref{NLA}~\citep{suzuki2023convergence}    &    ${\Theta}\left(
\kappa\mathscr{L}^3d/\epsilon^2\right)$    &    $\widetilde{\Theta}\left(\kappa^2\mathscr{L}d/\epsilon^2\right)$        \\  \cline{2-4}
& {\hyperlink{NULA}{{NULA}}} {\bf (Ours)} &    {$\Theta\left(\kappa^2d/\epsilon^2\right)$}         &            $\widetilde{\Theta}\Big(\kappa^2d^{1/2}/\epsilon\Big)$     \\ \bottomrule
    \end{tabular}}
    \caption{Comparison of algorithms in terms of the mixing time and number of particles to achieve $\epsilon$-approximate solutions in TV distance.  $\kappa \defeq \mathscr{L}/\mathscr{C}_{\textsf{LSI}}$. * represents that we do not need particle approximation for this method.}
    \label{summary}
\end{table}
\subsection{Proof sketches}
For the continuous-time results, we outline the proof of Theorem~\ref{convergence rate of the UMFLD} (and analogously Theorem~\ref{convergence of particle dynamics}) in this section to provide intuition for how choosing $\gamma=\sqrt{\mathscr{L}}$ can improve the decaying rate of \UMLD{}. We begin with a review of some notations of hypocoercivity in \citet{villani2009hypocoercivity,chen2023uniform}:
\begin{equation*}
    A_t=\nabla_v,\ C_t=\nabla_x,\ Y_t=\left(\|A_tu_t\|_{L^2(\mu_t)},\|A_t^2u_t\|_{L^2(\mu_t)}, \|C_tu_t\|_{L^2(\mu_t)}, \|C_tA_tu_t\|_{L^2(\mu_t)}\right)^{\mathsf{T}},
\end{equation*}
where $u_t=\log\frac{\mu_t}{\hat{\mu}_t}.$ Inheriting the analysis of Theorem 2.1 in \citet{chen2023uniform} and Lemma 32 in \citet{villani2009hypocoercivity}, we show that for a general $\gamma$, the Lyapunov functional \eqref{eq:Lyap1} with $S=[s_{ij}]\otimes I_d\in\mathbb{R}^{2d\times 2d}$ is decreasing along \UMLD{} satisfying
\begin{equation}
\label{decreasing Lyap}
\frac{\mathrm{d}}{\mathrm{d}t}\mathcal{E}(\mu_t)\leq -Y_t^{\mathsf{T}}\mathcal{K}Y_t,
\end{equation}
where $s_{11}=c$, $s_{12}=s_{21}=b$, $s_{22}=a$ and $\mathcal{K}$ is an upper triangle matrix with diagonal elements $(\gamma+2\gamma a-4\mathscr{L}b,\,2\gamma a,\,2b,\,2\gamma c).$ To ensure $S\succ 0$ and the right hand side of \eqref{decreasing Lyap} negative, the criteria of choosing positive constants $a,\,b,\,c$ should be $ac>b^2$ and $\mathcal{K}\succ 0.$ If we specify $\gamma=1$, we can choose $a=c=2\mathscr{L}$ and $b=1$ satisfying the criteria. Then we obtain $\lambda_{\textsf{min}}(\mathcal{K})=1$ and
\begin{equation*}
    \begin{aligned}
        \frac{\mathrm{d}}{\mathrm{d}t}\mathcal{E}(\mu_t)\leq -\lambda_{\textsf{min}}(\mathcal{K})Y_t^{\mathsf{T}}Y_t&\leq-\mathscr{C}_{\textsf{LSI}}(\mathcal{F}(\mu_t)-\mathcal{F}(\mu_*))-\frac{1}{2\lambda_{\textsf{max}}(S)}\mathsf{FI}_S(\mu_t\|\hat{\mu_t})\\
&\leq-\frac{\mathscr{C}_{\textsf{LSI}}}{6\mathscr{L}}(\mathcal{E}(\mu_t)-\mathcal{E}(\mu_*))
    \end{aligned}
\end{equation*}
Applying Gr{\"o}nwall's inequality leads to the decaying rate $O(\mathscr{C}_{\textsf{LSI}}/\mathscr{L})$ of \UMLD{} ($\gamma=1$) in \citet{chen2023uniform}. If we specify $\gamma=\sqrt{\mathscr{L}}$, we can choose $b=1/\sqrt{\mathscr{L}},\,a=2,\,c=1/\mathscr{L}$ satisfying the criteria. Then we obtain $\lambda_{\textsf{min}}(\mathcal{K})=2/\sqrt{\mathscr{L}}$ and
\begin{equation*}
    \begin{aligned}
        \frac{\mathrm{d}}{\mathrm{d}t}\mathcal{E}(\mu_t)\leq -\lambda_{\textsf{min}}(\mathcal{K})Y_t^{\mathsf{T}}Y_t&\leq-\frac{2\mathscr{C}_{\textsf{LSI}}}{\sqrt{\mathscr{L}}}(\mathcal{F}(\mu_t)-\mathcal{F}(\mu_*))-\frac{1}{\lambda_{\textsf{max}}(S)\sqrt{\mathscr{L}}}\mathsf{FI}_S(\mu_t\|\hat{\mu}_t)\\
        &\leq-\frac{\mathscr{C}_{\textsf{LSI}}}{3\sqrt{\mathscr{L}}}(\mathcal{E}(\mu_t)-\mathcal{E}(\mu_*))
    \end{aligned}
\end{equation*}
Applying Gr{\"o}nwall's inequality leads to the improved decaying rate $O(\mathscr{C}_{\textsf{LSI}}/\sqrt{\mathscr{L}})$ of \UMLD{} ($\gamma=\sqrt{\mathscr{L}}$) in our Theorem~\ref{convergence rate of the UMFLD}. We defer the whole proof to Appendix~\ref{Continuous-time results}.

For discretization errors, we outline the proof of Theorem~\ref{iter complexity of mean-field system} (and analogously Theorem \ref{iter complexity for particle system}) in this section.  Let $(\mu_{t})_{t\geq 0}$ and $(\Bar{\mu}_{t/h})_{t\geq 0}$ represent the law of~\ref{UMFLD} and \ref{dis-UMFLD} initialized at $\mu_0$. Let $\textbf{Q}_{kh}$ and $\textbf{P}_{kh}$ denote probability measures of \UMLD{} and \UMLA{} on the space of paths $C([0,kh],\mathbb{R}^{2d})$. Invoking Girsanov's theorem~\citep{girsanov1960transforming,kutoyants2004statistical,le2016brownian} and Assumption~\ref{smoothness}, we can upper bound the pathwise divergence between \UMLD{} and \UMLA{} in KL divergence for stepsize $h$ and $k=1,...,K$ under Assumptions~\ref{smoothness} and \ref{bounded grad}:
\begin{equation}
\label{discretization error in KL}
    \textsf{KL}(\textbf{Q}_{Kh}\|\textbf{P}_{Kh})\lesssim\frac{\mathscr{L}^4h^5}{\gamma}\sum_{k=0}^{K-1}\mathbb{E}_{\mathbf{Q}_{Kh}}\|x_{kh}\|^2+\frac{\mathscr{L}^2h^3}{\gamma}\sum_{k=0}^{K-1}\mathbb{E}_{\mathbf{Q}_{Kh}}\|v_{kh}\|^2+\frac{\mathscr{L}^4h^5K}{\gamma}+{\mathscr{L}^2h^4Kd}
\end{equation}
The derivation of~\eqref{discretization error in KL} is similar to that of \citet{zhang2023improved}; they establish the discretization error of \ref{discretization of ULD} in $q$-th order R{\'e}nyi divergence ($q\in[1,2)$), which has KL divergence as a special case  ($q=1$). Their smoothness assumption on the potential function $V$ is $(\mathscr{L},s)$-weak smoothness, which recovers $\mathscr{L}$-smoothness when $s=1$. 
We use many similar techniques of bounding the discretization error to those of~\citet{zhang2023improved}. Their 
Lemma 26 can be generalized to our Lemma~\ref{bound of the iterate difference} in the mean-field setting, which describes an intermediate process of deriving \eqref{discretization error in KL}.
Applying the data processing inequality, we can upper bound the KL divergence between the time marginal laws of the iterates by KL divergence between path measures:
\begin{equation*}
    \textsf{KL}(\mu_T\|\Bar{\mu}_K)\leq\textsf{KL}(\textbf{Q}_{Kh}\|\textbf{P}_{Kh}),
\end{equation*}
where $T=Kh$. Uniformly upper bounding the right-hand side of \eqref{discretization error in KL} requires obtaining uniform bounds for $\mathbb{E}_{\mathbf{Q}_{Kh}}\|x_{kh}\|^2$ and $\mathbb{E}_{\mathbf{Q}_{Kh}}\|v_{kh}\|^2$; 
If we were to rely on existing techniques~\cite{zhang2023improved}, we would need a $\chi^2$-convergence guarantee of \UMLD{}.
Given $\chi^2$-convergence is not established for \UMLD{} by previous works, we develop different techniques to uniformly upper bound the iterates of \UMLD{} and \UNLD{}. More specifically, we have
\begin{equation*}
    \mathbb{E}_{\textbf{Q}_T}\|(x_t,v_t)\|^2=W_2^2(\mu_t,\delta_0)\lesssim \underbrace{W_2^2(\mu_t,\mu_*)}_{\textsf{I}}+\underbrace{W_2^2(\mu_*,\delta_0)}_{\textsf{II}},\ t\in[0,T],
\end{equation*}
where $\delta_0$ is Dirac measure on $\textsf{0}\in\mathbb{R}^{2d}$, and $\textsf{II}$ is the second moment of $\mu_*$ denoted by $\textbf{m}_2^2$. Now we need to upper bound $\textsf{I}$. Under Assumption~\ref{LSI of PGD}, $\mu_*$ satisfies LSI implying Talagrand's inequality: $\textsf{I}\lesssim\textsf{KL}(\mu_t\|\mu_*)/\mathscr{C}_{\textsf{LSI}}$. Under Assumptions~\ref{convexity} and \ref{smoothness}, Lemma 4.2 in \citet{chen2023uniform} establishes the following relation between KL divergence and energy gap: 
\begin{equation}
\label{KL<F}
\textsf{KL}(\mu_t\|\mu_{*})\leq\mathcal{F}(\mu_t)-\mathcal{F}(\mu_*).
\end{equation}
Moreover, \citet{kazeykina2020ergodicity,chen2023uniform} demonstrate that $\mathcal{F}(\mu_t)$ is decreasing along \UMLD{}. According to two conclusions above, $\textsf{I}$ can be bounded as
\begin{equation*}
    \textsf{I}\lesssim\frac{\textsf{KL}(\mu_t\|\mu_{*})}{\mathscr{C}_{\textsf{LSI}}}\leq \frac{\mathcal{F}(\mu_t)-\mathcal{F}(\mu_*)}{\mathscr{C}_{\textsf{LSI}}}\leq\frac{\mathcal{F}(\mu_0)-\mathcal{F}(\mu_*)}{\mathscr{C}_{\textsf{LSI}}}\leq \frac{\mathcal{F}(\mu_0)}{\mathscr{C}_{\textsf{LSI}}},
\end{equation*}
where the last inequality follows from the assumption that $\mathcal{F}(\mu_*)\geq 0$. Therefore, under Assumptions~\ref{second moment bound} on $\textbf{m}_2^2$ and \ref{initialization and invariant} on $F(\mu_0)$, our Lemma \ref{moment bound mean-field} establishes the upper bound of $\mathbb{E}_{\textbf{Q}_T}\|(x_t,v_t)\|^2$ in terms of $\mathscr{L},\,\mathscr{C}_{\textsf{LSI}}$ and $d$, which implies the uniform upper bound of $ \textsf{KL}(\mu_T\|\Bar{\mu}_K)$.
Applying Pinsker's inequality $$\|\bar{\mu}_{K}-\mu_T\|_{\textsf{TV}}\lesssim\sqrt{\textsf{KL}(\mu_T\|\bar{\mu}_{K})},$$ 
we can convert the discretization error bound in KL divergence to that in TV distance.
Combining Pinsker's inequality and relation \eqref{KL<F}, we derive the continuous-time convergence of \UMLD{} in Theorem~\ref{convergence rate of the UMFLD} in TV distance:
\begin{equation}
\label{continuous in TV}
    \begin{aligned}
        \|\mu_{T}-\mu_*\|_{\textsf{TV}}\lesssim\sqrt{\textsf{KL}(\mu_T\|\mu_*)}\leq\sqrt{\mathcal{F}(\mu_T)-\mathcal{F}(\mu_*)}.
    \end{aligned}
\end{equation}
Applying the triangle inequality to $\|\bar{\mu}_{K}-\mu_*\|_{\textsf{TV}}$, the TV distance between the law of \ref{dis-UMFLD} at $Kh$ and the limiting distribution of \UMLD{}, we obtain the global convergence of \UMLA{}:
\begin{equation*}
    \begin{aligned}
        \|\bar{\mu}_{K}-\mu_*\|_{\textsf{TV}}&\leq\underbrace{\|\bar{\mu}_{K}-\mu_{T}\|_{\textsf{TV}}}_{\mathcal{B}}+\underbrace{\|\mu_{T}-\mu_*\|_{\textsf{TV}}}_{\mathcal{V}},
    \end{aligned}
\end{equation*}
where $\mathcal{V}$ vanishes exponentially fast as $T\rightarrow\infty$ and $\mathcal{B}$ is a vanishing bias as $h\rightarrow 0$. To ensure $\mathcal{V}+\mathcal{B}\leq\epsilon$, it suffices to choose $T=\widetilde{\Theta}(\sqrt{\mathscr{L}}/\mathscr{C}_{\textsf{LSI}})$ and specify $h,\,K$ as in Theorem~\ref{iter complexity of mean-field system}. The whole proof is deferred to Appendix \ref{Proof of the Discrete-time-space Convergence}. 

\subsection{Discussion of mixing time results}
We summarize the convergence results of \hyperlink{MULA}{\textsf{MULA}}, \hyperlink{NULA}{\textsf{NULA}} and several existing methods including \ref{discretization of ULD}, the EM-discretization of~\ref{MFLD} (referred to as \MLA{}~\citep{nitanda2022convex}), and its finite-particle system (referred to as \ref{NLA}~\citep{suzuki2023convergence}) 
in Table~\ref{summary}. For the mixing time to generate an $\epsilon$-approximate solution in TV distance, our proposed \UMLA{} and \UNLA{} achieve better dependence on $\mathscr{L}$, $d$ and $\epsilon$ than \MLA{} and \NLA, and keep the same dependence on $\mathscr{C}_{\textsf{LSI}}$ as \MLA{} and \NLA{}, which justifies that our methods are fast. For the number of particles, we improve the dependence on $\mathscr{L}$ for \UNLD{} ($\gamma=\sqrt{\mathscr{L}}$) when compared with \UNLD{} ($\gamma=1$) in \citet{chen2023uniform} and for \UNLA{} when compared with \NLA{}. Particularly, our dependence on the smoothness constant in the number of particle guarantee of \UNLA{} is $\Theta(\mathscr{L}^2)$ whereas the counterpart of \NLA{} is $\Theta(\mathscr{L}^4)$. However, our dependence on the LSI constant in the number of particle guarantee of \UNLA{} is $\Theta(\mathscr{C}^{-2}_{\textsf{LSI}})$ whereas the counterpart of \NLA{} is $\Theta(\mathscr{C}^{-1}_{\textsf{LSI}})$.

Note that \citet{nitanda2022convex} consider \MLA{} in the neural network setting where they specifically choose $F$ to be the objective~\eqref{empirical risk minimization} and propose assumptions on $l$, $h$ and $r$. \citet{suzuki2023convergence} consider \ref{NLA} in a setting where they specify that $F(\mu)=U(\mu)+\mathbb{E}_{\mu}[r(x)]$ and propose assumptions on $U$ and $r$. Consequently, they use different notations of the smoothness constant and establish the convergence rate in energy gap $\mathcal{F}(\Bar{\mu}_{K})-\mathcal{F}(\mu_*)$ instead of the TV distance. To make a fair comparison, we equivalently translate those smoothness constants into $\mathscr{L}$ and convert convergence rates of \MLA{} and \NLA{} to those in TV distance by relation \eqref{KL<F} and Pinsker's inequality (see Appendix \ref{Translation of the Compared Methods}).

\section{Applications of Algorithm~\ref{UNLA}}
\label{sec:ex}
In this section, we will show how Algorithm~\ref{UNLA} can be applied to several applications by verifying Assumptions \ref{convexity}-\ref{initialization and invariant for N-particle} hold for these examples. We present these results in full details in~\cref{Verification of Assumption}.
\subsection{Training mean-field neural networks}
Consider a two-layer mean-field neural network (with infinite depth), which can be parameterized as $h(\rho;a)\defeq\mathbb{E}_{x\sim\rho}[h(x;a)],$
where $h(x;a)$ represents a single neuron with trainable parameter $x$ and input $a$ (e.g. $h(x;a)=\sigma(x^{\mathsf{T}}a)$ for activation function $\sigma$); $\rho$ is the probability distribution of the parameter $x$. Given dataset $(a_i,b_i)_{i=1}^n$ and loss function $\ell$, we choose $F$ in objective \eqref{augmented mean-field opt} to be
\begin{equation}
\label{empirical risk minimization}
    F(\mu^X)=\frac{1}{n}\sum_{i=1}^n\ell(h(\mu^X;a_i),b_i)+\frac{\lambda'}{2}\mathbb{E}_{x\sim \mu^X}\|x\|^2,
\end{equation}
The objectives 
\eqref{empirical risk minimization} satisfy Assumptions~\ref{convexity}-\ref{bounded grad} for specific common choices $\ell$ and $h$ described in several works ~\citep{nitanda2022convex,chen2022uniform,chen2023uniform,suzuki2023convergence}. 
If there exists $\mathscr{L}>0$ such that the activation function satisfies $|h(x;a)|\leq\sqrt{\mathscr{L}}$ (also proposed in \citet{suzuki2023convergence}) and the convex loss function $\ell$ is quadratic or satisfies $|\partial_1\ell|\leq \sqrt{\mathscr{L}}$ (also proposed in \citet{nitanda2022convex}), $F$ satisfies Assumption~\ref{second moment bound} with $\lambda'\leq(2\pi)^3\exp(-8\mathscr{L})$. Finally, if in addition we assume $\ell$ is $\sqrt{\mathscr{L}}$-Lipschitz and choose $\lambda'\leq(2\pi)^3\exp(-8\mathscr{L})$, $\mu_0=\mathcal{N}(0,I_{2d})$ and $\mu_0^N=\mathcal{N}(0,I_{2Nd})$, Assumptions~\ref{initialization and invariant} and~\ref{initialization and invariant for N-particle} will be satisfied.
\subsection{Density estimation via MMD minimization}
The maximum mean discrepancy between two probability measures $\rho$ and $\pi$ is defined as $\mathcal{M}(\rho\|\pi)=\iint [k(x,x)-2k(x,y)+k(y,y)]\mathrm{d}\rho(x)\mathrm{d}\pi(y),$
where $k$ is a positive definite kernel. Similar to Example 2 in~\citet{suzuki2023convergence}, we consider the non-parametric density estimation using the Gaussian mixture model, which can be parameterized as $p(\rho;z)\defeq\mathbb{E}_{x\sim\rho}[p(x;z)],$ where $p(x;z)$ is the Gaussian density function of $z$ with mean $x$ and a user-specified variance $\sigma^2$.
Given a set of samples $\{z_i\}_{i=1}^n$ from the target distribution $p^*$, our goal is to fit $p^*$ by minimizing the empirical version of $\mathcal{M}(p(\rho;z)\|p^*)$, defined as
\begin{equation*}
\begin{aligned}
    \hat{\mathcal{M}}(\rho)=\iiint p(x;z)p(x';z')k(z,z')\mathrm{d}z\mathrm{d}z'\mathrm{d}\rho(x)\mathrm{d}\rho(x')-2\int\left(\frac{1}{n}\sum_{i=1}^n\int p(x;z)k(z,z_i)\mathrm{d}z\right)\mathrm{d}\rho(x).
    \end{aligned}
\end{equation*}
We choose $F$ in objective \eqref{augmented mean-field opt} to be
\begin{equation}
\label{MMD}
    F(\mu^X)=\hat{\mathcal{M}}(\mu^X)+\frac{\lambda'}{2}\mathbb{E}_{x\sim\mu^X}\|x\|^2,
\end{equation}
where $\lambda'>0$. \citet{suzuki2023convergence} show that objective~\eqref{MMD} satisfies Assumptions \ref{convexity}, \ref{LSI of PGD} and \ref{bounded grad} by choosing a smooth and light-tailed kernel $k$, such as Gaussian radial basis function (RBF) kernel defined as $k(z,z')\defeq\exp(-\|z-z'\|^2/2\sigma'^2)$ for $\sigma'>0$. We also verify that objective~\eqref{MMD} also satisfies our Assumption \ref{smoothness} with the same choice of kernel. With Gaussian RBF kernel $k$ ($\sigma'=\sigma$), we provide verification in Appendix~\ref{Verification of Assumption} that objective~\eqref{MMD} satisfies Assumptions \ref{second moment bound}-\ref{initialization and invariant for N-particle} when $\lambda'\leq 3\pi/25$,
$\mu_0=\mathcal{N}(0,I_{2d})$ and $\mu_0^N=\mathcal{N}(0,I_{2Nd})$.
\subsection{Kernel Stein discrepancy minimization}
\label{Kernel Stein discrepancy minimization}
Kernel Stein discrepancy (KSD) minimization is a method for sampling from a target distribution $\rho_*$ if we have the access to the score function $s_{\rho_*}(x)=\nabla\log\rho_*(x)$ \citep{chwialkowski2016kernel,liu2016kernelized}. 
For a positive definite kernel $k$, the Stein kernel is defined as
\begin{equation*}
\begin{aligned}
    u_{\rho_*}(x,x')=s_{\rho_*}^{\mathsf{T}}(x)k(x,x')s_{\rho_*}(x')+s_{\rho_*}^{\mathsf{T}}(x)\nabla_{x'}k(x,x')+\nabla_x^{\mathsf{T}}k(x,x')s_{\rho_*}(x')+\textsf{tr}(\nabla_{x,x'}k(x,x')).
    \end{aligned}
\end{equation*}
The KSD between $\rho$ and $\rho_*$ is defined as
    $\textsf{KSD}(\rho)=\iint u_{\rho_*}(x,x')\mathrm{d}\rho(x)\mathrm{d}\rho(x').$
We choose $F$ in \eqref{augmented mean-field opt} to be
\begin{equation}
\label{KSD obj}
    F(\mu^X)=\textsf{KSD}(\mu^X)+\frac{\lambda'}{2}\mathbb{E}_{x\sim\mu^X}\|x\|^2,
\end{equation}
where $\lambda'>0$. \citet{suzuki2023convergence} show that objective \eqref{KSD obj} satisfies Assumptions \ref{convexity}, \ref{LSI of PGD} and \ref{bounded grad} by choosing light-tailed kernel and assume the score function satisfies 
\begin{equation}
\label{assump on score}
\max_{k=1,2,3}\{\|\nabla^{\otimes k}\log\rho_*(x)\|_{\textsf{op}}\}\leq \mathscr{L}(1+\|x\|). 
\end{equation}

More specifically, if $\mu_*\propto\exp(-V)$, the potential function $V$ should satisfies $$\max_{k=1,2,3}\{\|\nabla^{\otimes k}\nabla V(x)\|_{\textsf{op}}\}\leq \mathscr{L}(1+\|x\|),$$ which subsumes many distributions. 
Choosing the same kernel as in \citet{suzuki2023convergence}, we verify in Appendix \ref{Verification of Assumption} that \eqref{KSD obj} also satisfies Assumption \ref{smoothness} and satisfies our Assumptions \ref{second moment bound}-\ref{initialization and invariant for N-particle} with $\lambda'\leq\min\{(2\pi)^3\exp(-4\mathscr{L}),\mathscr{L},d\}$, 
$\mu_0=\mathcal{N}(0,I_{2d})$ and $\mu_0^N=\mathcal{N}(0,I_{2Nd})$.
\section{Numerical experiments}
We verify our theoretical findings by providing empirical support in this section. Our experiment\footnote{Code for our experiments can be found at \url{https://github.com/QiangFu09/NULA}.} is to approximate a Gaussian function $f(z)=\exp(-\|z-m\|^2/2d)$ for $z\in\mathbb{R}^d$ and unknown $m\in\mathbb{R}^d$ by a mean-field two-layer neural network with \textsf{tanh} activation.
 Consider the empirical risk minimization problem \eqref{empirical risk minimization} with quadratic loss function $\ell$, $d=10^3$, $\lambda'=10^{-4}$ and $n$ randomly generated data samples from $f(z)$ ($n=100$), described by
\begin{equation*}
    F(\rho)=\frac{1}{2n}\sum_{i=1}^n(h(\mu;a_i)-f(a_i))^2+\frac{\lambda'}{2}\mathbb{E}_{x\sim\rho}[\|x\|^2].
\end{equation*}
$F$ satisfy Assumptions~\ref{convexity}-\ref{initialization and invariant for N-particle} with the choice of $\ell,\,h$, and thus we apply Algorithm~\ref{UNLA} for minimizing the objective above. Note that the number of neurons in the first hidden layer is equivalent to the number of particles in \UNLA{}, and we choose $N\in\{256, 512, 1024, 2048\}$. 
\label{NE}
\begin{figure*}[ht!]
\centering
\subfigure{
\label{256}
\includegraphics[width=6.0cm,height = 4.8cm]{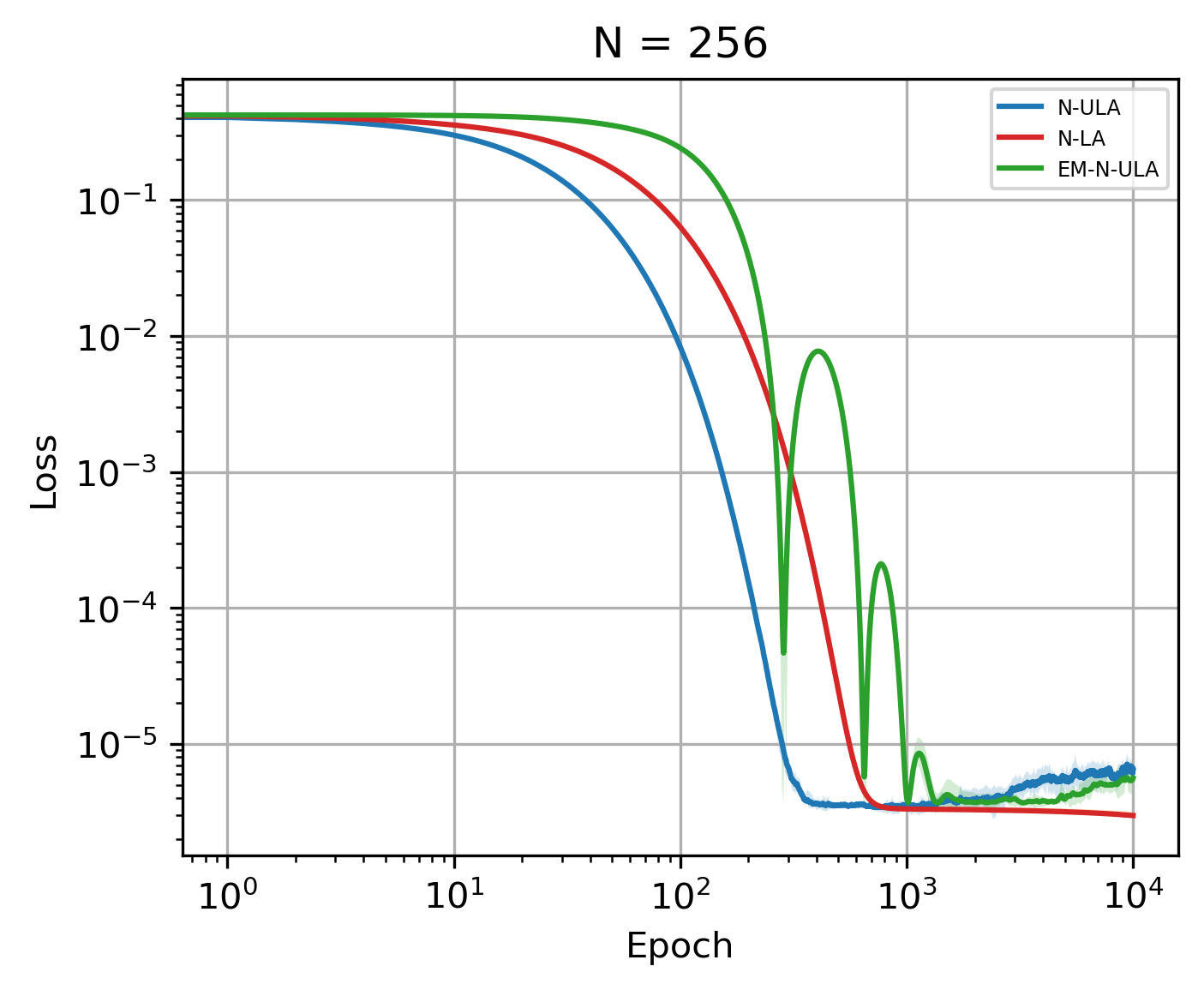}}\hspace{-1em}
\subfigure{
\label{512}
\includegraphics[width=6.0cm,height = 4.8cm]{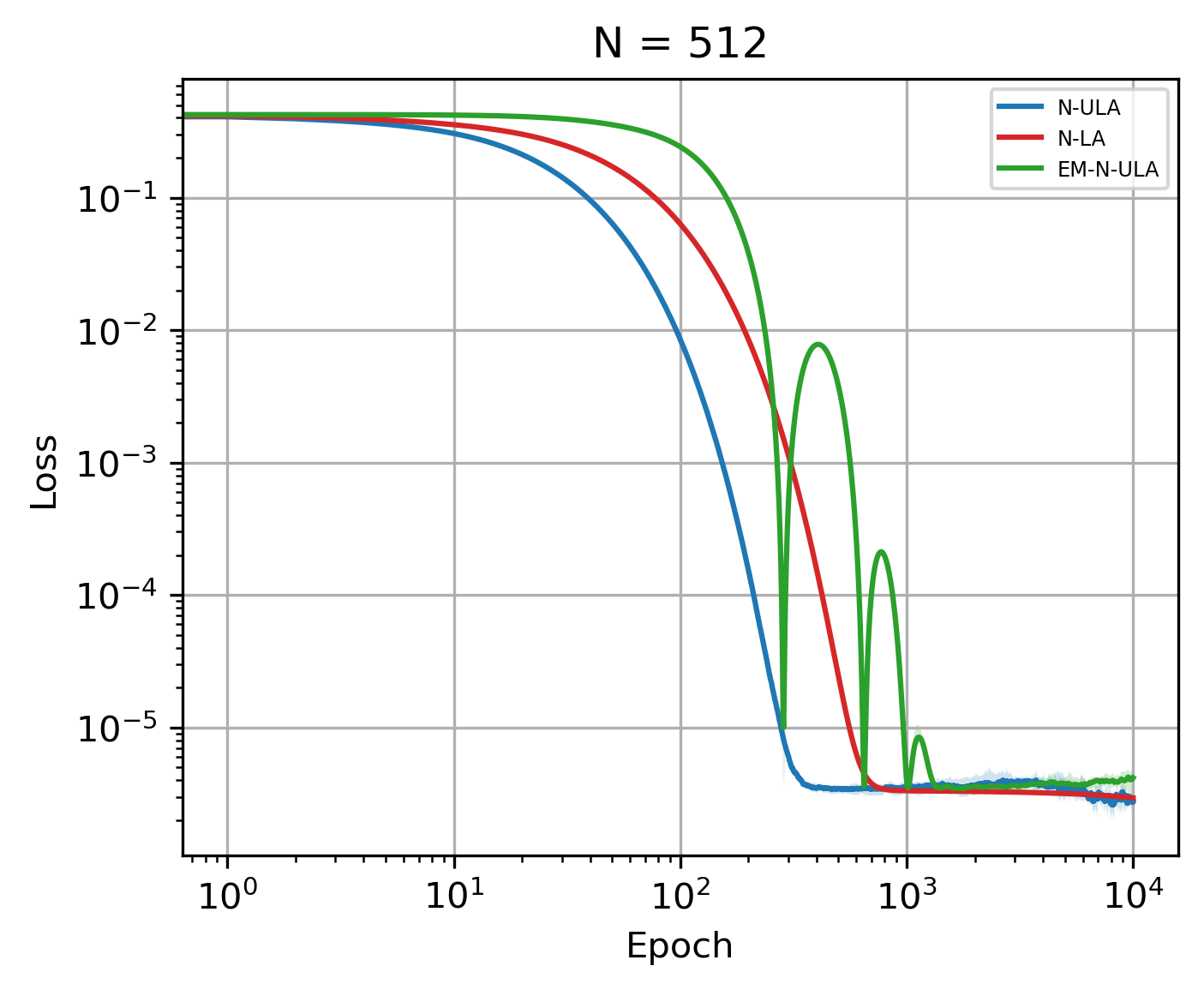}}\hspace{-1em}
\subfigure{
\label{1024}
\includegraphics[width=6.0cm,height = 4.8cm]{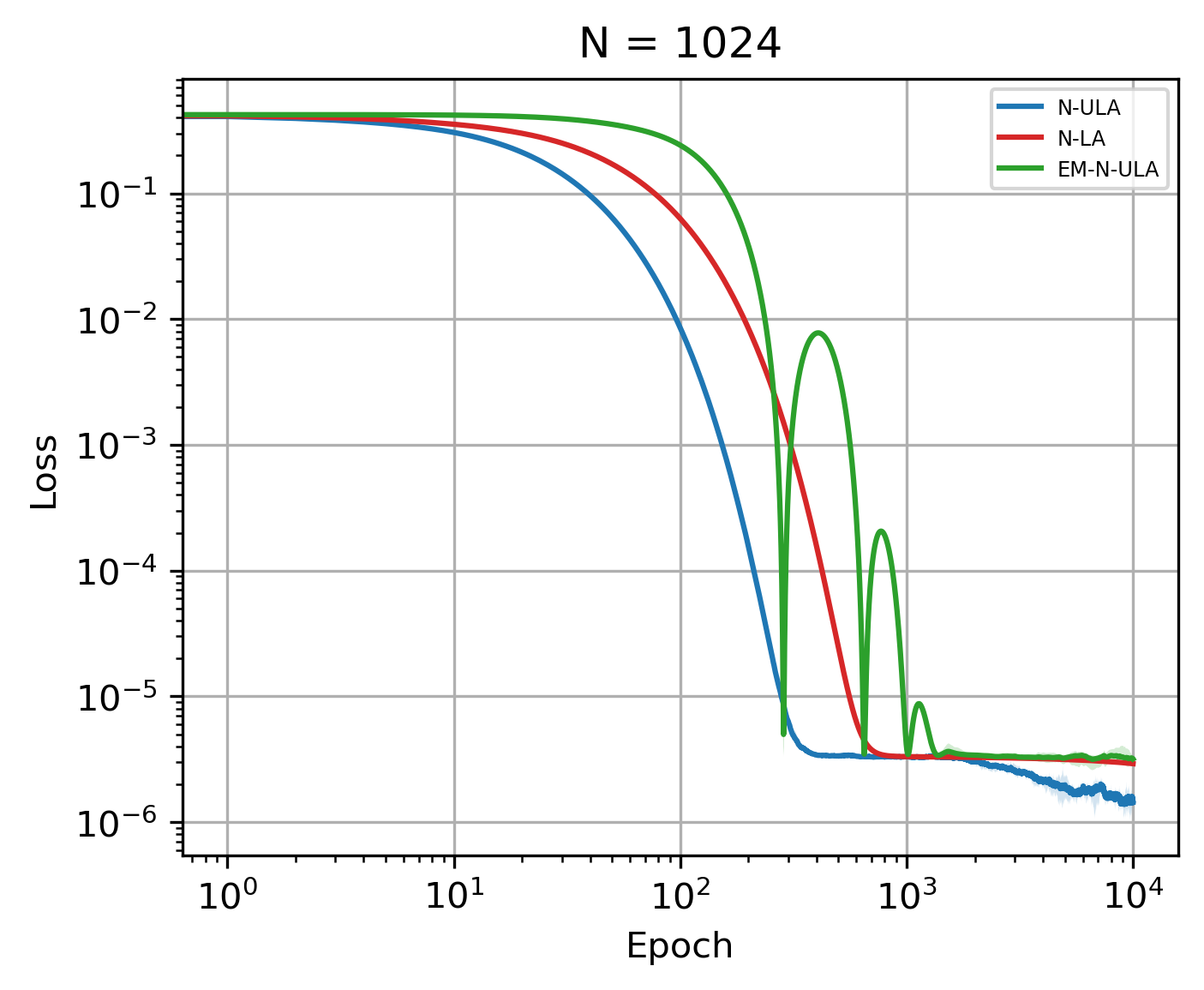}}\hspace{-1em}
\subfigure{
\label{2048}
\includegraphics[width=6.0cm,height = 4.8cm]{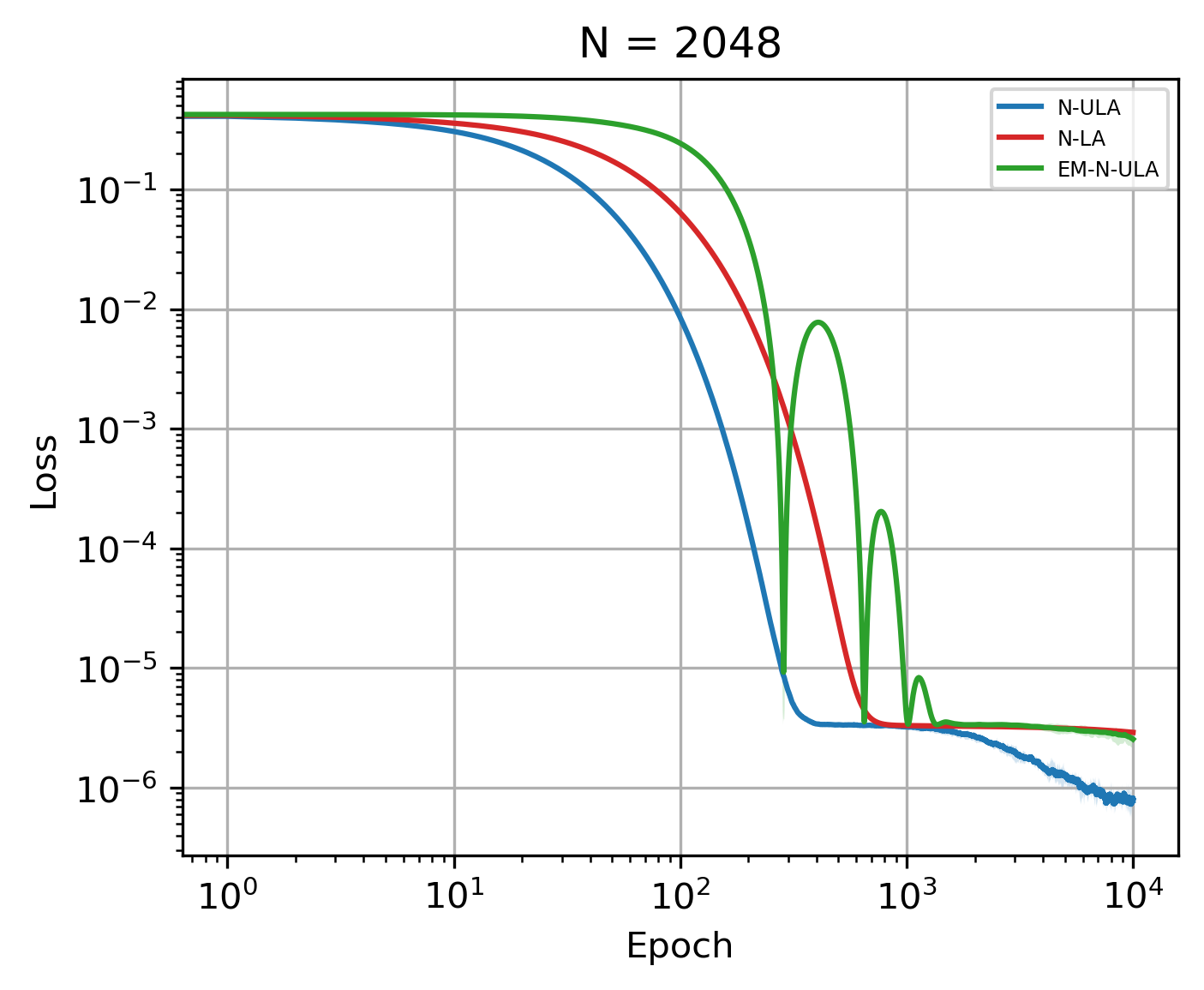}}
\caption{Evaluation on \UNLA{}, \NLA{} and \EMUNLA{} with different number of particles N where x-axis represents the training epochs and y-axis represents the value of $\frac{1}{2n}\sum_{i=1}^n(\frac{1}{N}\sum_{s=1}^Nh(x^s;a_i)-f(a_i))^2$. Our method often enjoys better performance in the high particle-approximation regime which is consistent with our theoretical findings.}
\label{outcomp}
\end{figure*}
The intrinsic derivative of $F$ for the $j$-th particle in our method is given
by 
$${D}_{\rho}F(\mu_{\textbf{x}},x^j)=\frac{1}{n}\sum_{i=1}^n(\frac{1}{N}\sum_{s=1}^Nh(x^s;a_i)-f(a_i))\nabla h(x^j;a_i)+\lambda' x^j.$$
Note that $\frac{1}{N}\sum_{s=1}^Nh(x_s;a)$ is in fact a two-layer neural network with $N$ neurons. Instead of fine-tuning $\gamma$ and stepsize $h$ in \UNLA{}, we directly fine-tune the value of $\varphi_0$, $\varphi_1$ and $\varphi_2$ in Algorithm~\ref{UNLA} by grid search. For simplifying the computation, we approximate $(\mathrm{B}_k^i)^x$ and $(\mathrm{B}_k^i)^v$ by $\eta\xi^x_k$ and $\eta\xi^v_k$ where $\xi^x_k$ and $\xi^v_k$ are independent standard Gaussian, and then we fine-tune the scaling scalar $\eta$. We compare our method (\UNLA{}) to~\ref{NLA} with stepsize $h_1$ and scaling scalar $\lambda_1$ given by,
\begin{equation*}
\label{NLA}
    x^j_{k+1}=x^j_{k} - h_1{D}_{\rho}F(\mu_{\textbf{x}_k},x^j_k) + \sqrt{2\lambda_1h_1}\xi^i_k
    \tag{{N-LA}}
\end{equation*}
for $i=1,...,N$, $k=1,...,K$ and $\xi^i_k\sim\mathcal{N}(0,I_d)$, and EM-UNLA (the EM discretization of the~\ref{Par-UMFLD} with stepsize $h_2$ and scaling scalar $\lambda_2$) whose update is given by \looseness=-1
\begin{equation*}
\label{EM-UNLA}
    \begin{aligned}
        x_{k+1}^j&=x_k^j + h_2v_k^j\\
        v_{k+1}^j&=(1-\gamma h_2)v_k^j - h_2{D}_{\rho}F(\mu_{\textbf{x}_k},x_k^j) + \sqrt{2\lambda_2h_2}\xi^i_k
    \end{aligned}
    \tag{{EM-N-ULA}}
\end{equation*}
for $i=1,...,N$, $k=1,...,K$ and $\xi^i_k\sim\mathcal{N}(0,I_d)$
in the same task. We choose $K=10^4$ and also fine-tune $h_1,\,\lambda_1$ and $h_2,\,\lambda_2$ to make fair comparison. We postpone our choice of hyperparameters to the Appendix \ref{Experimental Settings}. For each algorithm in our experiment, we initialize $x_0^j\sim\mathcal{N}(0,10^{-2}I_d)$ and $v_0^j\sim\mathcal{N}(0,10^{-2}I_d)$ for $j=1,...,N$, average 5 runs over random seeds in $\{0,1,2,3,4\}$ and generate the error bars by filling between the largest and the smallest value per iteration.
\cref{outcomp} illustrates the effectiveness of \UNLA{}. For each $N$, \UNLA{} enjoys faster convergence than \NLA{} and \ref{EM-UNLA}. Notably, there is an interesting phenomenon in our experiments. For $N=256$, both \UNLA{} and \EMUNLA{} suffer from convergence instability, which means that the loss will escape 
the stable convergence regime and slightly go up after many training epochs. However, \UNLA{} outperforms \NLA{} and \EMUNLA{} without convergence instability for $N=512,\,1024,\,2048$, and the loss of \UNLA{} even goes on decreasing when the losses of \NLA{} and \EMUNLA{} keep stable for $N=1024,\,2048$. This phenomenon matches our theory that we do not reduce the number of particles for \UNLA{} when compared with \NLA{} (see Table \ref{summary}). These observations suggest that our method performs better in the high particle-approximation regime. 
\cref{incomp} demonstrates this finding more transparently. The second row of \cref{outcomp} also suggests that EM discretization incurs a larger bias than EI.
\begin{figure}[t]
\centering
\subfigure{
\includegraphics[width=8.0cm,height = 6.8cm]{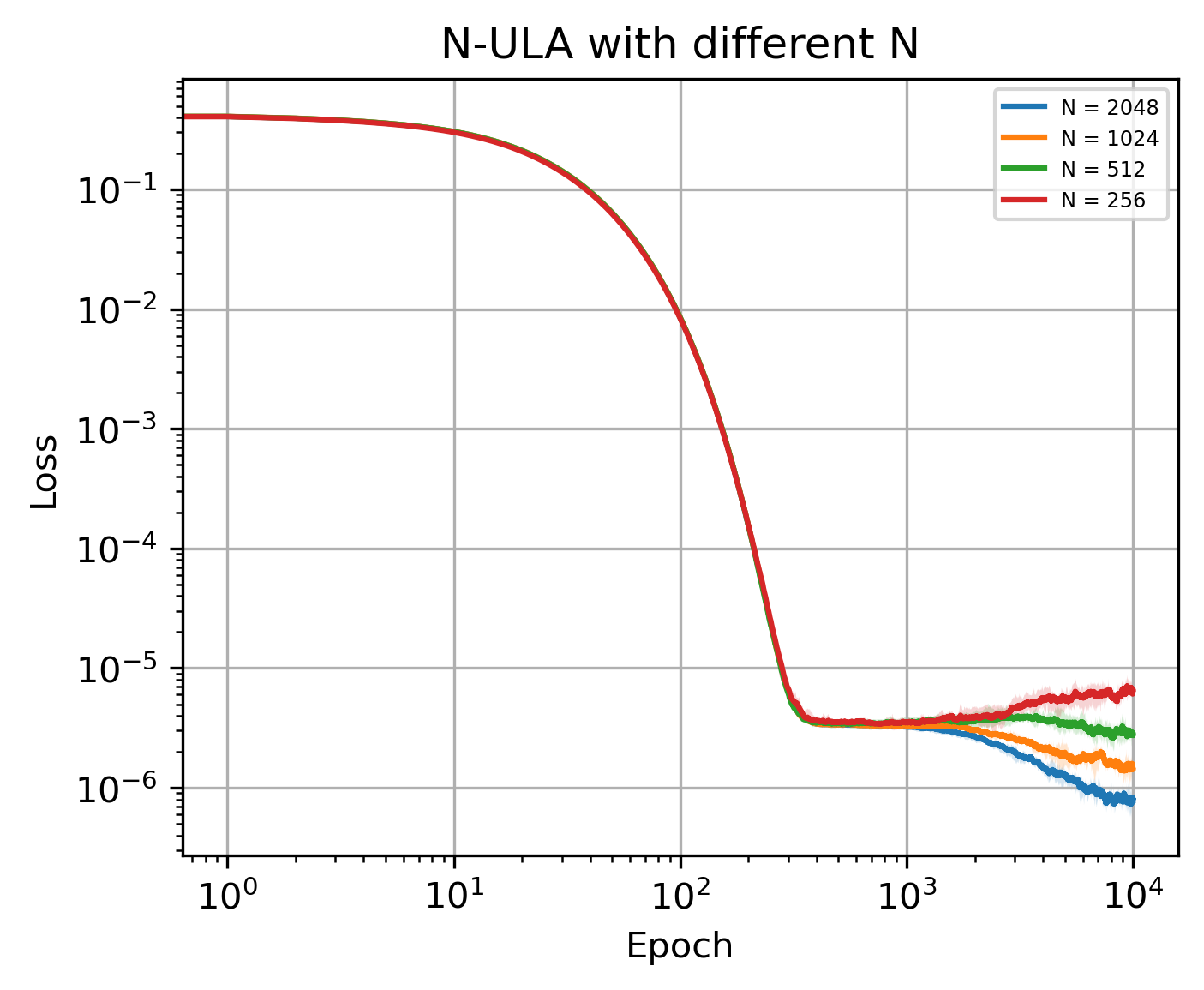}}
\caption{NULA with different number of particles}
\label{incomp}
\end{figure}
\section{Discussion}
To summarize, this paper (1) improves the convergence guarantees in~\citet{chen2023uniform} with a refined Lyapunov analysis (Theorems~\ref{convergence rate of the UMFLD} and~\ref{convergence of particle dynamics}); (2) discretizes
the~\ref{UMFLD} and~\ref{Par-UMFLD} with a scheme which results in smaller bias than the EM scheme; and (3) presents a novel discretization analysis of~\UMLD{} and~\UNLD. We also verify that these methods work when the objective is $W_1$ smooth. We now note several directions for future potential developments. First, it is unclear what the optimal choice of damping coefficient $\gamma$ is for \UMLD{} and \UNLD{}. Understanding whether the optimal choice has been found is of interest. Second, we obtain convergence rates for the \UMLA{} and \UNLA{} in TV distance, which are not consistent with the convergence rates of \UMLD{}, \UNLD{}, \MLA{} and \NLA{} in energy gap (e.g. $\mathcal{F}(\mu_t)-\mathcal{F}(\mu_*)$).
We hope to establish our results in the energy gap or KL divergence in the future. What's more, our techniques on uniformly bounding the iterates of~\UMLD{} and~\UNLD{} combined with Assumptions~\ref{second moment bound}-\ref{initialization and invariant for N-particle} generates an additional $\mathscr{C}_{\textsf{LSI}}$ after using Talagrand's inequality, which leads to non-improvement of $\mathscr{C}_{\textsf{LSI}}$ for \UMLA{} and \UNLA{}. We hope to explore whether it is possible to weaken those assumptions and refine the analysis of uniformly bounding the iterates to improve the dependence of $\mathscr{C}_{\textsf{LSI}}$ in the mixing time and number of particles of \UMLA{} and \UNLA{}.
\bibliography{ref}
\newpage
\appendix
\section{Supplementary background}
\subsection{Mean-field Langevin dynamics}
\label{MLD background}
The law $(\rho_t)_{t\geq 0}$ of \ref{MFLD} solves the following non-linear Fokker-Planck equation:
\begin{equation}
\begin{aligned}
    \frac{\partial\rho_t}{\partial t}&=\nabla\cdot(\rho_tD_{\rho}F(\rho_t,\cdot))+\Delta\rho_t=\nabla\cdot\left(\rho_t\nabla\log\frac{\rho_t}{\hat{\rho}_t}\right),
    \end{aligned}
\end{equation}
where $\hat{\rho}_t(x)\propto\exp\left(-\frac{\delta F}{\delta\rho}(\rho_t,x)\right)$.
Let $E(\rho)\defeq F(\rho)+\text{Ent}(\rho)$. The optimality condition of the EMO problem is
\begin{equation}
\label{opt condition for MLD}
    \frac{\delta E}{\delta\rho}=\frac{\delta F}{\delta\rho}+\log\rho+c=0,
\end{equation}
where $c$ is a constant. Given the condition \eqref{opt condition for MLD}, the solution of EMO problem $\rho_*$ satisfies
$\rho_*(x)=\hat{\rho}_*(x)\propto\exp\left(-\frac{\delta F}{\delta\rho}(\rho_*,x)\right),$
which solves $\nabla\cdot\left(\rho_t\nabla\log\frac{\rho_t}{\hat{\rho}_t}\right)=0$. Thus we conclude that \ref{MFLD} converges to the minimizer of EMO objective.

\subsection{N-particle Langevin dynamics}
The space-discretization of \ref{MFLD} is referred to as the \emph{N-particle Langevin dynamics},
\begin{equation}
    \label{NLD}
    \mathrm{d}x_t^i=-D_{\rho}F(\rho_{\textbf{x}_t},x_t^i)\mathrm{d}t+\sqrt{2}\mathrm{dB}_t,
    \tag{{N-LD}}
\end{equation}
where $\rho_{\textbf{x}_t}=\frac{1}{N}\sum_{i=1}^N\delta_{x^i_t}$. Let $\rho_t^i$ denotes the law of $x_t^i$ and $\rho_t^N$ denotes the joint law of $\textbf{x}_t\defeq(x_t^1,...,x_t^N)$. The joint law $(\rho_t^N)_{t\geq 0}$ of \ref{NLD} solves the following linear Fokker-Planck equation:
\begin{equation}
\label{f-p of NLD}
    \frac{\partial\rho_t^N}{\partial t}=\sum_{i=1}^N\nabla_i\cdot\left(\rho_t^ND_{\rho}F(\rho_{\textbf{x}_t},x_t^i)\right)+\Delta_i\rho_t^N=\sum_{i=1}^N\nabla_i\cdot\left(\rho_t^N\nabla_i\log\frac{\rho_t^N}{\rho_*^N}\right),
\end{equation}
where $\nabla_i\defeq\nabla_{x^i}$, $\Delta_i\defeq\Delta_{x^i}$ and $\rho_*^N(\textbf{x})\propto\exp(-NF(\rho_{\textbf{x}}))$. Define the \emph{N-particle free energy}:
\begin{equation}
\label{N-par free energy}
    E^N(\rho^N)=N\int F(\rho_{\textbf{x}})\rho^N(\mathrm{d}x)+\text{Ent}(\rho^N).
\end{equation}
The optimality condition of minimizing the N-particle free energy \eqref{N-par free energy} over $\mathcal{P}_2(\mathbb{R}^{Nd})$ is 
\begin{equation}
\label{opt condition for NLD}
    \frac{\delta E^N}{\delta\rho^N}=NF(\rho_{\textbf{x}})+\log\rho^N+c=0,
\end{equation}
where $c$ is a constant. Given the optimality condition \eqref{opt condition for NLD}, the minimizer of \eqref{N-par free energy} satisfies $\rho_*^N(x)\propto\exp(-NF(\rho_{\textbf{x}}))$, which is exactly the limiting distribution of \ref{NLD} according to \eqref{f-p of NLD}. Thus we conclude that \ref{NLD} converges to the minimizer of \eqref{N-par free energy}. 

\subsection{Mean-field underdamped Langevin dynamics}
\label{MULD background}
The law $(\mu_t)_{t\geq 0}$ of \ref{UMFLD} solves the following non-linear Fokker-Planck equation:
\begin{equation}
\label{f-p MULD}
    \begin{aligned}
        \frac{\partial\mu_t}{\partial t}&=\gamma\Delta_v\mu_t+\gamma\nabla_v\cdot(\mu_t v_t)-v\cdot\nabla_x\mu_t+D_{\rho}F(\mu_t^x,x_t)\cdot\nabla_v\mu_t\\
    &=\nabla\cdot\left(\mu_tJ_{\gamma}\nabla\log\frac{\mu_t}{\hat{\mu}_t}\right),
    \end{aligned}
\end{equation}
where $J_{\gamma}=\left(\begin{array}{cc}
    0 & 1 \\
    -1 & \gamma
\end{array}\right)$, $\nabla\defeq(\nabla_x,\nabla_v)^{\mathsf{T}}$ and $\hat{\mu}_t(x,v)\propto\exp\left(-\frac{\delta F}{\delta\rho}(\mu_t^X,x)-\frac{1}{2}\|v\|^2\right)$. The optimality condition of the augmented EMO problem is
\begin{equation}
\label{opt condition for MULD}
    \frac{\delta\mathcal{F}}{\delta\mu}=\frac{\delta F}{\delta\mu}+\log\mu+\frac{1}{2}\|v\|^2+c=0,
\end{equation}
where $\mathcal{F}$ is defined in \eqref{eq:Lyap1} and $c$ is a constant. Note that $\frac{\delta F(\mu^X)}{\delta\mu}=\frac{\delta F(\mu^X)}{\delta\rho}.$ Given the optimality condition \eqref{opt condition for MULD}, the solution of the augmented EMO problem satisfies $\mu_*(x,v)=\hat{\mu}_*(x,v)\propto\exp\left(-\frac{\delta F}{\delta\rho}(\mu_*^X,x)-\frac{1}{2}\|v\|^2\right)$, which solves $\nabla\cdot\left(\mu_tJ_{\gamma}\nabla\log\frac{\mu_t}{\hat{\mu}_t}\right)=0$. Thus we conclude that \ref{UMFLD} converges to the minimizer of the augmented EMO objective.
\subsection{N-particle underdamped Langevin dynamics}
\label{NULD background}
The law $(\mu_t^N)_{t\geq 0}$ of \ref{Par-UMFLD} solves the following linear Fokker-Planck equation:
\begin{equation}
\label{f-p of NULD}
    \begin{aligned}
        \frac{\partial\mu_t^N}{\partial t}&=\sum_{i=1}^N\left(\gamma\Delta_{v^i}\mu_t^N+\gamma\nabla_{v^i}\cdot(\mu_t^Nv_t^i)-v^i_t\cdot\nabla_{x^i}\mu_t^N+D_{\rho}F(\mu_{\textbf{x}_t},x_t^i)\cdot\nabla_{v^i}\mu_t^N\right)\\
        &=\sum_{i=1}^N\nabla_i\cdot\left(\mu^N_tJ_{\gamma}\nabla_i\log\frac{\mu^N_t}{\hat{\mu}^N_*}\right),
    \end{aligned}
\end{equation}
where $J_{\gamma}=\left(\begin{array}{cc}
    0 & 1 \\
    -1 & \gamma
\end{array}\right)$, $\nabla_i\defeq(\nabla_{x^i},\nabla_{v^i})^{\mathsf{T}}$ and $\hat{\mu}^N_*(x,v)\propto\exp\left(-NF(\mu_{\textbf{x}})-\frac{1}{2}\|v\|^2\right)$. Define the \emph{N-particle free energy}:
\begin{equation}
\label{N-par free energy 2}
    \mathcal{F}^N(\mu^N)=\int NF(\mu_{\textbf{x}})+\frac{1}{2}\|\textbf{v}\|^2 \mu^N(\mathrm{d}\textbf{x}\mathrm{d}\textbf{v})+\text{Ent}(\mu^N).
\end{equation}
The optimality condition of minimizing the N-particle free energy \eqref{N-par free energy 2} over $\mathcal{P}_2(\mathbb{R}^{2Nd})$ is 
\begin{equation}
\label{opt condition for NULD}
    \frac{\delta \mathcal{F}^N}{\delta\mu^N}=NF(\mu_{\textbf{x}})+\frac{1}{2}\|\textbf{v}\|^2+\log\mu^N+c=0,
\end{equation}
where $c$ is a constant. Given the optimality condition \eqref{opt condition for NULD}, the minimizer of \eqref{N-par free energy 2} satisfies $\mu_*^N(x)\propto\exp(-NF(\mu_{\textbf{x}})-\frac{1}{2}\|\textbf{v}\|^2)$, which is exactly the limiting distribution of \ref{Par-UMFLD} according to \eqref{f-p of NULD}. Thus we conclude that \ref{Par-UMFLD} converges to the minimizer of \eqref{N-par free energy 2}. 
\section{Helpful lemmas}
\begin{lemma}
\label{solution to NULA}
    The solution $(x_t,v_t)$ to the discrete-time process~\eqref{dis-UMFLD} for $t\in[kh,(k+1)h]$ is
    \begin{equation}
    \label{update of UMLA}
        \begin{aligned}
            x_{t}&=x_{kh}+\frac{1-e^{-\gamma(t-kh)}}{\gamma}v_{kh}-\frac{\gamma h-(1-e^{-\gamma(t-kh)})}{\gamma^2}D_{\rho}F(\mu_{kh}^X,x_{kh})+\mathrm{B}_{kh}^x,\\
            v_t&=e^{-\gamma(t-kh)}v_{kh}-\frac{1-e^{-\gamma(t-kh)}}{\gamma}D_{\rho}F(\mu_{kh}^X,x_{kh})+\mathrm{B}_{kh}^v,
        \end{aligned}
    \end{equation}
    where $(\mathrm{B}_{kh}^x,\mathrm{B}_{kh}^v)\in\mathbb{R}^{2d}$ is independent of $k$ and has the joint distribution
    \begin{equation*}
        \left[\begin{array}{c}
             \mathrm{B}_{kh}^x\\
             \mathrm{B}_{kh}^v 
        \end{array}\right]\sim\mathcal{N}\left(0, \left[\begin{array}{cc}
            \frac{2}{\gamma}\left(h-\frac{2(1-e^{-\gamma(t-kh)})}{\gamma}+\frac{1-e^{-2\gamma(t-kh)}}{2\gamma}\right) &  \frac{1}{\gamma}\left(1-2e^{-\gamma(t-kh)}+e^{-2\gamma(t-kh)}\right)\\
             * & 1-e^{-2\gamma(t-kh)}
        \end{array}\right]\otimes I_d\right)
    \end{equation*}
    The solution $(x_t^i,v_t^i)$ to the discrete-time process~\eqref{dis-par-UMFLD} for $i=1,...,N$ and $t\in[kh,(k+1)h]$ is
    \begin{equation}
    \label{update of UNLA}
        \begin{aligned}
            x^i_{t}&=x^i_{kh}+\frac{1-e^{-\gamma(t-kh)}}{\gamma}v^i_{kh}-\frac{\gamma h-(1-e^{-\gamma(t-kh)})}{\gamma^2}D_{\rho}F(\mu_{\textbf{x}_{kh}},x^i_{kh})+(\mathrm{B}^i_{kh})^x,\\
            v^i_t&=e^{-\gamma(t-kh)}v^i_{kh}-\frac{1-e^{-\gamma(t-kh)}}{\gamma}D_{\rho}F(\mu_{\textbf{x}_{kh}},x^i_{kh})+(\mathrm{B}^i_{kh})^v.
        \end{aligned}
    \end{equation}
    where $((\mathrm{B}^i_{kh})^x,(\mathrm{B}^i_{kh})^v)\in\mathbb{R}^{2d}$ is independent of $i,\,k$ and has the joint distribution
    \begin{equation*}
        \left[\begin{array}{c}
             (\mathrm{B}_{kh}^i)^x\\
             (\mathrm{B}_{kh}^i)^v 
        \end{array}\right]\sim\mathcal{N}\left(0, \left[\begin{array}{cc}
            \frac{2}{\gamma}\left(h-\frac{2(1-e^{-\gamma(t-kh)})}{\gamma}+\frac{1-e^{-2\gamma(t-kh)}}{2\gamma}\right)I_d &  \frac{1}{\gamma}\left(1-2e^{-\gamma(t-kh)}+e^{-2\gamma(t-kh)}\right)I_d\\
             \frac{1}{\gamma}\left(1-2e^{-\gamma(t-kh)}+e^{-2\gamma(t-kh)}\right)I_d & 1-e^{-2\gamma(t-kh)}I_d
        \end{array}\right]\right)
    \end{equation*}
\end{lemma}
\begin{proof}
The proof technique is similar to the proof of Lemmas 10 and 11 proposed in \citet{cheng2018underdamped}.
\end{proof}
Choosing $t=(k+1)h$ for \eqref{update of UNLA} generates the update parameters of Algorithm \ref{UNLA}:
\begin{equation}
\label{choice of varphi}
    \varphi_0=\frac{1-e^{-\gamma h}}{\gamma},\ \varphi_1=\frac{\gamma h-(1-e^{-\gamma h})}{\gamma^2},\ \varphi_2=e^{-\gamma h};
\end{equation}
\begin{equation}
\label{choice of Sigma}
    \Sigma_{11}=\frac{2}{\gamma}\left(h-\frac{2(1-e^{-\gamma h})}{\gamma}+\frac{1-e^{-2\gamma h}}{2\gamma}\right), \ \Sigma_{12}=\frac{1}{\gamma}\left(1-2e^{-\gamma h}+e^{-2\gamma h}\right), \ \Sigma_{22}=1-e^{-2\gamma h}.
\end{equation}
\begin{lemma}
\label{W1 smoothness}
    Suppose $D_{\rho}F:\mathcal{P}_2(\mathbb{R}^d)\times\mathbb{R}^d\rightarrow\mathbb{R}^d$ admits a continuous first variation ${\delta D_{\rho}F}:\mathcal{P}_2(\mathbb{R}^d)\times\mathbb{R}^d\rightarrow\mathbb{R}^d$. Then, $D_{\rho}F$ is $\mathscr{L}$-Lipschitz with respect to $W_1$ distance satisfying
    \begin{equation}
    \label{L-smoothness w1}
        \|D_{\rho}F(\rho_1,x)-D_{\rho}F(\rho_2,x)\|\leq\mathscr{L}W_1(\rho_1,\rho_2)
    \end{equation}
    with $
\mathscr{L}\defeq\sup_{\rho'\in\mathcal{P}_2(\mathbb{R}^d)}\sup_{x,x'\in\mathbb{R}^d}\left\|{D^2_{\rho}F}(\rho',x,x')\right\|_{\textsf{\em op}}
    $
\end{lemma}
\begin{proof}
    By the definition of functional derivative, we have
    \begin{equation}
    \label{proof of lemma2}
        \|D_{\rho}F(\rho_1,x)-D_{\rho}F(\rho_2,x)\|\leq\int_0^1\left\|\int\frac{\delta}{\delta\rho} D_{\rho}F((1-t)\rho_1+t\rho_2,x,x')(\rho_1-\rho_2)\mathrm{d}x'\right\|\mathrm{d}t
    \end{equation}
By Kantorovich duality and the definition of $\mathscr{L}$, which is the Liptschiz constant of $\frac{\delta}{\delta\rho} D_{\rho}F(\cdot,x)$, we obtain
    \begin{equation*}
        \left\|\int\frac{\delta}{\delta\rho} D_{\rho}F((1-t)\rho_1+t\rho_2,x,x')(\rho_1-\rho_2)\mathrm{d}x'\right\|\leq \mathscr{L}W_1(\rho_1,\rho_2).
    \end{equation*}
    Combining with \eqref{proof of lemma2}, we complete the proof.
\end{proof}
\begin{lemma}[Mean-field Entropy Sandwich, \citealt{chen2023uniform}, Lemma 4.2]
\label{Mean-field Sandwich}
    Assume $F$ satisfies Assumptions~\ref{convexity}-\ref{LSI of PGD}. Then for every $\mu\in\mathcal{P}_2(\mathbb{R}^{2d})$ we have
    \begin{equation}
        \textsf{\em KL}(\mu\|\mu_*)\leq\mathcal{F}(\mu)-\mathcal{F}(\mu_*)\leq\textsf{\em KL}(\mu\|\hat{\mu})\leq\left(1+\frac{\mathscr{L}}{\mathscr{C}_{\textsf{\em LSI}}}+\frac{\mathscr{L}^2}{2\mathscr{C}_{\textsf{\em LSI}}^2}\right)\textsf{\em KL}(\mu\|\mu_*).
    \end{equation}
\end{lemma}
\begin{lemma}[Particle System's Entropy Inequality, \citealt{chen2023uniform}, Lemma 4.2]
\label{Particle System's Entropy Inequality}
    Assume that $F$ satisfies Assumption \ref{convexity} and there exists a measure $\mu_*\in\mathcal{P}(\mathbb{R}^{2d})$ that admits the proximal Gibbs distribution $\mu_*(x,v)\propto\exp\left(-\frac{\delta F}{\delta\mu}(\mu^x_*,x)-\frac{1}{2}\|v\|^2\right)$. Then for all $\mu^N\in\mathcal{P}(\mathbb{R}^{2dN})$, we have
    \begin{equation}
        \textsf{\em KL}(\mu^N\|\mu_*^{\otimes N})\leq\mathcal{F}^N(\mu^N)-N\mathcal{F}(\mu_*).
    \end{equation}
\end{lemma}
\begin{lemma}[Information Inequality]
\label{Information Inequality}
     Let $X_1,...,X_N$ be measurable spaces, $\mu$ be a probability on the product space $X=X_1\times...\times X_N$ with $\mu=\mu^1\otimes...\otimes\mu^N$ and $\nu=\nu^1\otimes...\otimes\nu^N$ is a $\sigma$-finite measure. Then
    \begin{equation}
        \sum_{i=1}^N\textsf{\em KL}(\mu^i\|\nu^i)\leq\textsf{\em KL}(\mu\|\nu).
    \end{equation}
\end{lemma}
\begin{lemma}[Matrix Gr{\"o}nwall's Inequality, \citealt{zhang2023improved}]
\label{matrixGron}
    Let $x:\mathbb{R}_+\rightarrow\mathbb{R}^d$, and $c\in\mathbb{R}^d$, $A\in\mathbb{R}^{d\times d}$, where $A$ has non-negative entries. Suppose that the following inequality is satisfied componentwise:
    \begin{equation*}
        x(t)\leq c+\int_0^t Ax(s)\mathrm{d}s,\quad\text{for all }t\geq 0.
    \end{equation*}
    Then the following inequality holds where $I_d\in\mathbb{R}^{d\times d}$ is the $d$-dimensional identity matrix:
    \begin{equation*}
        x(t)\leq\left(AA^{\dag}e^{At}-AA^{\dag}+I_d\right)c.
    \end{equation*}
\end{lemma}
\begin{lemma}
\label{bound of the iterate difference}
    Let $(x_t,v_t)_{t\geq 0}$ and $(x_t^i,v_t^i)_{t\geq 0}$ respectively denote the iterates of the {\em\ref{UMFLD}} and {\em\ref{Par-UMFLD}}. Assume that $h\lesssim \mathscr{L}^{-1/2}\wedge \gamma^{-1}$. Under Assumption \ref{smoothness} and Assumption \ref{bounded grad}, for $t\in[kh,(k+1)h]$, we have
    \begin{equation*}
         \sup_{t\in[kh,(k+1)h]}\|x_t-x_{kh}\|\leq
2\mathscr{L}h^2\|x_{kh}\|+4h\|v_{kh}\|+2\mathscr{L}h^2+2\sqrt{2\gamma}h\sup_{t\in[kh,(k+1)h]}\|\mathrm{B}_t-\mathrm{B}_{kh}\|
    \end{equation*}
    \begin{equation*}
         \sup_{t\in[kh,(k+1)h]}\|x^i_t-x^i_{kh}\|\leq
2\mathscr{L}h^2\|x^i_{kh}\|+4h\|v^i_{kh}\|+2\mathscr{L}h^2+2\sqrt{2\gamma}h\sup_{t\in[kh,(k+1)h]}\|\mathrm{B}^i_t-\mathrm{B}^i_{kh}\|
    \end{equation*}
    for $i=1,...,N$.
\end{lemma}
\begin{proof}
We only prove the first relation, and the proof of the second relation is similar.
\begin{equation*}
\begin{aligned}
    \|x_t-x_{kh}\|&=\left\|\int_{kh}^t v_{\tau}\mathrm{d}\tau\right\|\leq h\|v_{kh}\|+\left\|\int_{kh}^t v_{\tau}-v_{kh}\mathrm{d}\tau\right\|\\
    &\leq h\|v_{kh}\|+\left\|\int_{kh}^t\int_0^{\tau}\gamma v_{\tau'}\mathrm{d}\tau'\mathrm{d}\tau\right\|+\left\|\int_{kh}^t\int_{kh}^{\tau}{D}_{\rho}F(\mu_{\tau'}^X,x_{\tau'})\mathrm{d}\tau'\mathrm{d}\tau\right\|+\left\|\int_{kh}^t\int_{kh}^{\tau}\sqrt{2\gamma}\mathrm{dB}_{\tau'}\mathrm{d}\tau\right\|\\
    &\leq h\|v_{kh}\|+\gamma h\left(h\|v_{kh}\|+\int_{kh}^t\|v_{\tau}-v_{kh}\|\mathrm{d}\tau\right)+\left\|\int_{kh}^t\int_{kh}^{\tau}{D}_{\rho}F(\mu_{\tau'}^X,x_{\tau'})\mathrm{d}\tau'\mathrm{d}\tau\right\|\\
    &\quad +\left\|\int_{kh}^t\int_{kh}^{\tau}\sqrt{2\gamma}\mathrm{dB}_{\tau'}\mathrm{d}\tau\right\|\\
    &\leq h\|v_{kh}\|+\gamma h\left(h\|v_{kh}\|+\int_{kh}^t\|v_{\tau}-v_{kh}\|\mathrm{d}\tau\right)+\mathscr{L}h\int_{kh}^t\left\|x_{\tau}-x_{kh}\right\|\mathrm{d}\tau+\mathscr{L}h^2\|x_{kh}\|\\
    &\quad+\mathscr{L}h^2+\sqrt{2\gamma}h\sup_{t\in[kh,(k+1)h]}\|\mathrm{B}_t-\mathrm{B}_{kh}\|
\end{aligned}
\end{equation*}
where the last inequality follows from Assumptions \ref{smoothness} and \ref{bounded grad}. Likewise for $V$:
\begin{equation*}
    \begin{aligned}
        \|v_t-v_{kh}\|&=\left\|\int_{kh}^t\gamma v_{\tau}\mathrm{d}\tau\right\|+\left\|\int_{kh}^t{D}_{\rho}F(\mu_{\tau}^X,x_{\tau})\mathrm{d}\tau\right\|+\left\|\int_{kh}^t\sqrt{2\gamma}\mathrm{dB}_t\right\|\\
        &\leq\gamma\left(h\|v_{kh}\|+\int_{kh}^t\|v_{\tau}-v_{kh}\|\mathrm{d}\tau\right)+\left\|\int_{kh}^t{D}_{\rho}F(\mu_{\tau}^X,x_{\tau})\mathrm{d}\tau\right\|+\sqrt{2\gamma}\sup_{t\in[kh,(k+1)h]}\|\mathrm{B}_t-\mathrm{B}_{kh}\|\\
        &\leq\gamma\left(h\|v_{kh}\|+\int_{kh}^t\|v_{\tau}-v_{kh}\|\mathrm{d}\tau\right)+\mathscr{L}\int_{kh}^t\|x_{\tau}-x_{kh}\|\mathrm{d}\tau+\mathscr{L}h+\mathscr{L}h\|x_{kh}\|\\
        &\quad+\sqrt{2\gamma}\sup_{t\in[kh,(k+1)h]}\|\mathrm{B}_t-\mathrm{B}_{kh}\|
    \end{aligned}
\end{equation*}
where the last inequality follows from Assumptions \ref{smoothness} and \ref{bounded grad}. Before applying matrix form of Gr{\"o}nwall's inequality, let $c=c_1+c_2$ with $c_2=\left[\begin{array}{c}
         h\|v_{kh}\|  \\
         0
    \end{array}\right]$,
\begin{equation*}
    A=\left[\begin{array}{cc}
        \mathscr{L}h & \gamma h \\
        \mathscr{L} & \gamma
    \end{array}\right],\,c_1=\left[\begin{array}{c}
        \mathscr{L}h^2\|x_{kh}\|+\gamma h^2\|v_{kh}\|+\mathscr{L}h^2+\sqrt{2\gamma}h\sup_{t\in[kh,(k+1)h]}\|\mathrm{B}_t-\mathrm{B}_{kh}\|\\
        \mathscr{L}h\|x_{kh}\|+\gamma h\|v_{kh}\|+\mathscr{L}h+\sqrt{2\gamma}\sup_{t\in[kh,(k+1)h]}\|\mathrm{B}_t-\mathrm{B}_{kh}\|
    \end{array}\right].
\end{equation*}
$c_1$ lies in the image space of $A$, and $\exp(A_t)c_1$ also lies in the image space of $A$. For the first component:
\begin{equation*}
\begin{aligned}
    \sup_{t\in[kh,(k+1)h]}\|x_t-x_{kh}\|&\leq h\exp\left((\mathscr{L}h+\gamma)h\right)(\mathscr{L}h\|x_{kh}\|+\gamma h\|v_{kh}\|+\mathscr{L}h+\sqrt{2\gamma}\sup_{t\in[kh,(k+1)h]}\|\mathrm{B}_t-\mathrm{B}_{kh}\|)\\
    &\quad+\frac{\mathscr{L}h\exp((\mathscr{L}h+\gamma)h)+\gamma}{\mathscr{L}h+\gamma}h\|v_{kh}\|\\
    &\leq 2h\left(\mathscr{L}h\|x_{kh}\|+2\|v_{kh}\|+\mathscr{L}h+\sqrt{2\gamma}\sup_{t\in[kh,(k+1)h]}\|\mathrm{B}_t-\mathrm{B}_{kh}\|\right)
\end{aligned}
\end{equation*}
where the second inequality comes from choosing $h\lesssim\frac{1}{\mathscr{L}^{1/2}}\land\frac{1}{\gamma}$.
\begin{equation*}
    ((AA^{\dagger}(\exp(Ah)-I)+I)c_2)_{(1)}=\frac{\mathscr{L}h\exp((\mathscr{L}h+\gamma)h)+\gamma}{\mathscr{L}h+\gamma}h\|v_{kh}\|\leq 2h\|v_{kh}\|
\end{equation*}
Combining relations above and  Lemma \ref{matrixGron} completes the proof.
\end{proof}
\begin{lemma}
\label{moment bound mean-field}
    Let $(x_t,v_t)_{t\geq 0}$ denote the iterates of the {\em\ref{UMFLD}} with $(x_0,v_0)\sim\mu_0=\mathcal{N}(0,I_{2d})$. Under Assumption \ref{second moment bound} and Assumption \ref{initialization and invariant}, we have
    \begin{equation}
        \mathbb{E}\|(x_t,v_t)\|^2\lesssim\frac{\mathscr{L}d}{\mathscr{C}_{\textsf{\em LSI}}}
    \end{equation}
\end{lemma}
\begin{proof}
    \begin{equation*}
    \begin{aligned}
        \mathbb{E}\|(x_t,v_t)\|^2=W_2^2(\mu_t,\delta_0)&\leq 2 W_2^2(\mu_t,\mu_*)+2W_2^2(\mu_*,\delta_0)\\
        &\leq\frac{2}{\mathscr{C}_{\textsf{LSI}}}\textsf{KL}(\mu_t\|\mu_*)+2\textbf{m}_2^2\\
        &\leq\frac{2}{\mathscr{C}_{\textsf{LSI}}}(\mathcal{F}(\mu_t)-\mathcal{F}(\mu_*))+2\textbf{m}_2^2\\
        &\leq\frac{2}{\mathscr{C}_{\textsf{LSI}}}(\mathcal{F}(\mu_0)-\mathcal{F}(\mu_*))+2\textbf{m}_2^2\\
        &\leq\frac{2}{\mathscr{C}_{\textsf{LSI}}}\mathcal{F}(\mu_0)+2\textbf{m}_2^2
    \end{aligned}
\end{equation*}
The second inequality follows from Talagrand's inequality which can be implied by Assumption \ref{LSI of PGD}.\footnote{Assumption \ref{LSI of PGD} states that the proximal Gibbs distribution satisfies the LSI. Note that $\mu_*$ also has the form of the proximal Gibbs distribution and thus satisfies LSI.} The third inequality follows from Lemma~\ref{Mean-field Sandwich}. The fourth inequality follows that $\frac{\mathrm{d}}{\mathrm{d}t}\mathcal{F}(\mu_t)<0$ along the \ref{UMFLD} (Proof of Theorem 2.1 in \citet{chen2023uniform}) and the last inequality follows from the assumption that $\mathcal{F}(\mu_*)\geq 0$. By the definition of $\mathcal{F}(\mu)$, we have $\mathcal{F}(\mu_0)=F(\mu^x_0)+\int\frac{1}{2}\|v\|^2\mu_0(\mathrm{d}x\mathrm{d}v)+\text{Ent}(\mu_0)$. Since $(x_0,v_0)\sim\mathcal{N}(0,I_{2d})$, we have $\int\frac{1}{2}\|v\|^2\mu_0(\mathrm{d}x\mathrm{d}v)\lesssim d$ and
\begin{equation*}
\begin{aligned}
    |\text{Ent}(\mu_0)|&=\left|\int\mu_0\log\mu_0\right|\\
    &=\frac{d}{2}\log(2\pi)+\frac{1}{2}\mathbb{E}_{\mu_0}\|\cdot\|^2\lesssim d.
    \end{aligned}
\end{equation*}
By Assumption \ref{initialization and invariant}, we have $F(\mu_0^x)\lesssim \mathscr{L}d$. By Assumption \ref{second moment bound}, we have $\textbf{m}_2^2\lesssim d$. Thus we have
\begin{equation*}
\begin{aligned}
    \mathbb{E}\|(x_t,v_t)\|^2&\leq\frac{2}{\mathscr{C}_{\textsf{LSI}}}\mathcal{F}(\mu_0)+2\textbf{m}_2^2\lesssim\frac{\mathscr{L}d}{\mathscr{C}_{\textsf{LSI}}}+d
    \end{aligned}
\end{equation*}
\end{proof}
\begin{lemma}
\label{moment bound par}
    Let ${(x^i_t,v^i_t)_{i=1}^N}$ denote the iterates of the {\em\ref{Par-UMFLD}} with $(x_0^i,v_0^i)\sim\mu_0^i=\mathcal{N}(0,I_{2d})$ for $i=1,...,N$ and $t\geq 0$. Under Assumption \ref{second moment bound} and Assumption \ref{initialization and invariant for N-particle}, we have
    \begin{equation}
        \frac{1}{N}\sum_{i=1}^N\mathbb{E}\|(x^i_t,v^i_t)\|^2\lesssim\frac{\mathscr{L}d}{\mathscr{C}_{\textsf{\em LSI}}}
    \end{equation}
\end{lemma}
\begin{proof}
    \begin{equation*}
\begin{aligned}
    \frac{1}{N}\sum_{i=1}^N\mathbb{E}\|(x^i_t,v^i_t)\|^2=\frac{1}{N}\sum_{i=1}^NW_2^2(\mu^i_t,\delta_0)
    &\leq \frac{2}{N}\sum_{i=1}^NW_2^2(\mu^i_t,\mu_*)+2W_2^2(\mu_*,\delta_0)\\
    &\leq\frac{2}{\mathscr{C}_{\textsf{LSI}}}\frac{1}{N}\sum_{i=1}^N\textsf{KL}(\mu^i_t\|\mu_*)+2\textbf{m}_2^2\\
        &\leq\frac{2}{\mathscr{C}_{\textsf{LSI}}}\frac{1}{N}\textsf{KL}(\mu^N_t\|\mu^{\otimes N}_*)+2\textbf{m}_2^2\\
        &\leq\frac{2}{\mathscr{C}_{\textsf{LSI}}}\left(\frac{1}{N}\mathcal{F}^N(\mu^N_t)-\mathcal{F}(\mu_*)\right)+2\textbf{m}_2^2\\
        &\leq\frac{2}{N\mathscr{C}_{\textsf{LSI}}}\mathcal{F}^N(\mu^N_0)+2\textbf{m}_2^2
\end{aligned}
\end{equation*}
The second inequality follows from Talagrand's inequality which can be implied by Assumption \ref{LSI of PGD}. The third inequality follows from Lemma \ref{Information Inequality}. The fourth inequality follows from Lemma~\ref{Particle System's Entropy Inequality} and the last inequality follows that $\frac{\mathrm{d}}{\mathrm{d}t}\mathcal{F}^N(\mu^N_t)<0$ along the~\ref{Par-UMFLD} (Proof of Theorem 2.2 in \citet{chen2023uniform}) and $\mathcal{F}(\mu_*)\geq 0$. By the definition of $\mathcal{F}^N(\mu^N)$, we have $\mathcal{F}^N(\mu_0^N)=\int (NF(\mu_{\textbf{x}})+\frac{1}{2}\|\textbf{v}\|^2)\mu_0^N(\mathrm{d}\textbf{x}\mathrm{d}\textbf{v})+\text{Ent}(\mu_0^N)$. Similar to the proof of Lemma~\ref{moment bound mean-field}, since $(\textbf{x},\textbf{v})\sim\mathcal{N}(0,I_{2Nd})$, we have $\int \frac{1}{2}\|\textbf{v}\|^2\mu_0^N(\mathrm{d}\textbf{x}\mathrm{d}\textbf{v})\lesssim Nd$ and $|\text{Ent}(\mu_0^N)|\lesssim Nd$. By Assumption \ref{initialization and invariant for N-particle} and Assumption \ref{second moment bound}, we also have $\int NF(\mu_{\textbf{x}})\mu_0^N(\mathrm{d}\textbf{x}\mathrm{d}\textbf{v})\lesssim N\mathscr{L}d$ and $\textbf{m}_2^2\lesssim d$.
Thus we have
\begin{equation*}
\begin{aligned}
    \frac{1}{N}\sum_{i=1}^N\mathbb{E}\|(x^i_t,v^i_t)\|^2&\leq\frac{2}{N\mathscr{C}_{\textsf{LSI}}}\mathcal{F}^N(\mu^N_0)+2\textbf{m}_2^2\\
    &=\frac{2}{N\mathscr{C}_{\textsf{LSI}}}\left(\int (NF(\mu_{\textbf{x}})+\frac{1}{2}\|\textbf{v}\|^2)\mu_0^N(\mathrm{d}\textbf{x}\mathrm{d}\textbf{v})+\text{Ent}(\mu_0^N)\right)+2\textbf{m}_2^2\\
    &\lesssim\frac{1}{N\mathscr{C}_{\textsf{LSI}}}(N\mathscr{L}d+Nd)+d\lesssim\frac{\mathscr{L}d}{\mathscr{C}_{\textsf{LSI}}}+d
    \end{aligned}
\end{equation*}
\end{proof}
\begin{lemma}[Girsanov's Theorem, (\citet{zhang2023improved}, Theorem 19)]
\label{Girsanov}
    Consider stochastic processes $(x_t)_{t\geq 0}$, $(b_t^{\textbf{P}})_{t\geq 0}$, $(b_t^{\textbf{Q}})_{t\geq 0}$ adapted to the same filtration, and $\sigma\in\mathbb{R}^{d\times d}$ any constant matrix (possibly degenerate). Let $\textbf{P}_T$ and $\textbf{Q}$ be probability measures on the path space $C([0,T];\mathbb{R}^d)$ such that $(x_t)_{t\geq 0}$ follows
    \begin{equation*}
        \begin{aligned}         \mathrm{d}x_t&=b_t^{\textbf{P}}\mathrm{d}t+\sigma \mathrm{dB}^{\textbf{P}}_t\quad\text{under }\textbf{P}_T,\\       \mathrm{d}x_t&=b_t^{\textbf{Q}}\mathrm{d}t+\sigma \mathrm{dB}^{\textbf{Q}}_t\quad\text{under }\textbf{Q}_T,
        \end{aligned}
    \end{equation*}
    where $\mathrm{B}^{\textbf{P}}$ and $\mathrm{B}^{\textbf{Q}}$ are $\textbf{P}_T$-Brownian motion and $\textbf{Q}_T$-Brownian motion. Suppose there exists a process $(y_t)_{t\geq 0}$ such that
    \begin{equation*}
        \sigma y_t=b_t^{\textbf{P}}-b_t^{\textbf{Q}},
    \end{equation*}
    and 
    \begin{equation*}
        \mathbb{E}_{\mathbf{Q}_T}\exp\left(\frac{1}{2}\int_0^T\left\|y_t\right\|^2\mathrm{d}t\right)< \infty.
    \end{equation*}
    If we define $\sigma^{\dag}$ as the Moore-Penrose pseudo-inverse of $\sigma$, then we have
    \begin{equation*}
        \frac{\mathrm{d}\mathbf{P}_T}{\mathrm{d}\mathbf{Q}_T}=\exp\left(\int_0^T\langle\sigma_t^{\dag}(b_t^{\textbf{P}_T}-b_t^{\textbf{Q}_T}),\mathrm{dB}^{\textbf{Q}_T}_t\rangle-\frac{1}{2}\int_0^T\|\sigma_t^{\dag}(b_t^{\textbf{P}_T}-b_t^{\textbf{Q}_T})\|^2\mathrm{d}t\right)
    \end{equation*}
    Besides, $(\tilde{\mathrm{B}}_t)_{t\in[0,T]}$ defined by $\mathrm{d}\Tilde{\mathrm{B}}_t\defeq \mathrm{dB}_t+\sigma_t^{\dag}(b_t^Y-b_t^X)$ is a $\mathbf{P}_T$-Brownian motion.
\end{lemma}

\section{Verification of assumptions}
\label{Verification of Assumption}
\subsection{Verification of Assumption~\ref{smoothness}}
\label{verification of smoothness}
Smoothness in $W_1$ distance has been verified for training mean-field neural networks in \citet{chen2022uniform}. Thus we only verify smoothness in $W_1$ distance for examples of density estimation via MMD minimization and KSD minimization. Lemma \ref{W1 smoothness} provides sufficient conditions for smoothness in $W_1$ distance. In particular, we have
\begin{equation}
\label{smooth tri}
    \|D_{\rho}F(\rho_1,x_1)-D_{\rho}F(\rho_2,x_2)\|\leq\|D_{\rho}F(\rho_1,x_1)-D_{\rho}F(\rho_2,x_1)\|+\|D_{\rho}F(\rho_2,x_1)-D_{\rho}F(\rho_2,x_2)\|
\end{equation}
\citet{suzuki2023convergence} verify that $\|D_{\rho}F(\rho_2,x_1)-D_{\rho}F(\rho_2,x_2)\|\leq\mathscr{L}\|x_1-x_2\|$ for three examples mentioned above. Thus it suffices to verify \eqref{L-smoothness w1} for the last two examples.
\paragraph{MMD minimization} We now prove that objective \eqref{MMD} satisfies Assumption~\ref{smoothness} with Gaussian RBF kernel. We choose $\sigma'$ in Gaussian RBF kernel $k$ to be $\sigma$ for brevity. We reformulate \eqref{MMD} as
\begin{equation}
\label{MMD assump}
    F(\rho)=\hat{\mathcal{M}}(\rho)+\frac{\lambda'}{2}\mathbb{E}_{x\sim\rho}\|x\|^2.
\end{equation}
According to the definition of $\hat{\mathcal{M}}$ in Section~\ref{sec:ex}, the intrinsic derivative of $F$ is
\begin{equation*}
\begin{aligned}
D_{\rho}F(\rho,x)&=D_{\rho}\hat{\mathcal{M}}(\rho,x)+\frac{\lambda'}{2}\|x\|^2\\
    &={2\iiint \nabla_x p(x;z)p(x';z')k(z,z')\mathrm{d}z\mathrm{d}z'\mathrm{d}\rho(x')}-{\frac{2}{n}\sum_{i=1}^n\int \nabla_x p(x;z)k(z,z_i)\mathrm{d}z}+\frac{\lambda'}{2}\|x\|^2
    \end{aligned}
\end{equation*}
We only need to prove $D_{\mu}\hat{\mathcal{M}}(\mu,x)$ is smooth. The second-order intrinsic derivative $D_{\rho}\hat{\mathcal{M}}(\rho,x)$ is
\begin{equation*}
\begin{aligned}
    {D^2_{\rho}\hat{\mathcal{M}}}(\rho,x,x')&=2\iint\nabla_x p(x;z)\otimes\nabla_{x'}p(x';z')k(z,z')\mathrm{d}z\mathrm{d}z'\\
    &=\frac{2}{(2\pi\sigma^2)^d\sigma^4}\iint(x-z)\otimes(x'-z')\exp\left(-\frac{\|x-z\|^2+\|x'-z'\|^2+\|z-z'\|^2}{2\sigma^2}\right)\mathrm{d}z\mathrm{d}z'
    \end{aligned}
\end{equation*}
From the relation $x\cdot\exp(-x^2/2\sigma^2)\leq\sigma$ for $x\geq 0$, we have
\begin{equation*}
\begin{aligned}
    \left\|{D^2_{\rho}\hat{\mathcal{M}}}(\rho,x,x')\right\|&\leq\frac{1}{(2\pi\sigma^2)^d\sigma^4}\iint\|x-z\|\|x'-z'\|\exp\left(-\frac{\|x-z\|^2+\|x'-z'\|^2+\|z-z'\|^2}{2\sigma^2}\right)\mathrm{d}z\mathrm{d}z'\\
    &\leq\frac{1}{(2\pi\sigma^2)^d\sigma^2}\iint\exp\left(-\frac{\|z-z'\|^2}{2\sigma^2}\right)\mathrm{d}z\mathrm{d}z'=\frac{1}{(2\pi\sigma^2)^{d/2}\sigma^2}
    \end{aligned}
\end{equation*}
According to Lemma~\ref{W1 smoothness} and \eqref{smooth tri}, $F$ defined in \eqref{MMD assump} satisfies Assumption~\ref{smoothness}. 
\paragraph{KSD minimization} We now prove that objective \eqref{KSD obj} satisfies Assumption~\ref{smoothness} with kernel 
\begin{equation}
\label{kernel}
k(x,x')=\exp\left(-\frac{\|x\|^2}{2\sigma_1^2}-\frac{\|x'\|^2}{2\sigma_1^2}-\frac{\|x-x'\|^2}{2\sigma_2^2}\right).
\end{equation}
We also assume the score function of $\mu_*$ satisfies \eqref{assump on score}. Under this assumption on score function and with this choice of kernel, \citet{suzuki2023convergence} show in their Appendix A that the Stein kernel $u_{\rho_*}$ satisfies $\sup_{x,x'\in\mathbb{R}^d}\max\{|u_{\rho_*}|,\|\nabla_xu_{\rho_*}\|,\|\nabla_x\nabla_{x'}u_{\rho_*}\|_{\textsf{op}}\}\leq \mathscr{L}$. We reformulate \eqref{KSD obj} as
\begin{equation}
\label{KSD assump}
    F(\rho)=\textsf{KSD}(\rho)+\frac{\lambda'}{2}\mathbb{E}_{x\sim\rho}\|x\|^2.
\end{equation}
Similarly, we only need to verify that $\textsf{KSD}$ is smooth with respect to $W_1$ distance. The intrinsic derivative of $\textsf{KSD}$ is
\begin{equation*}
    D_{\rho}\textsf{KSD}(\rho,x)=\int\nabla_x u_{\rho_*}(x,x')\mathrm{d}\rho(x').
\end{equation*}
The second-order intrinsic derivative of $D_{\rho}\textsf{KSD}(\rho,x)$ is
\begin{equation*}
     D^2_{\rho}\textsf{KSD}(\rho,x,x')=\nabla_x\nabla_{x'} u_{\rho_*}(x,x')
\end{equation*}
The following relation implies Assumption~\ref{smoothness} by Lemma \ref{W1 smoothness}.
\begin{equation*}
    \|D^2_{\rho}\textsf{KSD}(\rho,x,x')\|=\|\nabla_x\nabla_{x'} u_{\rho_*}(x,x')\|\leq \mathscr{L}
\end{equation*}
\subsection{Verification of Assumption~\ref{second moment bound}}
\label{verification of bounded second moment}
\paragraph{Training mean-field neural networks}
Denote $\hat{\mu}(x,v)=\hat{\mu}^X(x)\otimes\mathcal{N}(0,I_d)$ where $\hat{\mu}^X(x)\propto\exp\left(-\frac{\delta F}{\delta\rho}(\mu^X,x)\right)$. Since the second moment of $\mathcal{N}(0,I_d)$ is $O(d)$, it suffices to ensure $\mathbb{E}_{x\sim\hat{\mu}^X}\|x\|^2=O(d)$. We reformulate objective \eqref{empirical risk minimization} as:
\begin{equation}
    F(\rho)=\frac{1}{n}\sum_{i=1}^n\ell(h(\rho;a_i),b_i)+\frac{\lambda'}{2}\mathbb{E}_{x\sim \rho}[\|x\|^2].
\end{equation}
\begin{itemize}
    \item We will prove that Assumption \ref{second moment bound} holds if  $|h(x;a)|\leq\sqrt{\mathscr{L}}$ (such activation functions include \textsf{tanh} and \textsf{sigmoid}) and $|\partial_1 \ell|\leq\sqrt{\mathscr{L}}$ (such loss functions include logistic loss, Huber loss and log-cosh loss) or $\ell$ is quadratic. The functional derivative of $F$ is
    \begin{equation*}
    \begin{aligned}
        \frac{\delta F}{\delta \rho}(\mu^X,x)&=\frac{1}{n}\sum_{i=1}^n\left[\partial_1 \ell(h(\mu^X;a_i),b_i)h(x;a_i)\right]+\frac{\lambda'}{2}\|x\|^2
    \end{aligned}
    \end{equation*}
    Consider the case where $|\partial_1 \ell|\leq\sqrt{\mathscr{L}}$. Since  $|h(x;a)|\leq\sqrt{\mathscr{L}}$, we have $|\partial_1 \ell(h(\mu^X;a_i),b_i)h(x;a_i)|\leq \mathscr{L}$. Let $Z=\int \exp\left(-\frac{\delta F}{\delta \rho}(\mu^X,x)\right)\mathrm{d}x$, and we have
    \begin{equation}
    \label{second moment}
        \begin{aligned}
            \mathbb{E}_{\hat{\mu}^X}\|\cdot\|^2&=\frac{1}{Z}\int\|x\|^2\exp\left(-\frac{1}{n}\sum_{i=1}^n\left[\partial_1 \ell(h(\mu^X;a_i),b_i)h(x;a_i)\right]-\frac{\lambda'}{2}\|x\|^2\right)\mathrm{d}x\defeq\frac{Z'}{Z}
        \end{aligned}
    \end{equation}
    Now we bound $Z'$ and $Z$ respectively.
    \begin{equation*}
    \begin{aligned}
        Z'&\leq\int\|x\|^2\exp\left(\mathscr{L}-\frac{\lambda'}{2}\|x\|^2\right)\mathrm{d}x\lesssim\frac{\exp(\mathscr{L})d}{\lambda'},\\
        Z&\geq\int \exp\left(-\mathscr{L}-\frac{\lambda'}{2}\|x\|^2\right)\mathrm{d}x=\exp(-\mathscr{L})\left(\frac{2\pi}{\lambda'}\right)^{d/2}
        \end{aligned}
    \end{equation*}
Choose $\lambda'\leq (2\pi)^3\exp(-4\mathscr{L})$ which implies $\lambda'\leq\frac{(2\pi)^{\frac{d}{d-2}}}{\exp\left(\frac{4\mathscr{L}}{d-2}\right)}$, and we have
$\mathbb{E}_{\hat{\mu}^X}\|\cdot\|^2=\frac{Z'}{Z}\lesssim\frac{\exp(2\mathscr{L})}{\lambda'\left(\frac{2\pi}{\lambda'}\right)^{d/2}}d\leq d.$
Consider the case where $\ell$ is quadratic. $|h(\mu^X;a_i)|=|\int h(x;a_i)\mu^X(\mathrm{d}x)|\leq\int |h(x;a_i)|\mu(\mathrm{d}x)\leq\sqrt{\mathscr{L}}$, thus we have $|\partial_1 \ell(h(\mu^X;a_i),b_i)h(x;a_i)|=|(h(\mu^X;a_i)-b_i)h(x;a_i)|\leq \mathscr{L}+|b_i|\sqrt{\mathscr{L}}$. We can scale the label to ensure $\max_{i=1}^n|b_i|\leq \sqrt{\mathscr{L}}$, and we obtain $|\partial_1 \ell(h(\mu^X;a_i),b_i)h(x;a_i)|\leq 2\mathscr{L}$. The remaining proof keeps the same with $\lambda'\leq (2\pi)^3\exp(-8\mathscr{L})$. 
\item We will prove that Assumption \ref{second moment bound} holds if $|h(x;a)|\leq\sqrt{\mathscr{L}}(1+\|x\|)$ (such activation functions include \textsf{ReLU}, \textsf{GeLU}, Softplus, \textsf{SiLU}) and $|\partial_1 \ell|\leq\sqrt{\mathscr{L}}$. Under these conditions, we have $|\partial_1 \ell(h(\mu^X;a_i),b_i)h(x;a_i)|\leq \mathscr{L}(1+\|x\|)$. Then, based on \eqref{second moment}, we obtain
\begin{equation*}
\begin{aligned}
    Z'\leq\int\|x\|^2\exp\left(\mathscr{L}(1+\|x\|)-\frac{\lambda'}{2}\|x\|^2\right)\mathrm{d}x&\leq\exp(\mathscr{L})\int\|x\|^2\exp\left(\frac{3\mathscr{L}^2}{2\lambda'}-\frac{\lambda'}{3}\|x\|^2\right)\mathrm{d}x\\
    &\lesssim\exp\left(\mathscr{L}+\frac{3\mathscr{L}^2}{2\lambda'}\right)\frac{d}{\lambda'}.
    \end{aligned}
\end{equation*}
We also have
\begin{equation*}
\begin{aligned}
    Z\geq\int \exp\left(-\mathscr{L}(1+\|x\|)-\frac{\lambda'}{2}\|x\|^2\right)\mathrm{d}x&\geq\exp(\mathscr{L})\int\exp\left(-\frac{\mathscr{L}^2}{\lambda'}-\frac{3\lambda'}{4}\|x\|^2\right)\mathrm{d}x\\
    &=\exp\left(\mathscr{L}-\frac{\mathscr{L}^2}{\lambda'}\right)\left(\frac{4\pi}{3\lambda'}\right)^{d/2}
    \end{aligned}
\end{equation*}
Combining the upper bound of $Z'$ and the lower bound of $Z$, if $d\geq\frac{5\mathscr{L}^2}{\lambda'}\left(\log\frac{4\pi}{3}\right)^{-1}$, we obtain
\begin{equation*}
    \mathbb{E}_{\hat{\mu}^X}\|\cdot\|^2=\frac{Z'}{Z}\lesssim\exp\left(\frac{5\mathscr{L}^2}{2\lambda'}\right)\frac{d}{\lambda'}\left(\frac{3\lambda'}{4\pi}\right)^{d/2}\leq\exp\left(\frac{5\mathscr{L}^2}{2\lambda'}\right)\left(\frac{3}{4\pi}\right)^{d/2}d\leq d.
\end{equation*}
Note that $d\geq\frac{5\mathscr{L}^2}{\lambda'}\left(\log\frac{4\pi}{3}\right)^{-1}$ is possible for large-scale problems.
\end{itemize}
\paragraph{MMD minimization}
We now prove that objective \eqref{MMD} satisfies Assumption~\ref{second moment bound} with Gaussian RBF kernel. We choose $\sigma'$ in Gaussian RBF kernel $k$ to be $\sigma$ for brevity. We reformulate \eqref{MMD} as
\begin{equation}
    F(\rho)=\hat{\mathcal{M}}(\rho)+\frac{\lambda'}{2}\mathbb{E}_{x\sim\rho}\|x\|^2.
\end{equation}
According to the definition of $\hat{\mathcal{M}}(\rho)$ in Section~\ref{sec:ex}, the functional derivative of $\hat{\mathcal{M}}(\rho)$ is
\begin{equation}
    \frac{\delta \hat{\mathcal{M}}}{\delta\rho}(\rho,x)=\underbrace{2\iiint p(x;z)p(x';z')k(z,z')\mathrm{d}z\mathrm{d}z'\mathrm{d}\rho(x')}_{\textsf{P}}-\underbrace{\frac{2}{n}\sum_{i=1}^n\int p(x;z)k(z,z_i)\mathrm{d}z}_{\textsf{Q}}
\end{equation}
Next we bound each part of $\frac{\delta \hat{\mathcal{M}}}{\delta\rho}(\rho,x)$. For $\textsf{P}$, we have
\begin{equation*}
    \begin{aligned}
        \frac{1}{2}\textsf{P}&=\frac{1}{(2\pi\sigma^2)^{d}}\iiint\exp\left(-\frac{\|x-z\|^2}{2\sigma^2}-\frac{\|x'-z'\|^2}{2\sigma^2}-\frac{\|z-z'\|^2}{2\sigma^2}\right)\mathrm{d}z\mathrm{d}z'\mathrm{d}\rho(x')\\
        &=\frac{(\pi\sigma^2)^{\frac{d}{2}}}{(2\pi\sigma^2)^{d}}\iint\exp\left(-\frac{\|x-x'\|^2}{6\sigma^2}-\frac{3\|z'-\frac{2}{3}x'-\frac{1}{3}x\|^2}{4\sigma^2}\right)\mathrm{d}z'\mathrm{d}\rho(x')\\
        &=\left(\frac{1}{\sqrt{3}}\right)^d\int\exp\left(-\frac{\|x-x'\|^2}{6\sigma^2}\right)\mathrm{d}\rho(x')\leq \left(\frac{1}{\sqrt{3}}\right)^d
    \end{aligned}
\end{equation*}
where the last inequality follows from the relation $\exp\left(-\frac{\|x-x'\|^2}{6\sigma^2}\right)\leq 1$. For $\mathsf{Q}$, we have
\begin{equation*}
    \begin{aligned}
        \frac{1}{2}\mathsf{Q}&=\frac{1}{(2\pi\sigma^2)^{\frac{d}{2}}}\frac{1}{n}\sum_{i=1}^n\int\exp\left(-\frac{\|x-z\|^2}{2\sigma^2}-\frac{\|z-z_i\|^2}{2\sigma^2}\right)\mathrm{d}z\\
        &=\frac{1}{(2\pi\sigma^2)^{\frac{d}{2}}}\frac{1}{n}\sum_{i=1}^n\exp\left(-\frac{\|x\|^2+\|z_i\|^2}{2\sigma^2}+\frac{\|z_i+x\|^2}{4\sigma^2}\right)\
        \int\exp\left(-\frac{\|z-\frac{1}{2}z_i-\frac{1}{2}x\|^2}{\sigma^2}\right)\mathrm{d}z\\
        &=\left(\frac{1}{\sqrt{2}}\right)^d\frac{1}{n}\sum_{i=1}^n\exp\left(-\frac{\|x\|^2+\|z_i\|^2}{2\sigma^2}+\frac{\|z_i+x\|^2}{4\sigma^2}\right)\leq\left(\frac{1}{\sqrt{2}}\right)^d
    \end{aligned}
\end{equation*}
where the last inequality follows from the relation $\|z_i+x\|^2\leq 2\|z_i\|^2+2\|x\|^2$. Note that $\mathsf{P}\geq 0$ and $\textsf{Q}\geq 0$. Combining the bound of $\textsf{P}$ and $\textsf{Q}$, we obtain the bound of $\frac{\delta \hat{\mathcal{M}}}{\delta\rho}(\rho,x)$ as follows:
\begin{equation}
\begin{aligned}
    -\sqrt{2}\leq-2\left(\frac{1}{\sqrt{2}}\right)^d\leq\frac{\delta \hat{\mathcal{M}}(\mu)}{\delta\mu}(x)=\textsf{P}-\textsf{Q}
    \leq 2\left(\frac{1}{\sqrt{3}}\right)^d\leq \sqrt{3}
    \end{aligned}
\end{equation}
Let $\hat{\mu}^X(x)={\exp\left(-\frac{\delta F}{\delta \rho}(\mu^X,x)\right)}/{Z}$ where $Z=\int \exp\left(-\frac{\delta F}{\delta \rho}(\mu^X,x)\right) \mathrm{d}x$, and we have
    \begin{equation}
    \label{second moment II}
        \begin{aligned}
            \mathbb{E}_{\hat{\mu}^X}\|\cdot\|^2&=\frac{1}{Z}\int\|x\|^2\exp\left(-\frac{\delta \hat{\mathcal{M}}}{\delta\rho}(\mu^X,x)-\frac{\lambda'}{2}\|x\|^2\right)\mathrm{d}x\defeq\frac{Z'}{Z}
        \end{aligned}
    \end{equation}
    Now we bound $Z'$ and $Z$ respectively.
    \begin{equation*}
    \begin{aligned}
        Z'&\leq\int\|x\|^2\exp\left(\sqrt{2}-\frac{\lambda'}{2}\|x\|^2\right)\mathrm{d}x\lesssim\frac{\exp(\sqrt{2})d}{\lambda'},\\
        Z&\geq\int \exp\left(-\sqrt{3}-\frac{\lambda'}{2}\|x\|^2\right)\mathrm{d}x=\exp(-\sqrt{3})\left(\frac{2\pi}{\lambda'}\right)^{d/2}
        \end{aligned}
    \end{equation*}
    Thus in order to ensure $\mathbb{E}_{\hat{\mu}^X}\|\cdot\|^2=\frac{Z'}{Z}\lesssim\frac{\exp(\sqrt{2}+\sqrt{3})\lambda'^{\frac{d-2}{2}}}{(2\pi)^{\frac{d}{2}}}d\leq d$, it suffices to choose $\lambda'\leq 3\pi/25$.
    \paragraph{KSD minimization} Assume the score function $s_{\rho_*}$ satisfies \eqref{assump on score} and choose the kernel $k$ to be \eqref{kernel}, and the Stein kernel $u_{\rho_*}$ satisfies $\sup_{x,x'\in\mathbb{R}^d}\max\{|u_{\rho_*}|,\|\nabla_xu_{\rho_*}\|,\|\nabla_x^2u_{\rho_*}\|_{\textsf{op}}\}\leq \mathscr{L}$ \citep{suzuki2023convergence}. We now prove the following objective
    \begin{equation}
        F(\rho)=\textsf{KSD}(\rho)+\frac{\lambda'}{2}\mathbb{E}_{x\sim\rho}\|x\|^2
    \end{equation}
    satisfies Assumption \ref{second moment bound}, with $\textsf{KSD}$ defined by $\textsf{KSD}(\rho)=\iint u_{\rho_*}(x,x')\mathrm{d}\rho(x)\mathrm{d}\rho(x')$. The functional derivative of $\textsf{KSD}$ is
    \begin{equation*}
        \frac{\delta\textsf{KSD}}{\delta\rho}(\rho,x)=\int u_{\rho_*}(x,x')\mathrm{d}\rho(x').
    \end{equation*}
   The functional derivative is bounded as
    \begin{equation*}
        \left|\frac{\delta\textsf{KSD}}{\delta\rho}(\rho,x)\right|\leq\int |u_{\rho_*}(x,x')|\mathrm{d}\rho(x')\leq \mathscr{L}.
    \end{equation*}
Let $\hat{\mu}^X(x)={\exp\left(-\frac{\delta F}{\delta \rho}(\mu^X,x)\right)}/{Z}$ where $Z=\int \exp\left(-\frac{\delta F}{\delta \rho}(\mu^X,x)\right) \mathrm{d}x$, and we have
    \begin{equation}
        \begin{aligned}
            \mathbb{E}_{\hat{\mu}^X}\|\cdot\|^2&=\frac{1}{Z}\int\|x\|^2\exp\left(-\frac{\delta \textsf{KSD}}{\delta\rho}(\mu^X,x)-\frac{\lambda'}{2}\|x\|^2\right)\mathrm{d}x\defeq\frac{Z'}{Z}
        \end{aligned}
    \end{equation}
    Now we bound $Z'$ and $Z$ respectively.
    \begin{equation*}
    \begin{aligned}
        Z'&\leq\int\|x\|^2\exp\left(\mathscr{L}-\frac{\lambda'}{2}\|x\|^2\right)\mathrm{d}x\lesssim\frac{\exp(\mathscr{L})d}{\lambda'},\\
        Z&\geq\int \exp\left(-\mathscr{L}-\frac{\lambda'}{2}\|x\|^2\right)\mathrm{d}x=\exp(-\mathscr{L})\left(\frac{2\pi}{\lambda'}\right)^{d/2}
        \end{aligned}
    \end{equation*}
    Thus we have $\mathbb{E}_{\hat{\mu}^X}\|\cdot\|^2=\frac{Z'}{Z}\lesssim\frac{\exp(2\mathscr{L})d\lambda'^{\frac{d}{2}-1}}{(2\pi)^{\frac{d}{2}}}\leq d$ for $\lambda'\leq(2\pi)^3\exp\left(-{4\mathscr{L}}\right)$.
\subsection{Verification of Assumption~\ref{initialization and invariant}}
\label{verification of bounded init}
\paragraph{Training mean-field neural networks}
 Reformulate the objective \eqref{empirical risk minimization} with  $\mu_0=\mathcal{N}(0,I_{d})$:
    $$F(\rho)=\frac{1}{n}\sum_{i=1}^n\ell(h(\rho;a_i),b_i)+\frac{\lambda'}{2}\mathbb{E}_{x\sim \rho}[\|x\|^2].$$
    \begin{itemize}
        \item 
      If $l$ is $\sqrt{\mathscr{L}}$-Lipschitz, we have $|\ell(h(\rho;a),b)|\leq \sqrt{\mathscr{L}}|h(\rho;a)-b|$. If $|h(x;a)|\leq\sqrt{\mathscr{L}}$, we have $|h(\rho;a)|\leq\sqrt{\mathscr{L}}$. Since $\mu_0=\mathcal{N}(0,I_{2d})$, $\mathbb{E}_{x\sim \mu^X_0}[\|x\|^2]\lesssim d$. With $\lambda'\leq\min\{\mathscr{L},d\}$, we have $F(\mu^X_0)\lesssim \sqrt{\mathscr{L}}(\sqrt{\mathscr{L}}+\max_{i=1}^n|b_i|)+d$. We can normalize the data samples to ensure $\max_{i=1}^n|b_i|\lesssim d\land \sqrt{\mathscr{L}}$. Thus $F(\mu^X_0)\lesssim\mathscr{L}+d$.
     \item If $|h(x;a)|\leq\sqrt{\mathscr{L}}(1+\|x\|)$, we have $|h(\mu^X_0;a)|\leq\sqrt{\mathscr{L}}\int(1+\|x\|)\mu^X_0(\mathrm{d}x)\lesssim\sqrt{\mathscr{L}}d^{1/2}$. If $\ell$ is $\sqrt{\mathscr{L}}$-Lipschitz, we have $|l(h(\mu^X_0;a_i),b_i)|\leq \sqrt{\mathscr{L}}|h(\mu^X_0;a_i)-b_i|\lesssim \mathscr{L}d^{1/2}+\sqrt{\mathscr{L}}\max_{i=1}^n|b_i|$. We can normalize the data samples to ensure $\max_{i=1}^n|b_i|\lesssim d\land \sqrt{\mathscr{L}}$. Thus we have $F(\mu^X_0)\lesssim\mathscr{L}d+d$.
     \end{itemize}
\paragraph{MMD minimization}
Reformulate the objective~\eqref{MMD} with Gaussian RBF kernel ($\sigma'=\sigma$) and $\mu_0=\mathcal{N}(0,I_d)$:
\begin{equation}
\label{MMDappendix}
    F(\rho)=\hat{\mathcal{M}}(\rho)+\frac{\lambda'}{2}\mathbb{E}_{x\sim\rho}\|x\|^2,
\end{equation}
where
\begin{equation*}
\begin{aligned}
    \hat{\mathcal{M}}(\rho)&=\iiint p(x;z)p(x';z')k(z,z')\mathrm{d}z\mathrm{d}z'\mathrm{d}(\rho\times\rho)(x,x')-2\int\left(\frac{1}{n}\sum_{i=1}^n\int p(x;z)k(z,z_i)\mathrm{d}z\right)\mathrm{d}\rho(x)\\
    &=\frac{1}{{3}^{d/2}}\int\exp\left(-\frac{\|x-x'\|^2}{6\sigma^2}\right)\mathrm{d}(\rho\times\rho)(x,x')-\frac{2}{2^{d/2}}\frac{1}{n}\sum_{i=1}^n\int\exp\left(-\frac{\|x-z_i\|^2}{4\sigma^2}\right)\mathrm{d}\rho(x)\\
    &\leq\frac{1}{{3}^{d/2}}\int\exp\left(-\frac{\|x-x'\|^2}{6\sigma^2}\right)\mathrm{d}(\rho\times\rho)(x,x')\leq \frac{1}{{3}^{d/2}}\leq \mathscr{L}
    \end{aligned}
\end{equation*}
Thus $F(\mu^X_0)=\hat{\mathcal{M}}(\mu^X_0)+\frac{\lambda'}{2}\mathbb{E}_{x\sim\mu^X_0}\|x\|^2\lesssim \mathscr{L}+d,$ which satisfies Assumption~\ref{initialization and invariant}.
\paragraph{KSD minimization} Consider the same objective in \eqref{KSD obj} with  $\mu_0=\mathcal{N}(0,I_d)$:
\begin{equation*}
    F(\rho)=\textsf{KSD}(\rho)+\frac{\lambda'}{2}\mathbb{E}_{x\sim\rho}\|x\|^2.
\end{equation*}
If we choose kernel $k(x,x')=\exp\left(-\frac{\|x\|^2}{2\sigma_1^2}-\frac{\|x'\|^2}{2\sigma_1^2}-\frac{\|x-x'\|^2}{2\sigma_2^2}\right)$ and
assume the score function of $\rho_*$ satisfies $\max\{\|\nabla\log\rho_*(x)\|,\|\nabla^{\otimes 2}\log\rho_*(x)\|_{\textsf{op}},\|\nabla^{\otimes 3}\log\rho_*(x)\|_{\textsf{op}}\}\leq \mathscr{L}(1+\|x\|)$, then the Stein kernel $u_{\rho_*}$ satisfies $\sup_{x,x'\in\mathbb{R}^d}\max\{|u_{\rho_*}|,\|\nabla_xu_{\rho_*}\|,\|\nabla_x^2u_{\rho_*}\|_{\textsf{op}}\}\leq \mathscr{L}$ according to the statement of Appendix A in \citet{suzuki2023convergence}. We have
\begin{equation*}
\begin{aligned}
    F(\mu^X_0)&=\textsf{KSD}(\mu^X_0)+\frac{\lambda'}{2}\mathbb{E}_{x\sim\mu_0^X}\|x\|^2\\
    &=\iint u_{\rho_*}(x,x')\mathrm{d}\mu^X(x)\mathrm{d}\mu^X(x')+\frac{\lambda'}{2}\mathbb{E}_{x\sim\mu^X_0}\|x\|^2\\
    &\lesssim \mathscr{L}+d,
\end{aligned}
\end{equation*}
which satisfies Assumption~\ref{initialization and invariant}.
\subsection{Verification of Assumption~\ref{initialization and invariant for N-particle}}
\label{verification of bounded par-init}
\paragraph{Training mean-field neural networks}
 Similar to examples of training mean-field neural networks above, we initialize $\mu^N_0=\mathcal{N}(0,I_{2Nd})$.
    $$\mathbb{E}_{\textbf{x}\sim\mu^N}F(\mu_{\textbf{x}})\defeq\mathbb{E}_{\textbf{x}\sim \mu^N}\frac{1}{n}\sum_{i=1}^n\left[\ell\left(\frac{1}{N}\sum_{s=1}^Nh(x^s;a_i),b_i\right)\right]+\frac{\lambda'}{2}\mathbb{E}_{\textbf{x}\sim \mu^N}\frac{1}{N}\sum_{s=1}^N\left[\|x^s\|^2\right],$$ where $\textbf{x}=(x^1,...,x^N)$, $x^i\sim \mu^i$ for $i=1,...,N$ and $\mu^N=\otimes_{i=1}^N\mu^i=\text{Law}(x^1,...,x^N)$. %
    \begin{itemize}
    \item
    If $|h(x;a)|\leq\sqrt{\mathscr{L}}$ and $\ell$ is $\sqrt{\mathscr{L}}$-Lipschitz,  and $\mathbb{E}_{\textbf{x}_0\sim \mu_0^N}\frac{1}{n}\sum_{i=1}^n\left[\ell\left(\frac{1}{N}\sum_{i=1}^Nh(x_0^i;a_i),b_i\right)\right]\lesssim\sqrt{\mathscr{L}}(\sqrt{\mathscr{L}}+\max_{i=1}^n|b_i|)$ and thus $\mathbb{E}_{\mu_0^N}F(\mu_{\textbf{x}_0})\lesssim\mathscr{L}+\sqrt{\mathscr{L}}\max_{i=1}^n|b_i|+d$. We can normalize the data samples to ensure $\max_{i=1}^n|b_i|\lesssim d\land \sqrt{\mathscr{L}}$. Thus we have $\mathbb{E}_{\mu_0^N}F(\mu_{\textbf{x}_0})=O(\mathscr{L}+d)$.
    \item
    If $|h(x;a)|\leq\sqrt{\mathscr{L}}(1+\|x\|)$ and $\ell$ is $\sqrt{\mathscr{L}}$-Lipschitz, $\mathbb{E}_{\textbf{x}_0\sim \mu_0^N}\frac{1}{n}\sum_{i=1}^n\left[\ell\left(\frac{1}{N}\sum_{s=1}^Nh(x_0^s;a_i),b_i\right)\right]\leq\sqrt{\mathscr{L}}\left(\sqrt{\mathscr{L}}\frac{1}{N}\sum_{s=1}^N(1+\mathbb{E}_{\textbf{x}_0\sim \mu_0^N}\|x_0^s\|)+\max_{i=1}^n|b_i|\right)\lesssim\mathscr{L}d^{1/2}+\sqrt{\mathscr{L}}\max_{i=1}^n|b_i|$. We can normalize the data samples to ensure $\max_{i=1}^n|b_i|\lesssim d\land \sqrt{\mathscr{L}}$. Thus we have $\mathbb{E}_{\mu_0^N}F(\mu_{\textbf{x}_0})=O(\mathscr{L}d+d)$
\end{itemize}
\paragraph{MMD minimization}
Now we verify Assumption~\ref{initialization and invariant for N-particle} for the example of density estimation. We consider the N-particle approximation of the objective \eqref{MMDappendix} with the initialization $\mu^N_0=\mathcal{N}(0,I_{Nd})$.
\begin{equation*}
\begin{aligned}
    \mathbb{E}_{\mu^N}&\hat{\mathcal{M}}(\mu_{\textbf{x},\textbf{y}})\\
    &\defeq\mathbb{E}_{\textbf{x},\textbf{y}\sim\mu^N}\left[\frac{1}{N^2}\sum_{s=1}^N\sum_{t=1}^N\iint p(x^s;z)p(y^t;z')k(z,z')\mathrm{d}z\mathrm{d}z'
    -\frac{2}{nN}\sum_{i=1}^n\sum_{s=1}^N\int p(x^s;z)k(z,z_i)\mathrm{d}z\right]\\
    &\leq\mathbb{E}_{\textbf{x},\textbf{y}\sim\mu^N}\left[\frac{1}{N^2}\sum_{s=1}^N\sum_{t=1}^N\iint p(x^s;z)p(y^t;z')k(z,z')\mathrm{d}z\mathrm{d}z'\right]\\
    &=\left(\frac{1}{\sqrt{3}}\right)^d\mathbb{E}_{\textbf{x},\textbf{y}\sim\mu^N}\left[\frac{1}{N^2}\sum_{s=1}^N\sum_{t=1}^N\exp\left(-\frac{\|x^s-y^t\|^2}{6\sigma^2}\right)\right]\leq\left(\frac{1}{\sqrt{3}}\right)^d\leq \mathscr{L}
    \end{aligned}
\end{equation*}
where $\textbf{x}=(x^1,...,x^N)$ and $\textbf{y}=(y^1,...,y^N)$. Thus we can upper bound $\mathbb{E}_{\textbf{x}_0,\textbf{y}_0\sim\mu^N_0}F(\mu_{\textbf{x}_0,\textbf{y}_0})$ as follows:
\begin{equation*}
\mathbb{E}_{\textbf{x}_0,\textbf{y}_0\sim\mu^N_0}F(\mu_{\textbf{x}_0,\textbf{y}_0})=\mathbb{E}_{\textbf{x}_0,\textbf{y}_0\sim\mu^N_0}\hat{\mathcal{M}}(\mu_{\textbf{x},\textbf{y}})+\frac{\lambda'}{2}\mathbb{E}_{\textbf{x}_0\sim \mu_0^N}\frac{1}{N}\sum_{s=1}^N\left[\|x_0^s\|^2\right]\lesssim \mathscr{L}+d
\end{equation*}
which satisfies Assumption~\ref{initialization and invariant for N-particle}.
\paragraph{KSD minimization}
Similar to the verification of Assumption~\ref{initialization and invariant} above, we have the following relation for $\mu_0^N=\mathcal{N}(0,I_{Nd})$ under the same assumptions on the score function and kernel:
\begin{equation*}
\begin{aligned}
    \mathbb{E}_{\textbf{x}_0\sim\mu_0}F(\mu_{\textbf{x}_0})&=\textsf{KSD}(\mu_{\textbf{x}_0})+\frac{\lambda'}{2}\mathbb{E}_{\textbf{x}_0\sim \mu_0^N}\frac{1}{N}\sum_{s=1}^N\left[\|x_0^s\|^2\right]\\
    &=\mathbb{E}_{\textbf{x}_0\sim\mu_0}\frac{1}{N^2}\sum_{i=1}^N\sum_{j=1}^Nu_{\mu_*}(x_0^i,x_0^j)+\frac{\lambda'}{2}\mathbb{E}_{\textbf{x}_0\sim \mu_0^N}\frac{1}{N}\sum_{s=1}^N\left[\|x_0^s\|^2\right]\lesssim\mathscr{L}+d,
    \end{aligned}
\end{equation*}
which satisfies Assumption~\ref{initialization and invariant for N-particle}.
\section{Continuous-time results}
\label{Continuous-time results}
In this section, we give the explicit rate of Theorem 2.1 and Theorem 2.2 proposed by \citet{chen2023uniform} with a specific choice of parameters and then provide the detailed proof of Theorem \ref{convergence rate of the UMFLD} and Theorem \ref{convergence of particle dynamics} by reparameterizing $\gamma$.
\subsection{Proof of Theorem \ref{convergence rate of the UMFLD}}
Our proof is directly adapted from Theorem 2.1 in \citet{chen2023uniform} using hypocoercivity in \citet{villani2009hypocoercivity}. \citet{chen2023uniform} prove the Lyapunov functional
\begin{equation}
    \mathcal{E}(\mu_t)=\mathcal{F}(\mu_t)+\mathsf{FI}_S(\mu_t\|\hat{\mu}_t)
\end{equation}
is decaying along the~\ref{UMFLD} with $S=\left(\begin{array}{cc}
     c&b  \\
     b&a 
\end{array}\right)\otimes I_d$ and $\gamma=1$. Let $A_t=\nabla_v$, $B_t=v\cdot\nabla_x-{D}_{\rho}F(\mu_t^X,x)\cdot\nabla_v$, $C_t=[A_t,B_t]=A_tB_t-B_tA_t=\nabla_x$ and $Y_t=(\|A_tu_t\|,\|A_t^2u_t\|, \|C_tu_t\|, \|C_tA_tu_t\|)^{\mathsf{T}}$ where $u_t=\log\frac{\mu_t}{\hat{\mu}_t}$ and $\|\cdot\|\defeq\|\cdot\|_{L^2(\mu_t)}$. More specifically, \citet{chen2023uniform} prove that
\begin{equation}
    \frac{\mathrm{d}}{\mathrm{d}t}\mathcal{E}(\mu_t)\leq -Y_t^{\mathsf{T}}\mathcal{K}Y_t,
\end{equation}
where $$\mathcal{K}=\left(\begin{array}{cccc}
    1+2a-4\mathscr{L}b & -2b & -2a-2\mathscr{L}c & 0 \\
    0 & 2a & -2\mathscr{L}c & -4b \\
    0 & 0 & 2b & 0\\
    0 & 0 & 0 & 2c
\end{array}\right).$$ The choice of $a,\,b,\,c$ should satisfies $ac>b^2$ and $K\succ 0$. If we choose $a=c=2\mathscr{L}$ and $b=1$, the smallest eigenvalue of $\mathcal{K}$ is $\lambda_{\textsf{min}}(\mathcal{K})=1$, and thus we have 
\begin{align*}
\frac{\mathrm{d}}{\mathrm{d}t}(\mathcal{E}(\mu_t)-\mathcal{E}(\mu_*))&\leq-(\|A_tu_t\|^2+\|A_t^2u_t\|^2+\|C_tu_t\|^2+\|C_tA_tu_t\|^2)\\
&\leq-(\|A_tu_t\|^2+\|C_tu_t\|^2)=-\frac{1}{2}\mathsf{FI}(\mu_t\|\hat{\mu}_t)-\frac{1}{2}\mathsf{FI}(\mu_t\|\hat{\mu}_t)\\
&\leq-\mathscr{C}_{\textsf{LSI}}\textsf{KL}(\mu_t\|\hat{\mu}_t)-\frac{1}{2\lambda_{\textsf{max}}(S)}\mathsf{FI}_S(\mu_t\|\hat{\mu_t})\\
&\leq-\mathscr{C}_{\textsf{LSI}}(\mathcal{F}(\mu_t)-\mathcal{F}(\mu_*))-\frac{1}{4\mathscr{L}+2}\mathsf{FI}_S(\mu_t\|\hat{\mu_t})\\
&\leq-\frac{\mathscr{C}_{\textsf{LSI}}}{6\mathscr{L}}(\mathcal{E}(\mu_t)-\mathcal{E}(\mu_*))
\end{align*}
Applying Gr{\"o}nwall's inequality, we obtain
\begin{equation}
\label{convg of the UMLD with gamma=1}
    \mathcal{F}(\mu_t)-\mathcal{F}(\mu_*)\leq\mathcal{E}(\mu_t)-\mathcal{E}(\mu_*)\leq(\mathcal{E}(\mu_0)-\mathcal{E}(\mu_*))\exp\left(-\frac{\mathscr{C}_{\textsf{LSI}}}{6\mathscr{L}}t\right).
\end{equation}
Note that the proof in \citet{chen2023uniform} also considers the approximation technique to remove some restrictive assumptions they make, which we omit in our proof. Now we consider a more general $\gamma$ in the proof above. Analogous to the proof of Lemma 32 in \citet{villani2009hypocoercivity}, if we incorporate a general $\gamma$, the diagonal elements of upper triangular matrix $\mathcal{K}$ will become $(\gamma+2\gamma a-4\mathscr{L}b,\,2\gamma a,\,2b,\,2\gamma c)$. If we choose $\gamma=\sqrt{\mathscr{L}}$, $b=1/\sqrt{\mathscr{L}}$, $a=2$ and $c=1/\mathscr{L}$, the smallest eigenvalue of $K$ will become $\lambda_{\textsf{min}}(\mathcal{K})=2/\sqrt{\mathscr{L}}$. Similar to the previous proof, we have
\begin{equation*}
    \begin{aligned}
        \frac{\mathrm{d}}{\mathrm{d}t}(\mathcal{E}(\mu_t)-\mathcal{E}(\mu_*))&\leq-\frac{2}{\sqrt{\mathscr{L}}}(\|A_tu_t\|^2+\|A_t^2u_t\|^2+\|C_tu_t\|^2+\|C_tA_tu_t\|^2)\\
        &\leq-\frac{2}{\sqrt{\mathscr{L}}}(\|A_tu_t\|^2+\|C_tu_t\|^2)=-\frac{1}{\sqrt{\mathscr{L}}}\mathsf{FI}(\mu_t\|\hat{\mu}_t)-\frac{1}{\sqrt{\mathscr{L}}}\mathsf{FI}(\mu_t\|\hat{\mu}_t)\\
        &\leq-\frac{2\mathscr{C}_{\textsf{LSI}}}{\sqrt{\mathscr{L}}}\textsf{KL}(\mu_t\|\hat{\mu}_t)-\frac{1}{\lambda_{\textsf{max}}(S)\sqrt{\mathscr{L}}}\mathsf{FI}_S(\mu_t\|\hat{\mu}_t)\\
        &\leq-\frac{2\mathscr{C}_{\textsf{LSI}}}{\sqrt{\mathscr{L}}}(\mathcal{F}(\mu_t)-\mathcal{F}(\mu_*))-\frac{1}{3\sqrt{\mathscr{L}}}\mathsf{FI}_S(\mu_t\|\hat{\mu}_t)\\
        &\leq-\frac{\mathscr{C}_{\textsf{LSI}}}{3\sqrt{\mathscr{L}}}(\mathcal{E}(\mu_t)-\mathcal{E}(\mu_*))
    \end{aligned}
\end{equation*}
where the fourth inequality follows from $\lambda_{\textsf{max}}(S)=\frac{\frac{1}{\mathscr{L}}+2+\sqrt{\frac{1}{\mathscr{L}^2}+4}}{2}\leq\frac{1}{\mathscr{L}}+2\leq 3.$
Applying Gr{\"o}nwall's inequality, we obtain
\begin{equation}
\label{convg of the UMLD with gamma}
    \mathcal{F}(\mu_t)-\mathcal{F}(\mu_*)\leq\mathcal{E}(\mu_t)-\mathcal{E}(\mu_*)\leq(\mathcal{E}(\mu_0)-\mathcal{E}(\mu_*))\exp\left(-\frac{\mathscr{C}_{\textsf{LSI}}}{3\sqrt{\mathscr{L}}}t\right),
\end{equation}
which completes the proof of Theorem \ref{convergence rate of the UMFLD}. \cref{convg of the UMLD with gamma} exhibits a faster rate than the rate of \cref{convg of the UMLD with gamma=1}.
\subsection{Proof of Theorem \ref{convergence of particle dynamics}}
Our proof is directly adapted from Theorem 2.2 in \citet{chen2023uniform} using hypocoercivity in \citet{villani2009hypocoercivity}. 
\citet{chen2023uniform} prove that the Lyapunov functional
\begin{equation}
\mathcal{E}^N(\mu_t^N)=\mathcal{F}^N(\mu_t^N)+\textsf{FI}^N_S(\mu_t^N\|\mu_*^N)
\end{equation}
is decaying along the \ref{Par-UMFLD} with $S=\left(\begin{array}{cc}
     c&b  \\
     b&a 
\end{array}\right)\otimes I_d$ and $\gamma=1$. Let $u_t^N=\log\frac{\mu_t^N}{\mu^N_*}$, $\|\cdot\|\defeq\|\cdot\|_{L^2(\mu_t^N)}$ and $$Y_t^N=(\|\nabla_{\textbf{v}}u_t^N\|,\|\nabla_{\textbf{v}}^2u_t^N\|, \|\nabla_{\textbf{x}}u_t^N\|,\|\nabla_{\textbf{x}}\nabla_{\textbf{v}}u_t^N\|)^{\mathsf{T}}.$$ \citet{chen2023uniform} prove that 
\begin{equation}
    \frac{\mathrm{d}}{\mathrm{d}t}\mathcal{E}^N(\mu_t^N)\leq-(Y_t^N)^{\mathsf{T}}\mathcal{K}Y_t^N
\end{equation}
where
$$\mathcal{K}=\left(\begin{array}{cccc}
    1+2a-4\mathscr{L}b & -2b & -2a & 0 \\
    0 & 2a & -4\mathscr{L}c & -4b \\
    0 & 0 & 2b & 0\\
    0 & 0 & 0 & 2c
\end{array}\right).$$
The choice of $a,\,b,\,c$ should satisfies $ac>b^2$ and $K\succ 0$. If we choose $a=c=2\mathscr{L}$ and $b=1$, the smallest eigenvalue of $K$ is $\lambda_{\textsf{min}}(\mathcal{K})=1$, and thus we have 
\begin{equation}
\label{fisher upper bound}
\begin{aligned}
     \frac{\mathrm{d}}{\mathrm{d}t}\mathcal{E}^N(\mu_t^N)&\leq-(\|\nabla_{\textbf{v}}u_t^N\|^2+\|\nabla_{\textbf{v}}^2u_t^N\|^2+\|\nabla_{\textbf{x}}u_t^N\|^2+\|\nabla_{\textbf{x}}\nabla_{\textbf{v}}u_t^N\|^2)\\
     &\leq-(\|\nabla_{\textbf{v}}u_t^N\|^2+\|\nabla_{\textbf{x}}u_t^N\|^2)=-\textsf{FI}(\mu_t^N\|\mu_*^N)
\end{aligned}
\end{equation}
Since $\mu_*^N$ does not satisfy the uniform LSI, we can not utilize the same technique to upper bound $-\textsf{FI}(\mu_t^N\|\mu_*^N)$. \citet{chen2022uniform} and \citet{chen2023uniform} obtain the lower bound of the relative Fisher information $\textsf{FI}(\mu_t^N\|\mu_*^N)$ using other technique to circumvent the uniform LSI of $\mu_*^N$. We will directly provide the conclusion instead of providing many details about that technique in this paper, and we refer our readers to \citet{chen2022uniform,chen2023uniform} for the precise proof. 
\citet{chen2023uniform} propose that
\begin{equation*}
    \begin{aligned}
        \textsf{FI}(\mu_t^N\|\mu_*^N)&=\frac{1}{2}\textsf{FI}(\mu_t^N\|\mu_*^N)+\frac{1}{2}\textsf{FI}(\mu_t^N\|\mu_*^N)\\
        &\geq\frac{1}{2}\left[2(1-\varepsilon)\mathscr{C}_{\textsf{LSI}}-\frac{\mathscr{L}}{N}\left(16+12(\varepsilon^{-1}-1)\frac{\mathscr{L}}{\mathscr{C}_{\textsf{LSI}}}\right)\right](\mathcal{F}^N(\mu_t^N)-N\mathcal{F}(\mu_*))\\
        &\quad+\frac{1}{2}\textsf{FI}(\mu_t^N\|\mu_*^N)-\frac{\mathscr{L}d}{\mathscr{C}_{\textsf{LSI}}}(5\mathscr{C}_{\textsf{LSI}}+3(\varepsilon^{-1}-1)\mathscr{L})
    \end{aligned}
\end{equation*}
for $\varepsilon\in(0,1)$. If we choose $\varepsilon=1/2$ and $N\geq\frac{32\mathscr{L}}{\mathscr{C}_{\textsf{LSI}}}+\frac{24\mathscr{L}^2}{\mathscr{C}_{\textsf{LSI}}^2}$, we have
\begin{equation*}
    \begin{aligned}
        \textsf{FI}(\mu_t^N\|\mu_*^N)&\geq\frac{\mathscr{C}_{\textsf{LSI}}}{4}(\mathcal{F}^N(\mu_t^N)-N\mathcal{F}(\mu_*))+\frac{1}{2\lambda_{\textsf{max}}(S)}\textsf{FI}_S(\mu_t^N\|\mu_*^N)-\frac{\mathscr{L}d}{\mathscr{C}_{\textsf{LSI}}}(5\mathscr{C}_{\textsf{LSI}}+3\mathscr{L})\\
        &\geq\frac{\mathscr{C}_{\textsf{LSI}}}{4}(\mathcal{F}^N(\mu_t^N)-N\mathcal{F}(\mu_*))+\frac{1}{6\mathscr{L}}\textsf{FI}_S(\mu_t^N\|\mu_*^N)-\frac{\mathscr{L}d}{\mathscr{C}_{\textsf{LSI}}}(5\mathscr{C}_{\textsf{LSI}}+3\mathscr{L})\\
        &\geq \frac{\mathscr{C}_{\textsf{LSI}}}{24\mathscr{L}}(\mathcal{E}^N(\mu^N_t)-N\mathcal{E}(\mu_*))-\frac{\mathscr{L}d}{\mathscr{C}_{\textsf{LSI}}}(5\mathscr{C}_{\textsf{LSI}}+3\mathscr{L})
    \end{aligned}
\end{equation*}
Combining \eqref{fisher upper bound} with the lower bound of Fisher information above, we obtain
\begin{equation*}
    \frac{\mathrm{d}}{\mathrm{d}t}(\mathcal{E}^N(\mu_t^N)-N\mathcal{E}(\mu_*))\leq-\frac{\mathscr{C}_{\textsf{LSI}}}{24\mathscr{L}}(\mathcal{E}^N(\mu^N_t)-N\mathcal{E}(\mu_*))+\frac{\mathscr{L}d}{\mathscr{C}_{\textsf{LSI}}}(5\mathscr{C}_{\textsf{LSI}}+3\mathscr{L})
\end{equation*}
Applying Gr{\"o}nwall's inequality, we obtain
\begin{equation}
\label{convg of the UNLD with gamma=1}
\begin{aligned}
    \mathcal{F}^N(\mu_t^N)-N\mathcal{F}(\mu_*)&\leq\mathcal{E}^N(\mu_t^N)-N\mathcal{E}(\mu_*)\\
    &\leq(\mathcal{E}^N(\mu_0^N)-N\mathcal{E}(\mu_*))\exp\left(-\frac{\mathscr{C}_{\textsf{LSI}}}{24\mathscr{L}}t\right)+\frac{\mathscr{L}\mathrm{d}t}{\mathscr{C}_{\textsf{LSI}}}(5\mathscr{C}_{\textsf{LSI}}+3\mathscr{L})\exp\left(-\frac{\mathscr{C}_{\textsf{LSI}}}{24\mathscr{L}}t\right)\\
    &\leq (\mathcal{E}^N(\mu_0^N)-N\mathcal{E}(\mu_*))\exp\left(-\frac{\mathscr{C}_{\textsf{LSI}}}{24\mathscr{L}}t\right)+\frac{120\mathscr{L}^2d}{\mathscr{C}_{\textsf{LSI}}}+\frac{72\mathscr{L}^3d}{\mathscr{C}_{\textsf{LSI}}^2}
    \end{aligned}
\end{equation}
where the last inequality follows from $\exp(-x)\leq (1+x)^{-1}$ for $x>-1$. Now we consider a more general $\gamma$ in the proof above. Analogous to the proof of Lemma 32 in \citet{villani2009hypocoercivity}, if we incorporate $\gamma$, the diagonal elements of upper triangular matrix $\mathcal{K}$ will become $(\gamma+2\gamma a-4\mathscr{L}b,\,2\gamma a,\,2b,\,2\gamma c)$. If we choose $\gamma=\sqrt{\mathscr{L}}$, $b=1/\sqrt{\mathscr{L}}$, $a=2$ and $c=1/\mathscr{L}$, the smallest eigenvalue of $K$ will become $\lambda_{\textsf{min}}(\mathscr{K})=2/\sqrt{\mathscr{L}}$. Similar to the previous proof, we have 
\begin{equation*}
\begin{aligned}
     \frac{\mathrm{d}}{\mathrm{d}t}(\mathcal{E}^N(\mu_t^N)-N\mathcal{E}(\mu_*))&\leq-\frac{2}{\sqrt{\mathscr{L}}}(\|\nabla_{\textbf{v}}u_t^N\|^2+\|\nabla_{\textbf{v}}^2u_t^N\|^2+\|\nabla_{\textbf{x}}u_t^N\|^2+\|\nabla_{\textbf{x}}\nabla_{\textbf{v}}u_t^N\|^2)\\
     &\leq-\frac{2}{\sqrt{\mathscr{L}}}(\|\nabla_{\textbf{v}}u_t^N\|^2+\|\nabla_{\textbf{x}}u_t^N\|^2)=-\frac{2}{\sqrt{\mathscr{L}}}\textsf{FI}(\mu_t^N\|\mu_*^N)\\
     &\leq-\frac{\mathscr{C}_{\textsf{LSI}}}{2\sqrt{\mathscr{L}}}(\mathcal{F}^N(\mu_t^N)-N\mathcal{F}(\mu_*))-\frac{1}{\lambda_{\textsf{max}}(S)\sqrt{\mathscr{L}}}\textsf{FI}_S(\mu_t^N\|\mu_*^N)\\
     &\quad+\frac{2\sqrt{\mathscr{L}}d}{\mathscr{C}_{\textsf{LSI}}}(5\mathscr{C}_{\textsf{LSI}}+3\mathscr{L})\\
     &\leq-\frac{\mathscr{C}_{\textsf{LSI}}}{2\sqrt{\mathscr{L}}}(\mathcal{F}^N(\mu_t^N)-N\mathcal{F}(\mu_*))-\frac{1}{3\sqrt{\mathscr{L}}}\textsf{FI}_S(\mu_t^N\|\mu_*^N)+\frac{2\sqrt{\mathscr{L}}d}{\mathscr{C}_{\textsf{LSI}}}(5\mathscr{C}_{\textsf{LSI}}+3\mathscr{L})\\
     &\leq-\frac{\mathscr{C}_{\textsf{LSI}}}{6\sqrt{\mathscr{L}}}(\mathcal{E}^N(\mu^N_t)-N\mathcal{E}(\mu_*))+\frac{2\sqrt{\mathscr{L}}d}{\mathscr{C}_{\textsf{LSI}}}(5\mathscr{C}_{\textsf{LSI}}+3\mathscr{L})
\end{aligned}
\end{equation*}
Applying Gr{\"o}nwall's inequality, we obtain
\begin{equation}
\label{convg of the UNLD with gamma}
    \mathcal{F}^N(\mu_t^N)-N\mathcal{F}(\mu_*)\leq\mathcal{E}^N(\mu_t^N)-N\mathcal{E}(\mu_*)\leq\mathcal{E}_0^N\exp\left(-\frac{\mathscr{C}_{\textsf{LSI}}}{6\sqrt{\mathscr{L}}}t\right)+\frac{60\mathscr{L}d}{\mathscr{C}_{\textsf{LSI}}}+\frac{36\mathscr{L}^2d}{\mathscr{C}_{\textsf{LSI}}^2}
\end{equation}
where $\mathcal{E}_0^N\defeq\mathcal{E}^N(\mu_0^N)-N\mathcal{E}(\mu_*).$
This completes the proof of Theorem \ref{convergence of particle dynamics}. The convergence rate exhibited in \cref{convg of the UNLD with gamma} is faster and incurs a smaller bias than the rate exhibited in \cref{convg of the UNLD with gamma=1}.
\section{Discretization analysis}
\label{Proof of the Discrete-time-space Convergence}
In this section, we provide the proof of Theorem \ref{iter complexity of mean-field system} and Theorem \ref{iter complexity for particle system} establishing the global convergence of the discrete-time-space processes. Our discretization analysis is unified for the \hyperlink{MULA}{\textsf{MULA}} and \hyperlink{NULA}{\textsf{NULA}}.
\subsection{Proof of Theorem \ref{iter complexity of mean-field system}}
Suppose $\textbf{Q}_{Nh}$ is the joint law of the~\ref{UMFLD}
for $t\in[0,Nh]$
and $\textbf{P}_{Nh}$ is the joint law of the~\hyperlink{MULA}{\textsf{MULA}}
for $t\in[kh,(k+1)h]$ and $k=0,1,...,K-1$.
Applying Girsanov's theorem (Lemma \ref{Girsanov}), we have
\begin{equation*}
\begin{aligned}
    \textsf{KL}(\mathbf{Q}_{Kh}\|\mathbf{P}_{Kh})&=\mathbb{E}_{\mathbf{Q}_{Kh}}\log\frac{\mathrm{d}\mathbf{Q}_{Kh}}{\mathrm{d}\mathbf{P}_{Kh}}\\
    &=\mathbb{E}_{\mathbf{Q}_{Kh}}\sum_{k=0}^{K-1}\left(-\frac{1}{\sqrt{2\gamma}}\int_{kh}^{(k+1)h}\left\langle\left(\begin{array}{c}
         0\\
         {D}_{\rho}F(\mu_t^X,x_t)-{D}_{\rho}F(\mu_{kh}^X,x_{kh}) 
    \end{array}\right),\mathrm{dB}_t\right\rangle\right.\\
    &\left.\quad+\frac{1}{4\gamma}\int_{kh}^{(k+1)h}\left\|{D}_{\rho}F(\mu_t^X,x_t)-{D}_{\rho}F(\mu_{kh}^X,x_{kh})\right\|^2\mathrm{d}t\right)\\
    &=\frac{1}{4\gamma}\sum_{k=0}^{K-1}\int_{kh}^{(k+1)h}\mathbb{E}_{\mathbf{Q}_{Kh}}\left\|{D}_{\mu}F(\mu_t^X,x_t)-{D}_{\mu}F(\mu_{kh}^X,x_{kh})\right\|^2\mathrm{d}t
\end{aligned}
\end{equation*}
And we obtain
\begin{equation*}
    \begin{aligned}
        \textsf{KL}(\mathbf{Q}_{Kh}\|\mathbf{P}_{Kh})&=\frac{1}{4\gamma}\sum_{k=0}^{K-1}\int_{kh}^{(k+1)h}\mathbb{E}_{\mathbf{Q}_{Kh}}\left\|{D}_{\rho}F(\mu_t^X,x_t)-{D}_{\rho}F(\mu_{kh}^X,x_{kh})\right\|^2\mathrm{d}t\\
        &\leq\frac{{\mathscr{L}}^2}{2\gamma}\sum_{k=0}^{K-1}\int_{kh}^{(k+1)h}\mathbb{E}_{\mathbf{Q}_{Kh}}\|x_t-x_{kh}\|^2+W_1^2(\mu_t^X,\mu_{kh}^X)\mathrm{d}t\\
        &\leq\frac{\mathscr{L}^2}{2\gamma}\sum_{k=0}^{K-1}\int_{kh}^{(k+1)h}\mathbb{E}_{\mathbf{Q}_{Kh}}\|x_t-x_{kh}\|^2+\mathbb{E}_{\mathbf{Q}_{Kh}}\|x_t-x_{kh}\|^2\mathrm{d}t\\
        &=\frac{\mathscr{L}^2}{\gamma}\sum_{k=0}^{K-1}\int_{kh}^{(k+1)h}\mathbb{E}_{\mathbf{Q}_{Kh}}\|x_t-x_{kh}\|^2\mathrm{d}t
    \end{aligned}
\end{equation*}
where the first inequality follows from Assumption \ref{smoothness} and the last inequality follows from Lemma~\ref{bound of the iterate difference} and the inequality $\left(\frac{1}{n}\sum_{i=1}^nx_i\right)^2\leq \frac{1}{n}\sum_{i=1}^nx_i^2$:
\begin{equation*}
    \mathbb{E}_{\textbf{Q}_{Kh}}\|x_t-x_{kh}\|^2\leq 16\mathscr{L}^2h^4\mathbb{E}_{\textbf{Q}_{Kh}}\|x_{kh}\|^2+64h^2\mathbb{E}_{\textbf{Q}_{Kh}}\|v_{kh}\|^2+16\mathscr{L}^2h^4+32\gamma h^3d
\end{equation*}
Combined with Lemma~\ref{moment bound mean-field} and $\gamma=\sqrt{\mathscr{L}}$, the discretization error is upper bounded as follows:
\begin{equation*}
    \begin{aligned}
        \textsf{KL}(\mathbf{Q}_{Kh}\|\mathbf{p}_{Kh})&\leq\frac{16\mathscr{L}^4h^5K}{\gamma}\max_{0\leq k\leq K}\mathbb{E}_{\mathbf{Q}_{Kh}}\|x_{kh}\|^2+\frac{64\mathscr{L}^2h^3K}{\gamma}\max_{0\leq k\leq K}\mathbb{E}_{\mathbf{Q}_{Kh}}\|v_{kh}\|^2\\
        &\quad+\frac{16\mathscr{L}^4h^5K}{\gamma}+{32\mathscr{L}^2h^4Kd}\\
        &\lesssim\frac{\mathscr{L}^{9/2}h^5Kd}{ \mathscr{C}_{\textsf{LSI}}}+\frac{\mathscr{L}^{5/2}h^3Kd}{ \mathscr{C}_{\textsf{LSI}}}+{\mathscr{L}^{7/2}h^5K}+{\mathscr{L}^2h^4Kd}\\
        &=\frac{\mathscr{L}^{9/2}h^4Td}{ \mathscr{C}_{\textsf{LSI}}}+\frac{\mathscr{L}^{5/2}h^2Td}{ \mathscr{C}_{\textsf{LSI}}}+{\mathscr{L}^{7/2}h^4T}+{\mathscr{L}^2h^3Td}
    \end{aligned}
\end{equation*}
where $T=Kh$. By Lemma \ref{Mean-field Sandwich} and Theorem \ref{convergence rate of the UMFLD}, we obtain
\begin{equation}
\label{KL convg of the UMLD}
    \textsf{KL}(\mu_t\|\mu_*)\leq\mathcal{F}(\mu_t)-\mathcal{F}(\mu_*)\leq(\mathcal{E}(\mu_0)-\mathcal{E}(\mu_*))\exp\left(-\frac{\mathscr{C}_{\textsf{LSI}}}{3\sqrt{\mathscr{L}}}t\right)
\end{equation}
Combining with \eqref{KL convg of the UMLD}, we upper bound the TV distance between $\bar{\mu}_K$, the probability measure of \MULA{} at $Kh$ and $\mu_*$, the limiting distribution of \UMLD{} as follows:
\begin{equation*}
    \begin{aligned}
        \|\bar{\mu}_{K}-\mu_*\|_{\textsf{TV}}&\leq\|\bar{\mu}_{K}-\mu_{Kh}\|_{\textsf{TV}}+\|\mu_{Kh}-\mu_*\|_{\textsf{TV}}\\
        &=\|\mu_{Kh}-\bar{\mu}_{K}\|_{\textsf{TV}}+\|\mu_{Kh}-\mu_*\|_{\textsf{TV}}\\
        &\lesssim\sqrt{\textsf{KL}(\mu_{Kh}\|\Bar{\mu}_K)}+\sqrt{\textsf{KL}(\mu_{Kh}\|\mu_*)}\\
        &\lesssim\sqrt{\textsf{KL}(\mathbf{Q}_{Kh}\|\mathbf{p}_{Kh})}+\sqrt{\textsf{KL}(\mu_{Kh}\|\mu_*)}\\
        &\lesssim\frac{\mathscr{L}^{9/4}h^2T^{1/2}d^{1/2}}{ \mathscr{C}^{1/2}_{\textsf{LSI}}}+\frac{\mathscr{L}^{5/4}hT^{1/2}d^{1/2}}{ \mathscr{C}^{1/2}_{\textsf{LSI}}}+{\mathscr{L}^{7/4}h^2T^{1/2}}+{\mathscr{L}h^{3/2}T^{1/2}d^{1/2}}\\
        &\quad+(\mathcal{E}(\mu_0)-\mathcal{E}(\mu_*))^{1/2}\exp\left(-{\mathscr{C}_{\textsf{LSI}}T}/{6\sqrt{\mathscr{L}}}\right)
    \end{aligned}
\end{equation*}
where the first inequality follows from the triangle inequality of TV distance; the second inequality follows from Pinsker's inequality, and the fourth inequality follows from the data processing inequality. In order to ensure $\|\mu_{Kh}-\mu_*\|_{\textsf{TV}}\leq\frac{1}{2}\epsilon$, it suffices to choose $T=Kh=\widetilde{\Theta}\left(\frac{\sqrt{\mathscr{L}}}{\mathscr{C}_{\textsf{LSI}}}\right)$. In order to ensure $\|\bar{\mu}_{K}-\mu_{Kh}\|_{\textsf{TV}}\leq\frac{1}{2}\epsilon$, it suffices to choose the stepsize
\begin{equation}
    h={\Theta}\left(\frac{\mathscr{C}^{1/2}_{\textsf{LSI}}\epsilon}{\mathscr{L}^{5/4}T^{1/2}d^{1/2}}\right)=\widetilde{\Theta}\left(\frac{\mathscr{C}_{\textsf{LSI}}\epsilon}{\mathscr{L}^{3/2}d^{1/2}}\right),
\end{equation}
and the mixing time
\begin{equation}
    K=\frac{T}{h}=\widetilde{\Theta}\left(\frac{\mathscr{L}^2d^{1/2}}{\mathscr{C}_{\textsf{LSI}}^{2}\epsilon}\right).
\end{equation}
The choice of $T,\,h,\,K$ above ensures $\|\bar{\mu}_{K}-\mu_*\|_{\textsf{TV}}\leq\epsilon$.
\subsection{Proof of Theorem \ref{iter complexity for particle system}}
Suppose $Q_{Nh}^i$ is the joint law of the~\ref{Par-UMFLD} for the $i$-th particle
and $t\in[0,Kh]$; $P_{Nh}^i$ is the joint law of the~\hyperlink{NULA}{\textsf{NULA}} for the $i$-th particle.
Applying Girsanov's theorem (Lemma~\ref{Girsanov}), we have
\begin{equation*}
    \begin{aligned}
        \frac{1}{N}\sum_{i=1}^N\textsf{KL}(\mathbf{Q}^i_{Kh}\|\mathbf{P}^i_{Kh})&=\frac{1}{4\gamma}\sum_{k=0}^{K-1}\int_{kh}^{(k+1)h}\frac{1}{N}\sum_{i=1}^N\mathbb{E}_{\mathbf{Q}^i_{Kh}}\left\|{D}_{\rho}F({\mu_{\textbf{x}_t}},x^i_t)-{D}_{\rho}F(\mu_{\textbf{x}_{kh}},x^i_{kh})\right\|^2\mathrm{d}t\\
        &\leq\frac{\mathscr{L}^2}{2\gamma}\sum_{k=0}^{K-1}\int_{kh}^{(k+1)h}\frac{1}{N}\sum_{i=1}^N\mathbb{E}_{\mathbf{Q}^i_{Kh}}\|x^i_t-x^i_{kh}\|^2+W_1^2({\mu_{\textbf{x}_t}},\mu_{\textbf{x}_{kh}})\mathrm{d}t\\
        &\leq\frac{\mathscr{L}^2}{\gamma}\sum_{k=0}^{K-1}\int_{kh}^{(k+1)h}\frac{1}{N}\sum_{i=1}^N\mathbb{E}_{\mathbf{Q}^i_{Kh}}\|x^i_t-x^i_{kh}\|^2\mathrm{d}t\\
        &\leq\frac{16\mathscr{L}^4h^5}{\gamma}\frac{1}{N}\sum_{i=1}^N\sum_{k=1}^K\mathbb{E}_{\mathbf{Q}^i_{Kh}}\|x^i_{kh}\|^2+\frac{64\mathscr{L}^2h^3}{\gamma}\frac{1}{N}\sum_{i=1}^N\sum_{k=1}^K\mathbb{E}_{\mathbf{Q}^i_{Kh}}\|v_{kh}\|^2\\
        &\quad+\frac{16\mathscr{L}^4h^5K}{\gamma}+{32\mathscr{L}^2h^4Kd}
    \end{aligned}
\end{equation*}
where the first inequality follows from Assumption \ref{smoothness} and the last inequality follows from Lemma~\ref{bound of the iterate difference} and the inequality $\left(\frac{1}{n}\sum_{i=1}^nx_i\right)^2\leq \frac{1}{n}\sum_{i=1}^nx_i^2$:
\begin{equation*}
    \mathbb{E}_{\textbf{Q}^i_{Kh}}\|x^i_t-x^i_{kh}\|^2\leq 16\mathscr{L}^2h^4\mathbb{E}_{\textbf{Q}^i_{Kh}}\|x^i_{kh}\|^2+64h^2\mathbb{E}_{\textbf{Q}^i_{Kh}}\|v^i_{kh}\|^2+16\mathscr{L}^2h^4+32\gamma h^3d
\end{equation*}
for $t\in[kh,(k+1)h]$ and $k=0,1,...,K-1$.
Combining Lemma~\ref{moment bound par} and $\gamma=\sqrt{\mathscr{L}}$, the discretization error is upper bounded as follows:
\begin{equation*}
    \begin{aligned}
        \frac{1}{N}\sum_{i=1}^N\textsf{KL}(\mathbf{Q}^i_{Kh}\|\mathbf{P}^i_{Kh})
&\leq\frac{16\mathscr{L}^4h^5K}{\gamma}\frac{1}{N}\sum_{i=1}^N\max_{0\leq k\leq K}\mathbb{E}_{\mathbf{Q}^i_{Kh}}\|x^i_{kh}\|^2+\frac{64\mathscr{L}^2h^3K}{\gamma}\frac{1}{N}\sum_{i=1}^N\max_{0\leq k\leq K}\mathbb{E}_{\mathbf{Q}^i_{Kh}}\|v_{kh}\|^2\\
&\quad+\frac{16\mathscr{L}^4h^5K}{\gamma}+{32\mathscr{L}^2h^4Kd}\\
&\lesssim\frac{\mathscr{L}^{9/2}h^5Kd}{ \mathscr{C}_{\textsf{LSI}}}+\frac{\mathscr{L}^{5/2}h^3Kd}{ \mathscr{C}_{\textsf{LSI}}}+{\mathscr{L}^{7/2}h^5K}+{\mathscr{L}^{2}h^4Kd}\\
        &=\frac{\mathscr{L}^{9/2}h^4Td}{ \mathscr{C}_{\textsf{LSI}}}+\frac{\mathscr{L}^{5/2}h^2Td}{ \mathscr{C}_{\textsf{LSI}}}+{\mathscr{L}^{7/2}h^4T}+{\mathscr{L}^{2}h^3Td}
    \end{aligned}
\end{equation*}
where $T=Kh$. By Lemma~\ref{Particle System's Entropy Inequality} and Theorem \ref{convergence of particle dynamics}, we obtain
\begin{equation}
\label{par convg in KL}
    \frac{1}{N}\textsf{KL}(\mu_T^N\|\mu_*^{\otimes N})\leq\frac{1}{N}\mathcal{F}^N(\mu_T^N)-\mathcal{F}(\mu_*)\leq \frac{\mathcal{E}_0^N}{N}\exp\left(-\frac{\mathscr{C}_{\textsf{LSI}}}{6\sqrt{\mathscr{L}}}T\right)+\frac{60\mathscr{L}d}{N\mathscr{C}_{\textsf{LSI}}}+\frac{36\mathscr{L}^2d}{N\mathscr{C}_{\textsf{LSI}}^2},
\end{equation}
where $\mathcal{E}_0^N\defeq\mathcal{E}^N(\mu_0^N)-N\mathcal{E}(\mu_*).$
Combining with \eqref{par convg in KL}, we upper bound the averaged TV distance between $\bar{\mu}_K^i$ and $\mu_*$ over $N$ particles as follows:
\begin{equation*}
\begin{aligned}
    \frac{1}{N}\sum_{i=1}^N\|\bar{\mu}_{K}^i-{\mu}_*\|_{\textsf{TV}}&\leq\frac{1}{N}\sum_{i=1}^N\|\bar{\mu}_{K}^i-\mu_{Kh}^i\|_{\textsf{TV}}+\frac{1}{N}\sum_{i=1}^N\|\mu_{Kh}^i-\mu_*\|_{\textsf{TV}}\\
    &=\frac{1}{N}\sum_{i=1}^N\|\mu_{Kh}^i-\bar{\mu}_{K}^i\|_{\textsf{TV}}+\frac{1}{N}\sum_{i=1}^N\|\mu_{Kh}^i-\mu_*\|_{\textsf{TV}}\\
    &\lesssim\frac{1}{N}\sum_{i=1}^N\sqrt{\textsf{KL}({\mu}_{Kh}^i\|\bar{\mu}_{K}^i)}+\frac{1}{N}\sum_{i=1}^N\sqrt{\textsf{KL}(\mu_{Kh}^i\|\mu_*)}\\
    &\lesssim\sqrt{\frac{1}{N}\sum_{i=1}^N\textsf{KL}({\mu}_{Kh}^i\|\bar{\mu}_{K}^i)}+\sqrt{\frac{1}{N}\sum_{i=1}^N\textsf{KL}(\mu_{Kh}^i\|\mu_*)}\\
    &\leq\sqrt{\frac{1}{N}\sum_{i=1}^N\textsf{KL}(\mathbf{Q}_{Kh}^i\|\mathbf{P}_{Kh}^i)}+\sqrt{\frac{1}{N}\textsf{KL}(\mu_{Kh}^N\|\mu^{\otimes N}_*)}\\
    &\leq\sqrt{\frac{1}{N}\sum_{i=1}^N\textsf{KL}(\mathbf{Q}_{Kh}^i\|\mathbf{P}_{Kh}^i)}+\sqrt{\frac{1}{N}\mathcal{F}^N(\mu_{Kh}^N)-\mathcal{F}(\mu_*)}\\
    &\lesssim\frac{\mathscr{L}^{9/4}h^2T^{1/2}d^{1/2}}{ \mathscr{C}^{1/2}_{\textsf{LSI}}}+\frac{\mathscr{L}^{5/4}hT^{1/2}d^{1/2}}{ \mathscr{C}^{1/2}_{\textsf{LSI}}}+{\mathscr{L}^{7/4}h^2T^{1/2}}+{\mathscr{L}h^{3/2}T^{1/2}d^{1/2}}\\
    &\quad+\left(\frac{1}{N}\mathcal{E}^N(\mu_0^N)-\mathcal{E}(\mu_*)\right)^{1/2}\exp\left(-{\mathscr{C}_{\textsf{LSI}}T}/{12\sqrt{\mathscr{L}}}\right)+\frac{\mathscr{L}^{1/2}d^{1/2}}{N^{1/2}\mathscr{C}^{1/2}_{\textsf{LSI}}}+\frac{\mathscr{L}d^{1/2}}{N^{1/2}\mathscr{C}_{\textsf{LSI}}}
    \end{aligned}    
\end{equation*}
where the first inequality follows from the triangle inequality of TV distance; the second inequality follows from Pinsker's inequality; the third inequality follows from Jensen's inequality; the fourth inequality follows from data processing inequality and the information inequality (Lemma \ref{Information Inequality}) and the fifth inequality follows from Lemma \ref{Particle System's Entropy Inequality}. In order to ensure $\frac{1}{N}\sum_{i=1}^N\|\mu^i_{Kh}-\mu_*\|_{\textsf{TV}}\leq\frac{1}{2}\epsilon$, it suffices to choose $T=Kh=\widetilde{\Theta}\left(\frac{\sqrt{\mathscr{L}}}{\mathscr{C}_{\textsf{LSI}}}\right)$. In order to ensure $\frac{1}{N}\sum_{i=1}^N\|\bar{\mu}^i_{K}-\mu^i_{Kh}\|_{\textsf{TV}}\leq\frac{1}{2}\epsilon$, it suffices to choose the stepsize
\begin{equation}
    h={\Theta}\left(\frac{\mathscr{C}^{1/2}_{\textsf{LSI}}\epsilon}{\mathscr{L}^{5/4}T^{1/2}d^{1/2}}\right)=\widetilde{\Theta}\left(\frac{\mathscr{C}_{\textsf{LSI}}\epsilon}{\mathscr{L}^{3/2}d^{1/2}}\right),
\end{equation}
the mixing time
\begin{equation}
    K=\frac{T}{h}=\widetilde{\Theta}\left(\frac{\mathscr{L}^2d^{1/2}}{\mathscr{C}_{\textsf{LSI}}^{2}\epsilon}\right),
\end{equation}
and the number of particles
\begin{equation}
    N=\Theta\left(\frac{\mathscr{L}^{2}d}{\mathscr{C}_{\textsf{LSI}}^2\epsilon^2}\right).
\end{equation}
The choice of $T,\,h,\,K,\,N$ above ensures $\frac{1}{N}\sum_{i=1}^N\|\bar{\mu}^i_{K}-\mu_*\|_{\textsf{TV}}\leq\epsilon$.
\section{Experimental settings}
\label{Experimental Settings}
In our experiment, we use a mean-field two-layer neural network to approximate the Gaussian function,
\begin{equation*}
    f(z)=\exp\left(-\frac{\|z-m\|^2}{2d}\right).
\end{equation*}
We uniformly draw $m\sim\mathcal{N}(0,I_d)$ and $100$ points $\{z_i\}_{i=1}^{100}\sim\mathcal{N}(0,I_d)$ with $d=10^3$ and calculate the corresponding labels $\{f(z_i)\}_{i=1}^{100}$.
In this section, we give the actual updates of the methods involved in our experiment and provide the precise value of parameters in Table \ref{para_choice}. The update of the~\hyperlink{UNLA}{\textsf{NULA}} is given by
\begin{equation*}
    \begin{aligned}
    \qquad x^j_{k+1}&=x^j_k+\varphi_0\,v^j_{k}-\varphi_1\,
    {D}_{\mu}F(\mu_{\textbf{x}_{k}},x^j_{k})+\eta\xi^x_k,\\
    \qquad v^j_{k+1}&=\varphi_2\,v^j_{k}-\varphi_3\, {D}_{\mu}F(\mu_{\textbf{x}_{k}},x^j_{k})+\eta\xi^v_k.
    \end{aligned}
\end{equation*}
for $j=1,...,N$. The update of \ref{EM-UNLA} is given by
\begin{equation*}
     \begin{aligned}
        \qquad\ \quad\quad
        x_{k+1}^j&=x_k^j + h_2\,v_k^j,\\
        \qquad\ \quad\quad v_{k+1}^j&=(1-h_3)v_k^j - h_2\,{D}_{\mu}F(\mu_{\textbf{x}_k},x^j_k) + \sqrt{2\lambda_2h_2}\xi_k.
    \end{aligned}
\end{equation*}
for $j=1,...,N$. The update of the \ref{NLA} is given by
\begin{equation*}
    \begin{aligned}
        \ \ x^j_{k+1}=x^j_{k} - h_1\,{D}_{\mu}F(\mu_{\textbf{x}_k},x^j_k) + \sqrt{2\lambda_1h_1}\xi_k.
    \end{aligned}
\end{equation*}
for $j=1,...,N$.
\begin{table}[H]
    \centering
    \begin{tabular}{ccccccccccc}
    \toprule
        \textbf{Parameters} & $\varphi_0$ & $\varphi_1$ & $\varphi_2$ & $\varphi_3$ & $\eta$ & $h_1$ & $h_2$ & $h_3$ & $\lambda_1$ & $\lambda_2$ \\ \hline
         \textbf{Value}     &      $10^{-4}$ & $0.02$ &  $0.99$ &   $0.02$ &  $10^{-3}$ &  $10^{-2}$ & $10^{-2}$ & $10^{-2}$ & $10^{-4}$ & $10^{-4}$  \\ 
    \bottomrule
    \end{tabular}
    \caption{Choice of hyperparameters.}
    \label{para_choice}
\end{table}
\section{Methods for comparisons}
\label{Translation of the Compared Methods}
In this section, we review the convergence result of \MLA{} in \citet{nitanda2022convex} and \ref{NLA} in \citet{suzuki2023convergence}, which consider problem (1) in more specific settings. \citet{nitanda2022convex} suppose $F(\rho)=\mathbb{E}_{(a,b)\sim\mathcal{D}}\left[\ell(h(\rho;a),b)\right]+\frac{\lambda'}{2}\mathbb{E}_{x\sim \rho}\|x\|^2$ whereas \citet{suzuki2023convergence} suppose $F(\rho)=U(\rho)+\lambda'\mathbb{E}_{x\sim\rho}[r(x)]$. While our convergence results are established in TV distance, we consider more general settings compared with the previous two. Since the problem setting in \citet{nitanda2022convex} is only for training neural networks, we perform convergence analysis of the \MLA{} in \citet{suzuki2023convergence}'s setting to make a comparison with our results. Define the free energy
\begin{equation}
    \label{eq:free energy}
    E(\rho)=F(\rho)+\text{Ent}(\rho),
\end{equation}
where $\mu\in\mathcal{P}_2(\mathbb{R}^d)$. Let $\bar{\rho}_k$ denotes the law of $k$-th iterate of the \MLA{} and $\rho_*$ denotes the minimizer of~\eqref{eq:free energy}, and \citet{nitanda2022convex} obtain the following results in Theorem 2:
\begin{equation}
    E(\bar{\rho}_{k})-E(\rho_*)\leq\exp(-\mathscr{C}_{\textsf{LSI}}hk)(E(\bar{\rho}_0)-E(\rho_*))+\frac{\delta_{h}}{2\mathscr{C}_{\textsf{LSI}}},
\end{equation}
where $\delta_{hk}\defeq\mathbb{E}\|{D}_{\rho}F(\bar{\rho}_{k+1},x_{k+1})-{D}_{\rho}F(\bar{\rho}_{k},x_k)\|^2$ and $\mathbb{E}$ is taken under the joint law of $\bar{\rho}_{k+1}$ and $\bar{\rho}_k$. Now we bound $\delta_{hk}$ uniformly in $k$ with a different method from the one in \citet{nitanda2022convex}. We do not need to specify $F$ to be the objective of training nerual networks. Since $F$ is $\mathscr{L}$-smooth\footnote{We inherit the weaker smoothness assumption in \citet{suzuki2023convergence} with respect to $W_2$ distance.} and satisfies Assumption \ref{bounded grad}, we obtain
\begin{equation*}
    \begin{aligned}
        \mathbb{E}\|{D}_{\rho}F(\bar{\rho}_{k+1},x_{k+1})-{D}_{\rho}F(\bar{\rho}_{k},x_k)\|^2&\leq2\mathscr{L}^2\mathbb{E}(\|x_{k+1}-x_k\|^2+W_2^2(\bar{\rho}_{k+1},\bar{\rho}_k))\\
        &\leq 4\mathscr{L}^2\mathbb{E}\|x_{k+1}-x_k\|^2\\
        &=4\mathscr{L}^2\mathbb{E}\|-h{D}_{\rho}F(\bar{\rho}_{k},x_k)+\sqrt{2h}\xi\|^2\\
        &\leq4\mathscr{L}^2h^2\mathbb{E}\|{D}_{\rho}F(\bar{\rho}_{k},x_k)\|^2+8\mathscr{L}^2hd\\
        &\leq8\mathscr{L}^4h^2(1+\mathbb{E}\|x_{k}\|^2)+8\mathscr{L}^2hd
    \end{aligned}
\end{equation*}
We refer to Lemma 1 in \citet{suzuki2023convergence} to uniformly bound $\mathbb{E}\|x_{k}\|^2$. Before applying Lemma 1, we translate some constants in \citet{suzuki2023convergence} into our constants systems.
\citet{suzuki2023convergence} assumes that $\|{D}_{\rho}U(\rho,x)\|\leq R$, $\lambda_1I_d\preceq\nabla^2r(x)\preceq\lambda_2I_d$. We let $R=\mathscr{L}$ and $\lambda_2=\mathscr{L}$ (since this specification matches our Assumption \ref{bounded grad}). We prove Lemma 1 proposed by \citet{suzuki2023convergence} in the mean-field setting without particle approximation. But we also assume the decomposition $F(\rho)=U(\rho)+\mathbb{E}_{x\sim\rho}[r(x)]$ with $\|{D}_{\rho}U(\rho,x)\|\leq \mathscr{L}$ and $\lambda_1I_d\preceq\nabla^2r\preceq \mathscr{L}I_d$. Given the update of the~\MLA{}, if $h\leq\frac{\lambda_1}{2\mathscr{L}^2}$, we have
\begin{equation*}
    \begin{aligned}
    \mathbb{E}\|x_{k+1}\|^2&=\mathbb{E}\|x_k\|^2+h^2\mathbb{E}\|{D}_{\rho}F(\rho_k,x_k)\|^2+2hd-2h\mathbb{E}\left\langle x_k, {D}_{\rho}U(\rho_k,x_k)+\nabla r(x_k)\right\rangle\\
    &\leq\mathbb{E}\|x_k\|^2+\mathscr{L}^2h^2(1+\mathbb{E}\|x_k\|^2)+2hd+2h\mathscr{L}\mathbb{E}\|x_k\|-2h\lambda_1\mathbb{E}\|x_k\|^2\\
    &\leq(1-\lambda_1h)\mathbb{E}\|x_k\|^2+\mathscr{L}^2h^2+2hd+\frac{2\mathscr{L}^2h}{\lambda_1}
    \end{aligned}
\end{equation*}
Recursively, we obtain
\begin{equation}
    \label{eq:recursive of x_k}
    \mathbb{E}\|x_{k}\|^2\leq (1-\lambda_1h)^k\mathbb{E}\|x_{0}\|^2+\frac{\mathscr{L}^2h+2d}{\lambda_1}+\frac{2\mathscr{L}^2}{\lambda_1^2}\leq \mathbb{E}\|x_{0}\|^2+\frac{\mathscr{L}^2h+2d}{\lambda_1}+\frac{2\mathscr{L}^2}{\lambda_1^2}.
\end{equation}
If $x_0\sim\mathcal{N}(0,I_d)$, $\mathbb{E}\|x_{0}\|^2\lesssim d$. Thus \eqref{eq:recursive of x_k} implies $\mathbb{E}\|x_{k}\|^2\lesssim\mathscr{L}^2d$. Plugging into the inequality above, we obtain
\begin{equation}
    \mathbb{E}\|{D}_{\rho}F(\bar{\rho}_{k+1},x_{k+1})-{D}_{\rho}F(\bar{\rho}_{k},x_k)\|^2\lesssim\mathscr{L}^6h^2d+\mathscr{L}^2hd.
\end{equation}
Applying Lemma~\ref{Mean-field Sandwich} and pinsker's inequality, we obtain
\begin{equation*}
    \begin{aligned}
        \|\Bar{\rho}_K-\rho_*\|_{\textsf{TV}}\lesssim\sqrt{\textsf{KL}(\Bar{\rho}_K\|\rho_*)}&\leq\sqrt{E(\bar{\rho}_{K})-E(\rho_*)}\\
        &\lesssim\exp(-\mathscr{C}_{\textsf{LSI}}hK/2)(E(\bar{\rho}_0)-E(\rho_*))^{1/2}+\frac{\mathscr{L}^3hd^{1/2}}{\mathscr{C}_{\textsf{LSI}}^{1/2}}+\frac{\mathscr{L}h^{1/2}d^{1/2}}{\mathscr{C}_{\textsf{LSI}}^{1/2}}
    \end{aligned}
\end{equation*}
In order to ensure $\|\Bar{\rho}_K-\rho_*\|_{\textsf{TV}}\leq\epsilon$, it suffices to choose
\begin{equation}
    h=\Theta\left(\frac{\mathscr{C}_{\textsf{LSI}}\epsilon^2}{\mathscr{L}^3d}\right),\quad K=\widetilde{\Theta}\left(\frac{\mathscr{L}^3d}{\mathscr{C}_{\textsf{LSI}}^2\epsilon^2}\right).
\end{equation}
Now we translate the convergence results in \citet{suzuki2023convergence}. Define the free energy of the particle system:
\begin{equation}
E^N(\mu^N)=N\mathbb{E}_{\textbf{x}\sim\mu^N}F(\mu_{\textbf{x}})+\text{Ent}(\mu^N),
\end{equation}
where $\mu_{\textbf{x}}=\frac{1}{N}\sum_{i=1}^N\delta_{x^i}$.
Similar to the analysis above, Theorem 2 in \citet{suzuki2023convergence} implies the TV-convergence of the \ref{NLA}, given by
\begin{align*}
    \frac{1}{N}\sum_{i=1}^N\|\bar{\rho}_{K}^i-{\rho}_*\|_{\textsf{TV}}\lesssim\sqrt{\frac{1}{N}\sum_{i=1}^N\textsf{KL}(\bar{\rho}_{K}^i\|\rho_*)}&\leq\sqrt{\frac{1}{N}E^N(\rho_K^N)-E(\rho_*)}\\
    &\lesssim\exp\left(-\mathscr{C}_{\textsf{LSI}}hK/4\right)+h^{1/2}K^{1/2}(\mathscr{L}^3hd^{1/2}+\mathscr{L}h^{1/2}d^{1/2})\\
    &\quad+h^{1/2}K^{1/2}\frac{\mathscr{L}^2d^{1/2}}{N^{1/2}}
\end{align*}
In order to ensure $\|\Bar{\rho}_K-\rho_*\|_{\textsf{TV}}\leq\epsilon$, it suffices to choose
\begin{equation}
    h=\Theta\left(\frac{\mathscr{C}_{\textsf{LSI}}\epsilon^2}{\mathscr{L}^3d}\right),\quad K=\widetilde{\Theta}\left(\frac{\mathscr{L}^3d}{\mathscr{C}_{\textsf{LSI}}^2\epsilon^2}\right),\quad N=\Theta\left(\frac{\mathscr{L}^4d}{\mathscr{C}_{\textsf{LSI}}\epsilon^2}\right).
\end{equation}
\end{document}